\documentclass[a4paper,10pt]{amsart}

\usepackage[T1]{fontenc}
\usepackage{lmodern}
\usepackage{dsfont}
\usepackage{amsmath}
\usepackage{amsfonts}
\usepackage{amssymb}
\usepackage{amsthm}
\usepackage{mathtools}
\usepackage{thmtools}
\usepackage{xcolor}
\usepackage{tikz-cd}
\usepackage{enumitem}
\usepackage{xfrac}
\usepackage[pdftex,
            pdfauthor={},
            pdftitle={Equivalence Relations on Vertex Operator Algebras, I: Genus},
            pdfsubject={},
            pdfkeywords={},
            pdfproducer={},
            pdfcreator={}]{hyperref}
\usepackage{csquotes}
\theoremstyle{plain}
\usetikzlibrary{arrows, decorations.pathmorphing}

\newtheorem{thm}{Theorem}[section]

\newtheorem{prop}[thm]{Proposition}
\newtheorem{propph}[thm]{Proposition\textsuperscript{ph}}
\newtheorem{cor}[thm]{Corollary}
\newtheorem{lem}[thm]{Lemma}

\newtheorem{conj}[thm]{Conjecture}

\newtheorem*{thm*}{Theorem}
\newtheorem*{conj*}{Conjecture}
\newtheorem*{claim*}{Claim}
\newtheorem*{prop*}{Proposition}

\theoremstyle{definition}
\newtheorem{defi}[thm]{Definition}
\newtheorem{defiph}[thm]{Definition\textsuperscript{ph}}

\newtheorem*{nota*}{Notation}
\newtheorem{rem}[thm]{Remark}
\newtheorem{remph}[thm]{Remark\textsuperscript{ph}}

\newtheorem{ex}[thm]{Example}

\newtheorem{prob}[thm]{Problem}

\newcommand{\Q}{\mathbb{Q}}
\newcommand{\Z}{\mathbb{Z}}
\newcommand{\Ns}{\mathbb{Z}_{>0}}
\newcommand{\N}{\mathbb{Z}_{\geq0}}
\newcommand{\C}{\mathbb{C}}
\newcommand{\R}{\mathbb{R}}

\newcommand{\tr}{\operatorname{tr}}
\renewcommand{\i}{\mathrm{i}}
\newcommand{\End}{\operatorname{End}}
\newcommand{\Aut}{\operatorname{Aut}}
\newcommand{\Hom}{\operatorname{Hom}}
\newcommand{\Ind}{\operatorname{Ind}}
\newcommand{\rk}{\operatorname{rk}}

\newcommand{\sign}{\operatorname{sign}}
\newcommand{\voa}{vertex operator algebra}
\newcommand{\Voa}{Vertex operator algebra}
\newcommand{\VOA}{Vertex Operator Algebra}
\newcommand{\vosa}{vertex operator subalgebra}

\newcommand{\svoa}{vertex operator superalgebra}

\newcommand{\fpvosa}{fixed-point vertex operator subalgebra}

\newcommand{\vac}{\textbf{1}}

\newcommand{\id}{\operatorname{id}}

\newcommand{\lcm}{\operatorname{lcm}}

\newcommand{\ee}{\mathfrak{e}}
\newcommand{\g}{\mathfrak{g}}
\newcommand{\hh}{\mathfrak{h}}

\newcommand{\Inn}{\operatorname{Inn}}

\renewcommand{\sl}{\mathfrak{sl}}

\newcommand{\Irr}{\operatorname{Irr}}

\newcommand{\strat}{strongly rational}
\newcommand{\Strat}{Strongly rational}
\newcommand{\II}{I\!I}
\renewcommand{\O}{\operatorname{O}}
\newcommand{\Co}{\operatorname{Co}}
\newcommand{\spn}{\operatorname{span}}
\newcommand{\Com}{\operatorname{Com}}
\newcommand{\no}{\,{\raise0.25em\hbox{$\mathop{\hphantom{\cdot}}\limits^{_{\circ}}_{^{\circ}}$}}\,}

\newcommand{\Rep}{\operatorname{Rep}}

\newcommand{\Vect}{\operatorname{Vect}}
\newcommand{\Cat}{\mathcal{C}}
\newcommand{\gen}{\operatorname{gen}}
\newcommand{\hgen}{\operatorname{hgen}}
\newcommand{\bgen}{\operatorname{bgen}}

\newcommand{\mass}{\operatorname{mass}}

\usepackage{array}
\newcolumntype{H}{>{\setbox0=\hbox\bgroup}c<{\egroup}@{}}

\newcommand{\sAA}{1^{24}}
\newcommand{\sBB}{1^82^8}
\newcommand{\sCC}{1^63^6}
\newcommand{\sDD}{2^{12}}
\newcommand{\sEE}{1^42^24^4}
\newcommand{\sFF}{1^45^4}
\newcommand{\sGG}{1^22^23^26^2}
\newcommand{\sHH}{1^37^3}
\newcommand{\sII}{1^22^14^18^2}
\newcommand{\sJJ}{2^36^3}
\newcommand{\sKK}{2^210^2}

\newcommand{\gAA}{\II_{24,0}(1)}
\newcommand{\gBB}{\II_{16,0}(2_{\II}^{+10})}
\newcommand{\gCC}{\II_{12,0}(3^{-8})}
\newcommand{\gDD}{\II_{12,0}(2_{\II}^{-10}4_{\II}^{-2})}
\newcommand{\gEE}{\II_{10,0}(2_{2}^{+2}4_{\II}^{+6})}
\newcommand{\gFF}{\II_{8,0}(5^{+6})}
\newcommand{\gGG}{\II_{8,0}(2_{\II}^{+6}3^{-6})}
\newcommand{\gHH}{\II_{6,0}(7^{-5})}
\newcommand{\gII}{\II_{6,0}(2_{5}^{-1}4_{1}^{+1}8_{\II}^{-4})}
\newcommand{\gJJ}{\II_{6,0}(2_{\II}^{+4}4_{\II}^{-2}3^{+5})}
\newcommand{\gKK}{\II_{4,0}(2_{\II}^{-2}4_{\II}^{-2}5^{+4})}
\newcommand{\gLL}{\II_{0,0}(1)}

\begin{document}

\title[Equivalence Relations on Vertex Operator Algebras, I]{Equivalence Relations on Vertex Operator Algebras, I: Genus}
\author[Sven Möller and Brandon~C. Rayhaun]{Sven Möller\textsuperscript{\lowercase{a}} and Brandon~C. Rayhaun\textsuperscript{\lowercase{b}}}

\begin{abstract}
In this first of a series of two papers, we investigate two different equivalence relations obtained by generalizing the notion of genus of even lattices to the setting of \voa{}s (or two-dimensional chiral algebras). The bulk genus equivalence relation was defined in \cite{Hoe03} and groups (suitably regular) \voa{}s according to their modular tensor category and central charge. Hyperbolic genus \cite{Mor21} tests isomorphy after tensoring with a hyperbolic plane vertex algebra. Physically, two rational chiral algebras are said to belong to the same bulk genus if they live on the boundary of the same 2+1d topological quantum field theory; they belong to the same hyperbolic genus if they can be related by current-current exactly marginal deformations after tensoring a non-chiral compact boson.

As one main result, we prove the conjecture that the  hyperbolic genus defines a finer equivalence relation than the bulk genus. This is based on a new, equivalent characterization of the hyperbolic genus that uses the maximal lattice inside a \voa{} and its commutant (or coset).

We discuss the implications of these constructions for the classification of rational conformal field theory. In particular, we propose a program for (partially) classifying $c=32$, holomorphic \voa{}s (or chiral conformal field theories), and obtain novel lower bounds, via a generalization of the Smith--Minkowski--Siegel mass formula, on the number of \voa{}s at higher central charges. Finally, we conjecture a Siegel--Weil identity which computes the ``average'' torus partition function of an ensemble of chiral conformal field theories defined by any hyperbolic genus, and interpret this formula physically in terms of disorder-averaged holography.
\end{abstract}

\thanks{\textsuperscript{a}{Universität Hamburg, Hamburg, Germany}}
\thanks{\textsuperscript{b}{Yang Institute for Theoretical Physics, State University of New York, Stony Brook, NY, United States of America}}

\thanks{Email: \href{mailto:math@moeller-sven.de}{\nolinkurl{math@moeller-sven.de}}, \href{mailto:brandonrayhaun@gmail.com}{\nolinkurl{brandonrayhaun@gmail.com}}}

\maketitle

\setcounter{tocdepth}{1}
\tableofcontents
\setcounter{tocdepth}{2}


\pagebreak
\section{Introduction}


\subsection{For Mathematicians} (An introduction for physicists follows below.) \Voa{}s and their representation categories axiomatize the notions of chiral algebras and fusion rings of $2$-dimensional conformal field theories in physics. Originally introduced to solve the monstrous moonshine conjecture \cite{Bor86,FLM88,Bor92}, \voa{}s have found numerous applications in mathematics, e.g., in algebraic geometry, group theory, Lie theory and the theory of automorphic forms. The various applications to mathematical physics will be addressed in \autoref{sec:physintro}.

The theory of (suitably regular) \voa{}s is in some ways similar to the theory of even lattices over the integers. This has been leveraged to achieve classification results in small central charges. For instance, it is shown in \cite{Hoe17,ELMS21,MS23,HM23} that holomorphic vertex operator (super) algebras are almost all completely governed by the Leech lattice $\Lambda$ in central charges up to~$24$. Recall that the Leech lattice provides the solution to the densest sphere packing problem in dimension~$24$ \cite{CKMRV17}.

Beyond central charge~$24$, a full classification of \voa{}s \emph{up to isomorphism} becomes hopelessly unfeasible, as is evidenced by the sheer number of lattice \voa{}s alone \cite{Kin03}. However, one is led to ask if there are coarser (and hence perhaps more manageable) equivalence relations on \voa{}s that still capture important aspects of the theory. We shall address this question in a series of two papers, this text furnishing the first part. The second part will appear in \cite{MR24b}.

\medskip

In this text, we shall primarily be concerned with generalizations of the notion of genus for even lattices to \voa{}s, building on earlier work in \cite{Hoe03,Hoe17,Mor21}. Two even lattices $L$ and $M$ are in the same genus if one of the following equivalent conditions is satisfied (see \autoref{defi:latticegenus}):
\begin{enumerate}[label=(\alph*)]
\setcounter{enumi}{1}
\item\label{item:latgen2intro} $L'/L\cong M'/M$ and $\sign(L)=\sign(M)$.
\item\label{item:latgen3intro} $L\oplus\II_{1,1}\cong M\oplus\II_{1,1}$.
\end{enumerate}
Here, $L'/L$ denotes the discriminant form of $L$ and $\II_{1,1}$ is the unique even, unimodular lattice of signature $(1,1)$, sometimes called the hyperbolic plane. In a sense, the discriminant form \cite{Nik80,CS99} captures the local information on a lattice, and the (finite) number of isomorphism classes of even lattices in a fixed genus measures the failure of the local-global principle for quadratic forms over the integers.

Both definitions of the lattice genus admit natural generalizations to \voa{}s, the former being considered in \cite{Hoe03} and the latter in \cite{Mor21} (see also \cite{Hoe17}). On the one hand, two (\strat{}) \voa{}s $V$ and $W$ are in the same \emph{bulk genus} (see \cite{Hoe03} and \autoref{defi:bulkgenus}) if
\begin{equation*}
\Rep(V)\cong\Rep(W)\quad\text{and}\quad c(V)=c(W). 
\end{equation*}
Here, $\Rep(V)$ denotes the representation category of $V$, which is a modular tensor category \cite{Hua08b}, and $c(V)\in\Q$ denotes the central charge of $V$.

Every positive-definite, even lattice $L$ uniquely defines a lattice \voa{} $V_L$. By a generalization, we mean a notion of \voa{} genus that recovers the lattice genus for lattice \voa{}s. Indeed, the representation category of a lattice \voa{} $V_L$ is the pointed (and pseudo-unitary) modular tensor category $\mathcal{C}(L'/L)$ characterized by the discriminant form $L'/L$ of $L$, and the central charge $c(V_L)$ equals the rank $\rk(L)=\sign(L)$ of $L$.

On the other hand, two \voa{}s are in the same \emph{hyperbolic genus} (see \cite{Mor21} and \autoref{defi:hypgen}) if
\begin{equation*}
V\otimes V_{\II_{1,1}}\cong W\otimes V_{\II_{1,1}},
\end{equation*}
respecting choices of Cartan subalgebras. Here $V_{\II_{1,1}}$ denotes the lattice vertex algebra associated with the hyperbolic plane. Once again, it is not difficult to see that this defines a generalization of the lattice genus in the above sense.

One of the main results of this text is:

\smallskip

\noindent\textbf{\autoref{cor:hypbulk}.} \emph{If $V$ and $W$ are \strat{} \voa{}s belonging to the same hyperbolic genus, then they belong to the same bulk genus.}

\smallskip

That is, we show that the hyperbolic genus is a refinement of the bulk genus. This was conjectured in \cite{Mor21}. The converse of \autoref{cor:hypbulk} is false.

\medskip

In order to prove \autoref{cor:hypbulk}, we develop an alternative characterization of the hyperbolic genus, which is inspired by the results of \cite{Hoe17}. Rather than the ``external'' condition in \autoref{defi:hypgen}, we provide an ``internal'' definition.

In any (\strat{}) \voa{} $V$ there is a maximal lattice \voa{} $V_L\subset V$ \cite{Mas14}, which together with the Heisenberg commutant $C\coloneqq\Com_{V}(V_L)$ forms a dual pair in $V$. It follows that $V$ is a simple-current extension
\begin{equation*}
V\cong \bigoplus_{\alpha+L\in A} C^{\tau(\alpha+L)}\otimes V_{\alpha+L}
\end{equation*}
with gluing map $\tau$, also called mirror extension (see \autoref{prop:asslatdecomp} for details).

Now, two \voa{}s $V$ and $W$ are in the same \emph{hyperbolic genus} according to the alternative definition (see \autoref{defi:commgenus}) if
\begin{enumerate}
\item they have isomorphic Heisenberg commutants,
\item their associated lattices are in the same lattice genus,
\item their gluing maps are compatible in a certain sense.
\end{enumerate}

As another main result of this paper, we prove:

\smallskip

\noindent\textbf{\autoref{thm:hypcomm}.} \emph{Two \strat{} \voa{}s $V$ and $W$ belong to the same hyperbolic genus as in \autoref{defi:commgenus} if and only if they are in the same hyperbolic genus according to \autoref{defi:hypgen}.}

\smallskip

We point out that \autoref{defi:commgenus} simplifies considerably if $V$ (and hence $W$) is holomorphic, i.e.\ if $\Rep(V)\cong\Vect$. In that case, $V$ and $W$ are in the same hyperbolic genus if and only if they have isomorphic Heisenberg commutants and $c(V)=c(W)$ (see \autoref{cor:holhyp}). In fact, the same statement holds more generally when $\Rep(V)$ is unpointed (see \autoref{cor:unpointed}).

\medskip

As a main example, we consider in \autoref{sec:holo} holomorphic \voa{}s, i.e.\ \voa{}s with the trivial module category $\Rep(V)\cong\Vect$. In that case, the central charge must be a non-negative multiple of $8$.

The most interesting case is perhaps that of central charge~$24$, i.e.\ the bulk genus with $\Rep(V)\cong\Vect$ and $c(V)=24$. Modulo the moonshine uniqueness conjecture, \autoref{conj:moonshineuniqueness} \cite{FLM88}, these \voa{}s have been classified in \cite{Hoe17,MS23,MS21,HM22,LM22}, rigorously proving \cite{Sch93}. In total, there are $71$ such \voa{}s in this bulk genus, called the Schellekens \voa{}s. Their Heisenberg commutants were determined in \cite{Hoe17}; hence \autoref{cor:hypbulk} immediately provides the decomposition of this bulk genus into $12$ hyperbolic genera (see \autoref{table:12}).

We also consider the bulk genus of \strat{}, holomorphic \voa{}s of central charge~$32$. By \cite{Kin03} it contains over a billion lattice \voa{}s alone, but it may still be feasible to classify these \voa{}s up to hyperbolic equivalence, i.e.\ up to isomorphism of Heisenberg commutants (see \autoref{prop:hypgenhol}). In \autoref{sec:holhyp} and \autoref{subsec:c=32} we make some inroads into this problem.

\medskip

As further application, in \autoref{subsec:holomass}, we derive asymptotic lower bounds on the number of \voa{}s in infinite families of hyperbolic genera, in the special case of holomorphic \voa{}s. This is based on a mass formula (in the sense of Smith--Minkowski--Siegel) for hyperbolic genera proved in \cite{Mor21} (see \autoref{thm:massformula}) and a rewriting of the appearing subgroup index in terms of the Heisenberg commutant, i.e.\ in terms of our alternative characterization of the hyperbolic genus (see \autoref{cor:massconstant}).

Finally, in \autoref{subsec:SiegelWeil}, as a generalization of the mass formula for the \voa{}s in a hyperbolic genus, we conjecture a Siegel--Weil identity (see \autoref{conj:generalizedSiegelWeil}). We prove this conjecture for one hyperbolic genus of holomorphic \voa{}s of central charge~$24$, leaving the general case for future work. The conjecture implies an expression for the ``average'' character of a \voa{} in a hyperbolic genus.

\medskip

In view of the second part of our treatise \cite{MR24b}, we emphasize the general picture that emerges. Each equivalence relation on even lattices (like the genus) splits up into two notions of equivalence relations for (\strat{}) \voa{}s, one that (like the hyperbolic genus) is a more ``classical'' analog, while the other (like the bulk genus) is a more honest ``quantum'' analog. The former will always be a refinement of the latter. We depict this in \autoref{fig:summary}.

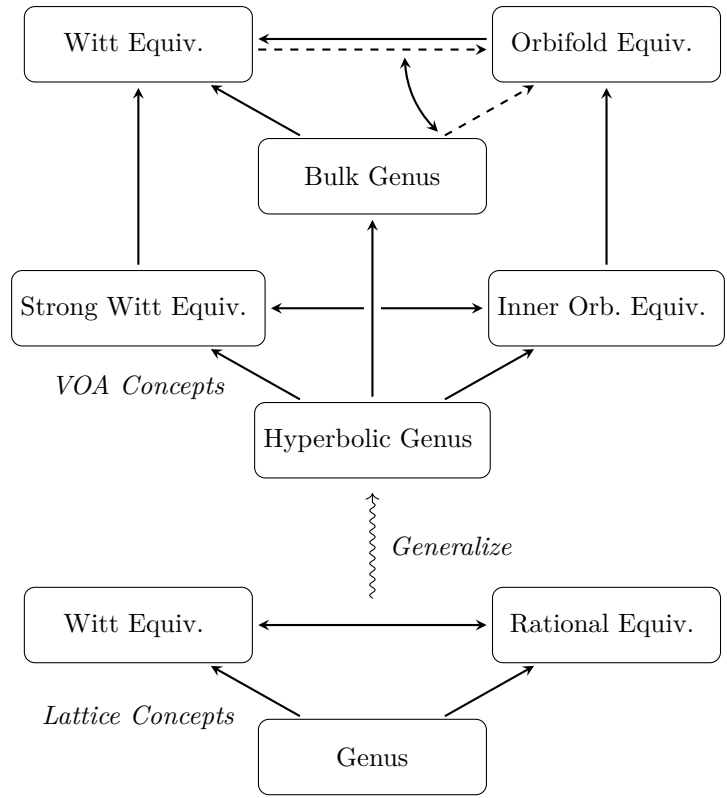
\begin{figure}[ht]
\begin{tikzpicture}[decoration={snake,post length=3pt,amplitude=1pt,segment length=4pt},scale=.7]
\def \tt {-2};
\tikzstyle{equiv}=[draw,rectangle, rounded corners, minimum width=3cm, minimum height=1cm]
\tikzstyle{a}=[thick,->,>=stealth,shorten >=2pt,shorten <=2pt]
\tikzstyle{t}=[thick,<->,>=stealth,shorten >=2pt,shorten <=2pt]
\tikzstyle{b}=[thick,<->,>=stealth,shorten >=2pt,shorten <=2pt]
\tikzstyle{d}=[thick,dashed,->,>=stealth,shorten >=2pt,shorten <=2pt]
\tikzstyle{dd}=[thick,double equal sign distance, -Implies]
\tikzstyle{s}=[draw,decorate,->]
\node[equiv,align=center] (witt) at (0,0) {Witt Equiv.\ };
\node[equiv,align=center] (orbifold) at (8.8,0) {Orbifold Equiv.\ };
\node[equiv,align=center] (rational) at (0,-5) {Strong Witt Equiv.\ };
\node[equiv,align=center] (inner) at (8.8,-5) {Inner Orb.\ Equiv.\ };
\node[equiv,align=center] (bulk) at (4.4,-2.5) {Bulk Genus};
\node[equiv,align=center] (hyperbolic) at (4.4,-7.5) {Hyperbolic Genus };
\node[] at (0,-6.5) {\emph{VOA Concepts}};
\draw[b] (5.7,-1.73) to [bend left=20] (5,-.15);
\draw[d] ([shift={(0,-.1)}]witt.east) to node[anchor=north,shift={(-1,0)}] {} ([shift={(0,-.1)}]orbifold.west);
\draw[a] ([shift={(0,.1)}]orbifold.west) to node[anchor=south] {}  ([shift={(0,.1)}]witt.east);
\draw[d] (bulk) -- node[anchor=west,shift={(-.5,-.4)}] {} (orbifold);
\draw[a] (hyperbolic) -- node[anchor=west,shift={(-.5,-.4)}] {} (inner);
\draw[a] (rational) -- node[anchor=west] {} (witt);
\draw[a] (inner) -- node[anchor=east] {} (orbifold);
\draw[a] (hyperbolic) -- node[anchor=west,shift={(-.3,.35)}] {} (rational);
\draw[a] (bulk) -- node[anchor=east,shift={(.32,-.35)}] {} (witt);
\draw[t] (rational) -- node[anchor=south,shift={(1.4,0)}] {} (inner);
\draw[draw=white,fill=white] (4.25,-5.15) rectangle ++(0.3,0.3);
\draw[a] (hyperbolic) -- node[anchor=east,shift={(0,.85)}] {} (bulk);
\path[s] (4.4,-10.5) to node[anchor=west,shift={(0.1,0)}] {\emph{Generalize}} (4.4,-8.5);
\node[] at (0,-10.75+\tt) {\emph{Lattice Concepts}};
\node[equiv,align=center] (lattice_witt) at (0,-9+\tt) {Witt Equiv.\ };
\node[equiv,align=center] (lattice_rat) at (8.8,-9+\tt) {Rational\ Equiv.\ };
\node[equiv,align=center] (lattice_gen) at (4.4,-11.5+\tt) {Genus};
\draw[a] (lattice_gen) -- (lattice_witt);
\draw[a] (lattice_gen) -- (lattice_rat);
\draw[t] (lattice_witt) -- (lattice_rat);
\end{tikzpicture}
\caption{Interrelations between various equivalence relations on \strat{} vertex operator algebras studied in our work. The solid arrows are rigorously established. The dashed arrows are conjectural. We restrict ourselves in this paper to the notions of genus, hyperbolic genus and bulk genus, and address the remaining concepts in the second part \cite{MR24b}.}
\label{fig:summary}
\end{figure}


\subsection{For Physicists}\label{sec:physintro}

What are the physically and mathematically useful ways of organizing the space of quantum field theories? One motivation for asking this question is the old and venerable problem of classifying rational conformal field theories. This problem is far too difficult to tackle all at once, and in fact difficult enough that even identifying what the ``right'' tractable subproblems are is worthy of study in its own right. For example, one direction is to define natural equivalence relations on quantum field theories which allow one to systematically explore theory space in manageable chunks. Indeed, with a judiciously chosen equivalence relation (or, equivalence relations) in hand, one could attempt to 1) classify theories within an equivalence class, one class at a time, 2) classify theories \emph{modulo} the equivalence relation, or 3) some combination of both.

Our approach  is to take inspiration from the well-developed theory of integral lattices. In that setting, there are several mathematically interesting ways to organize the space of lattices into equivalence classes. Furthermore, there is a useful analogy between lattices and chiral algebras, which goes through the fact that any integral, positive-definite lattice $L$ defines a rational chiral algebra $V_L$ describing $\rk(L)$-many chiral free bosons $\varphi^i$. (In this construction, $L$ determines which vertex operators $e^{i\lambda\cdot \varphi}$ are included in the algebra.) It is then natural to attempt to translate various lattice-theoretic ideas and constructions across this analogy, and ask whether they can be generalized beyond chiral free bosons, to arbitrary chiral algebras. 

The lattice-theoretic concept we focus on in this work is the notion of a \emph{genus} of lattices, which we describe below and review in more detail in \autoref{sec:lat}. In a companion paper \cite{MR24b}, we study Witt equivalence, rational equivalence, and neighborhood, and investigate their relationship to orbifolding in two-dimensional quantum field theory. We offer \autoref{fig:summary} as a glossary of the interrelations between the various concepts arising in this paper and in \cite{MR24b}. The general picture we uncover is that each lattice equivalence relation splits into two equivalence relations on chiral algebras: one which is more ``classical'' in the sense that it is closer to the lattice definition, and one which is more honestly ``quantum''. The former always defines a finer equivalence relation than the latter. 

\medskip


We start by reminding the reader what it means for two  lattices to belong to the same genus. For simplicity, we restrict our attention here to lattices which are even and positive-definite. There are at least two equivalent definitions we could take as our starting point. As we will see, these definitions, when extrapolated, lead to two genuinely different equivalence relations on chiral algebras.

To state the first definition, recall that, to any even, positive-definite  lattice $L$, one can associate a metric group $D$, i.e.\ a finite abelian group equipped with a function $q\colon D\to\C^\times$ such that $b(x,y) \coloneqq q(x+y)q(x)^{-1}q(y)^{-1}$ is a non-degenerate bilinear form and $q(x)=q(-x)$. This group is called the discriminant form of $L$, and is defined as $D=L'/L$, where $L'$ is the dual lattice.   One then declares that two lattices belong to the same genus if they have the same rank and their discriminant forms are isometric. 

We would like to interpret this equivalence relation physically, in terms of free chiral boson operator algebras $V_L$. Rational chiral algebras are relative quantum field theories, and hence they should be thought of as living on the boundary of a topological field theory in one dimension higher \cite{Wit89,EMSS89} (see \autoref{fig:3d2d}). On the other hand, it is known that a metric group (and the choice of an integer, which we can take to be the rank of the lattice, $r\coloneqq\rk(L)$) is precisely the data needed to define a 2+1d abelian topological field theory.\footnote{A 2+1d topological field theory is said to be abelian if the fusion rules of its anyons are described by a finite abelian group.} In fact, the abelian topological field theory defined by the pair $(D,r)=(L'/L,\rk(L))$ \emph{is} the bulk theory which supports $V_L$ on its boundary. Thus, we learn that two lattices $L$ and $K$ belong to the same genus if and only if their corresponding free chiral boson theories $V_L$ and $V_K$ arise on the boundary of the same bulk topological field theory. These words can be naturally extended to apply to arbitrary chiral algebras.

\smallskip

\noindent\textbf{Definition}.  \emph{Two rational chiral algebras belong to the same bulk genus if they arise on the boundary of the same 2+1d bulk topological field theory.}

\smallskip

Mathematically, it is equivalent to say that two rational chiral algebras belong to the same bulk genus if their central charges agree, and if their representation categories are equivalent as modular tensor categories. This definition was first explicitly considered in \cite{Hoe03}. We affix the word ``bulk'' onto genus here to differentiate this notion from another, inequivalent definition of genus which we introduce shortly. 

The most immediate virtue of partitioning the space of rational chiral algebras into bulk genera is that there are conjecturally only finitely many theories in a fixed bulk genus (cf.\ Conjecture 3.5 of \cite{Hoe03} and \autoref{conj:bulkfiniteness} below).  In particular, it becomes a conceivably tractable problem to attempt to classify all of the rational chiral algebras within a fixed bulk genus. This problem is quite natural when it is rephrased in the language of condensed matter physics: in that context, one is fixing a 2+1d topological phase of matter, and asking for the complete list of its gapless chiral edge modes. As a partial result in the direction of \autoref{conj:bulkfiniteness}, we establish in \autoref{prop:partialfiniteness} that, under the assumption that there are only finitely many chiral conformal field theories of a fixed central charge, the orientation reversal of any Chern--Simons theory with semisimple gauge algebra supports only finitely many gapless chiral boundary conditions.

Now, 2+1d topological theories, which can be thought of as labeling different bulk genera, can themselves be organized in terms of the number of anyons $p$ they possess and their chiral central charge $c$. Here too, one encounters only finitely many topological quantum field theories for each tuple $(p,c)$ \cite{BNRW16}, and much is  known about the classification of topological field theories which have a small numbers of anyons \cite{RSW09,NRWW23,NRW23}. Thus, one could imagine systematically walking through the space of rational chiral algebras, one bulk genus at a time, in order of increasing central charge $c$ and number of anyons $p$. In the language of conformal field theory, the number of anyons $p$ of a topological field theory is interpreted as the number of irreducible representations (or primaries) of the chiral algebra(s) that the topological field theory supports on its boundary. Thus, loosely speaking, this way of thinking organizes the space of chiral algebras into chunks of finite size based on complexity: starting with theories with fewer degrees of freedom (lower central charge $c$) and simpler operator content (smaller~$p$). 

There are several bulk genera for which the classification problem has been solved. For example, consider the rational chiral algebras whose only primary operator is the identity, so that $p=1$; we refer to such theories as \emph{chiral conformal field theories}, or alternatively as \emph{holomorphic vertex operator algebras}. The absence of any non-trivial primaries implies that a chiral conformal field theory lives at the boundary of a topological quantum field theory without any non-trivial anyons; the latter is often called an invertible topological field theory. In 2+1d, the invertible topological field theories are simply powers of the $E_8$ phase, and hence are labeled by an integer $n\in\Z$ which specifies the chiral central charge $c=8n$. Thus, the collection of chiral conformal field theories with a fixed central charge $c=8n$ defines a bulk genus which is labeled by the unique invertible 2+1d topological field theory of chiral central charge $c$. In his seminal work, Schellekens \cite{Sch93} predicted that there are $74$ non-trivial chiral conformal field theories with central charge $c\leq 24$, a result which was  turned into rigorous mathematics in a series of papers \cite{Hoe17,MS23,MS21,HM22,LM22}, modulo a still open conjecture about the uniqueness of the monster conformal field theory (cf.\ \autoref{conj:moonshineuniqueness}). 

More recently, analogous results were obtained for rational chiral algebras in most of the bulk genera labeled by topological field theories with $p\leq 4$ anyons and central charge $c\leq 24$ \cite{SR23,Ray23,HM23}, though the methods used apply in much greater levels of generality.\footnote{See also, e.g., \cite{MMS88,GHM16,KLP21,TW17,CM19,GLTY20} for prior related work.} Very roughly, the approach taken in \cite{SR23,Ray23} is to argue (under mild assumptions) that any rational chiral algebra with $p>1$ occurs as a subalgebra of a chiral conformal field theory (i.e.\ a rational chiral algebra with $p=1$), and can be recovered as a coset. These works operated under the assumption of unitarity,  or more modestly, that the conformal dimensions of all non-trivial primaries are strictly positive. We strengthen the results of op.\ cit.\ by demonstrating, with the help of our \autoref{prop:essentiallypositive}, that they continue to hold even once one drops the assumption of unitary. That is, the unitary topological field theories with at most 4 anyons do not admit any unitarity-violating gapless chiral boundary conditions. (See \autoref{subsec:glue} and \autoref{subsec:bulkex} for a more detailed discussion.) 

Thus, in some sense, the entire classification of rational chiral algebras with non-trivial representation theory is controlled just by the  chiral conformal field theories, which have trivial representation categories. It therefore behooves us to understand the latter class as best we can.

\medskip

One bottleneck in classifying chiral conformal field theories beyond $c=24$ is that the number of theories very quickly explodes. For example, King has used a refinement \cite{Kin03} of the Smith--Minkowski--Siegel mass formula to lower bound the number of even, unimodular lattices $L$ with $\rk(L)=32$ (each of which corresponds to a $c=32$ chiral conformal field theory built out of free bosons), and found that there are more than a billion. Though the rapid proliferation of lattices is a practical obstacle to classifying chiral conformal field theories, in principle, lattices are very well-understood objects. For example, there are even algorithms for enumerating all of the unimodular lattices in a given dimension \cite{Kne57}. One could therefore imagine an approach to the classification of chiral conformal field theories wherein one morally treats lattices, or free boson theories, as an infinitely abundant resource: what would be left to do if a powerful oracle could hand you any lattice you wanted?

This is where the second vertex-algebraic generalization of the notion of lattice genus enters. To explain this, we take as our starting point the following alternative characterization of a genus of lattices: two lattices $L$ and $K$ belong to the same genus if and only if $L\oplus \II_{1,1}$ is isometric to $K\oplus \II_{1,1}$, where $\II_{1,1}$ is the unique, even unimodular lattice of signature $(1,1)$. Just as the previous formulation in terms of  discriminant forms could be directly translated to the setting of rational chiral algebras, this definition also admits a natural extension: one declares that $V$ and $W$ belong to the same \emph{hyperbolic genus} if $V\otimes V_{\II_{1,1}}\cong W\otimes V_{\II_{1,1}}$ as conformal vertex algebras, where we have added the prefix ``hyperbolic'' to differentiate this notion of genus from bulk genus.\footnote{The name comes from the fact that the lattice $\II_{1,1}$ has hyperbolic signature, and is often called the hyperbolic plane.} This definition was first introduced in \cite{Mor21}.

For a physicist, the lattice vertex algebra $V_{\II_{1,1}}$ is unnatural to work with because $\II_{1,1}$ is of indefinite signature, so we offer a more physical interpretation of hyperbolic genus.\footnote{See also \cite{Mor23} and \autoref{remph:currentcurrent} for an equivalent characterization of hyperbolic equivalence in terms of current-current deformations.} In order to do this, we must take a brief digression and explain a certain canonical decomposition that any rational chiral algebra possesses into a free and interacting sector. 

Consider a rational chiral algebra $V$ which has a continuous global symmetry group $G$ of rank $r$. After choosing a Cartan subalgebra, one obtains, by Noether's theorem, spin-1 conserved currents $J^i(z)$ which generate a maximal torus of $G$, where $i=1,\dots, r$. These currents can in turn be bosonized, $J^i(z) = \partial \varphi^i(z)$, so that $V$ admits an operator subalgebra $H$ which is generated by $r$ chiral bosons. For certain values of $\lambda \in \mathbb{R}^r$, the vertex operator $e^{i\lambda \cdot \varphi(z)}$ built out of these chiral bosons will appear as an operator inside $V$; the collection of all such $\lambda$ will actually form a lattice $L$ whose dual $L'$ contains the charge lattice of $V$. In particular, the algebra of the corresponding vertex operators closes onto a lattice chiral algebra $V_L$ which contains $H$ as a conformal subalgebra. We refer to $L$ as the associated lattice of $V$, and $V_L$ as the \emph{free sector} of $V$. 

On the other hand, we can define an interacting part of $V$ as ``the complement of the free sector.'' More precisely, we define the \emph{interacting sector} $C$ of $V$ as the space of all operators in $V$ which commute with the operators in the subalgebra $V_L$, i.e.\ $C$ is the coset of $V$ by its free sector. By construction, $C$ will be a rational chiral algebra which does not possess any continuous global symmetries, i.e.\ it does not have spin-1 conserved currents. Rather than factorizing into a trivially decoupled tensor product $C\otimes V_L$ of its free sector and its interacting part, $V$ will generally decompose into a direct sum of finitely many irreducible representations of $C\otimes V_L$, 
\begin{equation*}
    V\cong \bigoplus_{\alpha+L\in A} C^{\tau(\alpha+L)}\otimes V_{\alpha+L}.
\end{equation*}
Here, the irreducible modules of $V_L$ are denoted $V_{\alpha+L}$, where $\alpha+L$ takes values in the discriminant form $L'/L$, and $A<L'/L$ is the subgroup of irreducible $V_L$-modules which appear in $V$. The irreducible modules of $C$ are written $C^\beta$, and $\tau$ describes how the representations of $V_L$ are glued to those of $C$; that is, $\tau$ is a map which associates to each $\alpha+L\in A$ an irreducible $C$-module $C^{\tau(\alpha+L)}$ which appears next to $V_{\alpha+L}$ in the decomposition of $V$. 

In \autoref{subsec:equivchar}, we prove the following alternative characterization of hyperbolic genus.

\smallskip

\noindent \textbf{\autoref{thm:hypcomm}}. \emph{Two rational chiral algebras belong to the same hyperbolic genus if and only if a) their interacting parts are isomorphic, b) their associated lattices belong to the same lattice genus, and c) their free and interacting sectors are glued together ``in a compatible manner''.}

\smallskip

One immediate corollary of this characterization (\autoref{cor:hypbulk}) is that if two rational chiral algebras belong to the same hyperbolic genus, then they belong to the same bulk genus. Thus, hyperbolic equivalence is a finer equivalence relation than bulk equivalence, and we may contemplate partitioning each bulk genus into different hyperbolic genera. 

It turns out that, in the case of chiral conformal field theories, one can determine whether two theories are hyperbolically equivalent just by checking that their interacting sectors are the same (see \autoref{cor:holhyp}, and also \autoref{cor:unpointed} for a more general statement). In particular, a hyperbolic genus of chiral conformal field theories precisely consists of theories which differ only in their free sectors, but whose interacting parts are identical; for example, the more than one billion free boson theories predicted by King all collapse into a single hyperbolic equivalence class. 

Thus, classifying chiral conformal field theories modulo hyperbolic equivalence formalizes the idea of operating in a world where lattices are cheap. This suggests to us a program of classifying $c=32$ chiral conformal field theories wherein one attempts to enumerate all possibilities for the interacting sector $C$, but ignores the issue of producing a complete list of  free sectors $V_L$ to which $C$ can be glued.  We do not attempt to completely solve this problem, but we offer several sources of hyperbolic genera of $c=32$ chiral conformal field theories in \autoref{subsec:c=32}.

\medskip

We conclude this physics introduction by highlighting one more application of these constructions: it turns out that the hyperbolic genus of a chiral conformal field theory  defines an ensemble of conformal field theories whose disorder average admits a ``holographic'' dual in the spirit of \cite{MW20,ACHT21,DHJ21}.

Let us first describe a natural measure on a hyperbolic genus, following \cite{Mor21}. Once again, let $V$ be a chiral algebra and $\{J^i(z)\}$ a linearly independent collection of $r$ commuting currents which generate a maximal torus of the continuous symmetry group of $V$. Recall that $V$ decomposes into an interacting part $C$ and a free sector $V_L$. Any symmetry $\phi\colon V\to V$ which preserves the currents $J^i(z)$, in the sense that it sends $\spn_{\C}\{J^i(z)\}$ to itself, will shuffle vertex operators $e^{i\lambda\cdot \varphi}$ in $V_L$ among themselves, and hence induce an isometry $\phi^\ast$ of the lattice $L$. We call $G_V$ the subgroup of $\Aut(L)$ which is induced by symmetries of $V$, and define a probability measure which assigns $V$ the weight $1/|G_V|$.

It turns out that, for $W$ some fixed chiral algebra with associated lattice $K$, one has the following formula (Theorem 4.16 of \cite{Mor21}):
\begin{equation*}
   \mathrm{mass}(W)\coloneqq\sum_{V\in\mathrm{hgen}(W)} \frac{1}{|G_V|} \propto \sum_{L\in \mathrm{gen}(K)} \frac{1}{|\Aut(L)|} \eqqcolon ~\mathrm{mass}(K)
\end{equation*}
where $\mathrm{hgen}(W)$ is the hyperbolic genus of $W$ and $\mathrm{gen}(K)$ is the lattice genus of $K$. The sum over lattices on the right-hand side is known as the \emph{mass} of the genus of $K$, and is exactly computable by a widely-celebrated formula of Smith--Minkowski--Siegel \cite{CS88}.  Furthermore, the constant of proportionality between $\mathrm{mass}(W)$ and $\mathrm{mass}(K)$  is described in \autoref{subsec:mass}, and its explicit computation is facilitated by several technical lemmas that we prove. Because $\mathrm{mass}(W)$ straightforwardly furnishes a lower bound on the number of chiral algebras in $\mathrm{hgen}(W)$, one may use this formula to estimate the growth of the number of theories as a function of central charge, which we carry out in an example involving chiral conformal field theories in \autoref{subsec:holomass}.

In \autoref{conj:generalizedSiegelWeil}, we write down a generalization of this mass formula to a Siegel--Weil identity in the case that $W$ is a chiral conformal field theory of arbitrary central charge $c$. Our formula, which we check in examples, says that, when $r\geq 4$,
\begin{equation*}
    \sum_{V\in\mathrm{hgen}(W)} \frac{Z_V(\tau)}{|G_V|} = \mathrm{mass}(W)\sum_{\gamma\in \Gamma_\infty\backslash \mathrm{SL}_2(\Z)} \epsilon(\gamma)^c \frac{\mathrm{ch}_C(\gamma\tau)}{\eta(\gamma\tau)^r}
\end{equation*}
where $Z_V(\tau)$ is the genus-1 partition function of $V$, i.e.
\begin{equation*}
    Z_V(\tau) = \tr_Vq^{L_0-c/24}, \ \ q=e^{2\pi\i \tau}
\end{equation*}
and similarly $\mathrm{ch}_C(\tau)$ is the vacuum character of the interacting sector $C$ of $W$. Here, $\Gamma_\infty=\{\pm\left(\begin{smallmatrix} 1 & n \\ 0 & 1 \end{smallmatrix}\right)\mid n\in\Z\}$ and $\epsilon$ is the ``multiplier system'' for the Dedekind eta function $\eta$. This can be thought of as a formula for the ``average'' partition function of a chiral conformal field theory in the hyperbolic genus of $W$. In \autoref{subsec:SiegelWeil}, we explain a procedure for evaluating the right-hand side of this equation in terms of Eisenstein series for congruence subgroups, described in \autoref{app:eisenstein}. Moreover, we provide a kind of holographic interpretation of the right-hand side as a sum of a particular bulk topological field theory over a family of $\mathrm{SL}_2(\Z)$ black holes. This formula furnishes a kind of chiral generalization of the results of \cite{MW20,ACHT21,DHJ21} to arbitrary chiral conformal field theories, not just those of Narain type. 
%


\subsection*{Notation}

All Lie algebras and vertex (operator) algebras will be over the base field $\C$. All categories will be enriched over $\Vect=\Vect_\C$.

This is a mathematical manuscript. However, from time to time we use the notation Theorem$^{\mathrm{ph}}$, Proposition$^{\mathrm{ph}}$, etc.\ to denote statements that are established or well-defined only at a physics level of rigor.


\subsection*{Acknowledgments}
We thank Yichul Choi, Gerald Höhn, Hannes Knötzele, Zohar Komargodski, Robert McRae, Yuto Moriwaki, Sunil Mukhi, Nils Scheithauer, Sahand Seifnashri, Shu-Heng Shao and Xiao-Gang Wen for helpful discussions.

Sven Möller acknowledges support from the DFG through the Emmy Noether Program and the CRC 1624 \emph{Higher Structures, Moduli Spaces and Integrability}, project numbers 460925688 and 506632645. Brandon Rayhaun gratefully acknowledges NSF grant PHY-2210533.


\section{Preliminaries}\label{sec:prelim}

In this section, we introduce the central mathematical notions of this text, namely even lattices, modular tensor categories and \voa{}s. We also introduce the notion of a Cartan subalgebra of a \voa{} and describe the corresponding root-space decomposition. Finally, we comment on the various connections of these objects to physics.


\subsection{Lattices}\label{sec:lat}

We begin by reviewing various lattice theoretic notions, including genus and mass. We shall restrict our attention to even lattices. This discussion serves as motivation for studying related concepts for (suitably regular) \voa{}s in subsequent sections. 

\medskip

Let $L=(L,\langle\cdot,\cdot\rangle)$ be a \emph{lattice}, i.e.\ a free $\Z$-module $L$ of rank $\rk(L)\in\N$ equipped with a non-degenerate symmetric bilinear form $\langle\cdot,\cdot\rangle\colon L\times L\to\Q$. The lattice $L$ is called \emph{even} if $\langle\alpha,\alpha\rangle/2\in\Z$ for all $\alpha\in L$.
For an even lattice $L$, let
\begin{equation*}
L'=\{v\in L\otimes_\Z\Q\mid\langle v,\alpha\rangle\in\Z\text{ for all }\alpha\in L\}\supseteq L
\end{equation*}
denote the dual lattice, viewed as a lattice in the ambient space $L\otimes_\Z\Q$ of $L$ via the bilinear form. Moreover, let $L'/L$ be the \emph{discriminant form} of $L$, a finite abelian group equipped with a non-degenerate quadratic form $Q\colon L'/L\to\Q/\Z$ induced by $\langle\cdot,\cdot\rangle$ (see, e.g., \cite{Nik80,CS99}). A finite abelian group $D$ equipped with a non-degenerate quadratic form $Q$ is sometimes referred to as a \emph{metric group}. Lastly, let
$\sign(L)=(r_+,r_-)$ with $r_++r_-=\rk(L)$ denote the signature of $L$ over $\R$, i.e.\ the signature of $L\otimes_\Z\R$. The lattice $L$ is called \emph{unimodular} if $L'/L$ is trivial. We shall typically only be interested in lattices up to isomorphism, i.e.\ up to integral equivalence.

Given a metric group $(D,Q)$, and a pair of non-negative integers $(t_+,t_-)$, one may ask when there exists an even lattice $L$ with $L'/L\cong D$ and $\mathrm{sign}(L)=(t_+,t_-)$.  A necessary and sufficient condition for the existence of such a lattice is given in Theorem~1.10.1 of \cite{Nik80} (see also \cite{CS99}). Here, we cite only its Corollary~1.10.2, which states that it is sufficient that $t_+-t_-\equiv \mathrm{sign}(Q)\pmod 8 $ and that $t_++t_->l(D)$, where $l(D)$ is the minimum number of generators of $D$ as an abelian group and $\mathrm{sign}(Q)\in\Z_8$ is the ``signature (mod $8$)'' of $Q$ (see Theorem 1.3.3 of \cite{Nik80}). This is a kind of \emph{reconstruction} result, which one hopes to recover for modular tensor categories and \voa{}s (see \autoref{subsec:reconstruction}).

\medskip

Given an even lattice $L$, we can consider its group of isometries $\Aut(L)=\O(L)$, which is finite when $L$ is positive-definite. For an isometry $\nu\in\Aut(L)$ of order~$m$, we say that $\nu$ has Frame shape $\prod_{t\mid m}t^{b_t}$ with $b_t\in\Z$ if the extension of $\nu$ to $L\otimes_\Z\C$ has the characteristic polynomial $\prod_{t\mid m}(x^t-1)^{b_t}$.

An in many ways special lattice is the Leech lattice $\Lambda$. It is the unique positive-definite, even, unimodular lattice of rank~$24$ without roots. We remark that the elements of $\Aut(\Lambda)$ are uniquely specified by their Frame shapes. 

\medskip

For two even lattices $L$ and $M$ (or two isomorphism classes of such lattices) it is well-known that the following are equivalent \cite{Nik80,CS99}.

\begin{defi}[Lattice Genus]\label{defi:latticegenus}
Two even lattices belong to the same \emph{genus} if any one of the following equivalent conditions hold:
\begin{enumerate}[label=(\alph*)]
\item\label{item:latgen1} $L\otimes_\Z\R\cong M\otimes_\Z\R$ and  $L\otimes_\Z\Z_p\cong M\otimes_\Z\Z_p$ for all primes $p$.
\item\label{item:latgen2} $L'/L\cong M'/M$ and $\sign(L)=\sign(M)$.
\item\label{item:latgen3} $L\oplus\II_{1,1}\cong M\oplus\II_{1,1}$.
\end{enumerate}
\end{defi}
Here (and only here), $\Z_p$ denotes the ring of $p$-adic integers, and $\II_{1,1}$ is the unique even, unimodular lattice of signature $(1,1)$, sometimes called the \emph{hyperbolic plane}. This defines an equivalence relation.
We denote the genus (or equivalence class) of a lattice $L$ by $\gen(L)$. Each genus contains only finitely many isomorphism classes of even lattices.

\medskip

There are also tools, specifically the Smith--Minkowski--Siegel mass formula (see, e.g., \cite{CS88,CS99}), for computing the \emph{mass} of a genus of positive-definite, even
lattices, which is defined as
\begin{equation*}
    \mass(L) = \sum_{M\in\mathrm{gen}(L)} \frac{1}{|\Aut(M)|},
\end{equation*}
where $|\Aut(M)|$ is the order of the automorphism group of $M$, which is finite since $M$ is positive-definite, and the sum runs over the finitely many isomorphism classes of lattices in the genus of $L$. For example, if $L=(E_8)^{n}$ for $n\in\N$, where $E_8$ is the (unimodular) root lattice of type $E_8$, then the mass of $L$ is a weighted sum over the positive-definite, even, unimodular lattices of dimension $8n$, and the result is known to take the form
\begin{equation}\label{eqn:unimodularmass}
\mass((E_8)^n) = \frac{|B_{4n}|}{8n}\prod_{1\leq j<4n}\frac{|B_{2j}|}{4j}
\end{equation}
where the $B_k$ are Bernoulli numbers. 

Mass formulae have played an important role in lattice theory. For example, if one purports to have classified all lattices in a given genus, say via Kneser's neighborhood method \cite{Kne57}, then the mass formula furnishes an explicit check of correctness. Furthermore, $2\cdot\mass(L)$ is straightforwardly seen to be a lower bound on the number of lattices in $\mathrm{gen}(L)$, albeit a relatively weak one. One may therefore estimate, e.g., how the number of positive-definite, even lattices $L$ with a given discriminant form $D$ grows with $\rk(L)$.


\subsection{Fusion Categories}\label{sec:cat}

In the following, we briefly and telegraphically review some categorical notions that are relevant to this treatise. The central objects are spherical (pivotal) fusion categories and modular tensor categories. We refer to \cite{JS93,BK01,Tur10,EGNO15} for any details which we have omitted. 

\medskip

A fusion category (over $\C$) is a $\C$-linear, semisimple, rigid monoidal category with finitely many simple objects up to isomorphism, finite-dimensional hom-spaces and with its monoidal unit being simple. We index the finite set $\Irr(\Cat)$ of isomorphism classes of simple objects of a fusion category $\Cat$ by $0,\dots,\rk(\Cat)-1$ with $0$ representing the monoidal unit, where $\rk(\Cat)=|\Irr(\Cat)|$ is the rank of $\Cat$. For an object $X\in\Cat$, we denote by $X^*$ the dual object. We use $\Cat\boxtimes\mathcal{D}$ to denote the Deligne product of two fusion categories.

One says that a fusion category is \emph{pseudo-unitary} if the global Frobenius--Perron dimension equals the global categorical dimension. Equivalently, such a fusion category admits a (unique) spherical pivotal structure with respect to which the categorical dimensions $d_X$ of all simple objects are positive, and hence coincide with the Frobenius-Perron dimensions (which are always positive real numbers). See \cite{ENO05} for a thorough discussion. By a slight abuse of notation, when given a spherical fusion category that is pseudo-unitary, we shall mean that the spherical structure is actually the one with the mentioned property.

A fusion category is called \emph{pointed} if all of its simple objects are invertible with respect to the monoidal tensor product. It is known that every pointed fusion category is of the form $\Vect_G^\omega$ for $G$ a finite group; this is the category of $G$-graded vector spaces, with the associativity isomorphisms twisted by a 3-cocycle $\omega\in H^3(G,\C^\times)$.

\medskip

A braiding on a fusion category is a natural isomorphism $c_{X,Y}\colon X\otimes Y\to Y\otimes X$, $X,Y\in\Cat$, that satisfies the hexagon axioms. Given a fusion subcategory $\mathcal{D}\subset\Cat$, we define the Müger centralizer $\mathcal{D}'$ to be the full subcategory of objects $X$ of $\Cat$ such that $c_{Y,X}\circ c_{X,Y}=\id_{X\otimes Y}$ for all $Y$ in $\Cat$. A spherical, braided fusion category is said to be a \emph{modular tensor category} if the Müger center $\Cat'=\Vect$ is trivial.

We note that spherical pivotal structures on braided fusion categories are in bijection with ribbon structures (given by the ribbon twists $\theta_X\colon X\to X$). We consider two modular tensor categories to be the same if they are ribbon (and in particular braided and tensor) equivalent. Given a modular tensor category $\Cat$, we define the \emph{ribbon-reversed} category $\overline{\Cat}$ to be the modular tensor category $\Cat$ with its twist inverted and its braiding reversed, i.e.\ $\overline{c}_{X,Y}\coloneqq c_{Y,X}^{-1}$. A ribbon-reversed equivalence of $\Cat$ and $\mathcal{D}$ is a ribbon equivalence $\Cat\cong\overline{\mathcal{D}}$. Given a spherical fusion category $\Cat$, we use $Z(\Cat)$ to denote its Drinfeld center, which is guaranteed to be a modular tensor category.

Recall that the un-normalized $S$-matrix of a modular tensor category $\mathcal{C}$ is defined as $\tilde{s}_{X,Y}=\tr(c_{Y,X^\ast}\circ c_{X^\ast,Y})$. Further, we encode the twists $\theta_X$ of $\mathcal{C}$ into a diagonal $T$-matrix $\tilde{t}_{X,Y}=\delta_{X,Y}\theta_Y$, where $X,Y\in\mathrm{Irr}(\mathcal{C})$. These matrices satisfy \cite{BK01}
\begin{equation}\label{eqn:modularrelations}
    (\tilde{s}\tilde{t})^3=p^+\tilde{s}^2,\quad\tilde{s}^2 = p^+p^- C,\quad C\tilde{t}=\tilde{t}C,\quad C^2=\mathrm{id}
\end{equation}
where $p^\pm = \sum_{X\in\mathrm{Irr}(\mathcal{C})} d_X^2 \theta_X^{\pm 1}$ are the Gauss sums of $\mathcal{C}$ and $C_{X,Y} = \delta_{X,Y^\ast}$ is the charge conjugation matrix of $\mathcal{C}$. The \emph{chiral central charge} $c(\mathcal{C})\in\Q/8\Z$ of $\mathcal{C}$ is a rational number defined modulo 8 by the equation 
\begin{equation}\label{eqn:chiralcentralcharge}
e^{2\pi\i c(\mathcal{C})/8} = p^+/D,\quad D=\sqrt{\sum_{X\in\mathrm{Irr}(\mathcal{C})} d_X^2}.
\end{equation}
We define the \emph{normalized} $S$-matrix of $\mathcal{C}$ by $s=\tilde{s}/D$.

\medskip

The pointed braided fusion categories are, up to braided monoidal equivalence, in bijection with pre-metric groups, i.e.\ finite abelian groups $D$ together with a quadratic form $Q\colon D\to\Q/\Z$ \cite{JS93} (see also \cite{EGNO15}). Let $\mathcal{C}(D)$ denote the pointed braided fusion category corresponding to $D=(D,Q)$. Any $\mathcal{C}(D)$ is pseudo-unitary. Indeed, on each $\mathcal{C}(D)$, there are $\Hom(D,\{\pm1\})$ many choices of ribbon (or equivalently spherical pivotal) structures, each making $\mathcal{C}(D)$ into a ribbon fusion category. The pseudo-unitarity means that there is a unique choice of such structure such that all categorical dimensions are positive (i.e.\ equal to $1$). In the following, let us by $\mathcal{C}(D)$ denote the ribbon fusion category corresponding to this choice. Finally, if $D=(D,Q)$ is a metric group, i.e.\ if the bilinear form $B\colon D\times D\to\Q/\Z$ associated with $Q$ is non-degenerate, then $\mathcal{C}(D)$ is a modular tensor category.
(We refer the reader to the discussion in Section~2.4 of \cite{HM23} and the references cited therein for a more detailed discussion.)


\subsection{Vertex Algebras}

We refer readers to, e.g., \cite{Bor86,FLM88,Kac98,FBZ04} for any background on vertex (operator) algebras we omit here. A \emph{vertex algebra} is a vector space $V$ together with a linear map
\begin{equation*}
Y(\cdot,x)\colon V\to\End_\C(V)[[x^{\pm1}]],\quad a\mapsto Y(a,x)=\sum_{n\in\Z}a_nx^{-n-1}
\end{equation*}
and a vacuum vector $0\neq\vac\in V$ subject to a number of well-motivated axioms.

A \emph{conformal vertex algebra} is a $\Z$-graded vertex algebra $V=\bigoplus_{n\in\Z}V_n$ equipped with a distinguished vector $\omega\in V_2$ called conformal vector, which must satisfy the following conditions. First, its modes $L_n=\omega_{n+1}$, $n\in\Z$, should generate the Virasoro algebra, i.e.\
\begin{equation*}
    [L_n,L_m] = (n-m)L_{n+m}+\frac{1}{12}\delta_{n+m,0}(n^3-n)c,
\end{equation*}
for some $c\in\C$ called the central charge. Second, the conformal vector should be compatible with the grading in the sense that $L_0 a = n a$ when $a\in V_n$. Lastly, we require that $L_{-1}a = a_{-2}\vac$.

A conformal vertex algebra is called a \emph{vertex operator algebra} if $\dim(V_n)<\infty$ for all $n\in\Z$ and $V_n=\{0\}$ for $n$ sufficiently small.

\medskip

We shall typically demand a number of regularity conditions on the \voa{}s we consider. A vertex operator algebra is said to be of CFT-type if $V=\bigoplus_{n\geq 0}V_n$ and $V_0 = \C \vac$. It is said to be $C_2$-cofinite (or lisse) if $V/C_2(V)$ is finite-dimensional, where $C_2(V) = \spn_{\C}\{a_{-2}b\mid a,b\in V\}$. A vertex operator algebra is self-contragredient if $V\cong V^\ast$ as $V$-modules, where $V^\ast$ is the contragredient dual of $V$. Finally, the following defines the class of vertex operator algebras to which we mainly restrict ourselves in this work.

\begin{defi}[Strong Rationality]\label{defi:strongrationality}
A vertex operator algebra is \emph{strongly rational} if it is simple, rational (see \cite{DLM97} for the definition), $C_2$-cofinite, self-contragredient and of CFT-type.
\end{defi}

A strongly rational vertex operator algebra has finitely many irreducible modules $M\in\Irr(V)$, which we typically label by $V^0=V,V^1,V^2,\dots$, and every module is a direct sum of irreducible ones. We use $\rho(M)$ to denote the smallest eigenvalue of $L_0$ when acting on $M$. We refer to $V$ itself as the vacuum module, which has $\rho(V)=0$. Strong rationality also implies that the central charge $c$ and all $\rho(M)$ for $M\in\Irr(V)$ are rational numbers \cite{AM88,DLM00}.

Crucially, for a \strat{} \voa{} $V$, the category of representations $\Rep(V)$ carries the structure of a modular tensor category \cite{Hua08b}, with monoidal unit $V$ and the contragredient duals $M^*$ being the rigid duals.

\medskip

We consider a further regularity condition on \voa{}s:
\begin{defi}[Positivity]\label{defi:positivity}
A \strat{} vertex operator algebra $V$ is said to be \emph{positive} if $\rho(M)>0$ for every irreducible module $M\in\Irr(V)$ other than the vacuum module $V$.
\end{defi}
This entails that the categorical dimensions of all irreducible modules are positive (and coincide with the quantum dimensions), i.e.\ that the modular tensor category $\Rep(V)$ is pseudo-unitary \cite{DJX13,DLN15}. Moreover, this implies that the central charge $c$ is non-negative (and positive as long as the \voa{} is not finite-dimensional) \cite{DM04b}.

\medskip

Let $(U,\omega')$ be a vertex operator subalgebra of $(V,\omega)$, where the conformal vector $\omega'$ of $U$ does not have to coincide with the conformal vector $\omega$ of $V$. The \emph{commutant} of $U$ inside $V$ is defined as
\begin{equation*}
V/U\coloneqq\Com_V(U)=\{a\in V \mid a_nU=0, \text{ for all }n\geq 0\}
\end{equation*}
\cite{FZ92,LL04}. We call a vertex operator subalgebra $U$ of $V$ \emph{primitive} if it is its own double commutant, $U\cong \Com_V(\Com_V(U))$.

A vertex operator subalgebra $(U,\omega')$ is called a conformal subalgebra of $(V,\omega)$ if $\omega=\omega'$. In that case, the commutant $\Com_V(U)$ is independent of $U$ and given by the center of $V$ (which is the trivial \voa{} $\C\vac$ if $V$ is simple).

\medskip

Finally, we introduce some examples of \strat{} \voa{}s, which shall play a role in this text.

\begin{ex}[Lattice \VOA{}s]\label{ex:lattice}
Given a positive-definite, even lattice $L$, we can consider the corresponding \emph{lattice \voa{}} $V_L$ of central charge $c=\rk(L)$ \cite{Bor86,FLM88}.

$V_L$ is \strat{} and its representation category is the (pseudo-unitary and pointed) modular tensor category $\Rep(V_L)=\mathcal{C}(L'/L)$ associated with the discriminant form $L'/L$ \cite{Don93,DL93,DLM00}.
\end{ex}

We remark that because all metric groups $D=(D,Q)$ can be realized as discriminant forms $L'/L$ of positive-definite, even lattices $L$ with large enough rank \cite{Nik80} (see \autoref{sec:lat}), it follows that all pointed, pseudo-unitary modular tensor categories $\Cat(D)$ can be realized as representation categories of \strat{} \voa{}s, namely of lattice \voa{}s. This is the easiest instance of reconstruction, which we shall discuss in \autoref{subsec:reconstruction}.

The above example can be generalized to not necessarily positive-definite, even lattices $L$, in which case the same construction only yields a conformal vertex algebra $V_L$. However, it still follows from the results in \cite{Don93,DLM97,DL93} that the representation category of a lattice vertex algebra $V_L$ is the pointed modular tensor category $\mathcal{C}(L'/L)$ associated with the discriminant form $L'/L$. Note that it is in particular shown that all weak $V_L$-modules are direct sums of the finitely many irreducible modules, which are indexed by the cosets in $L'/L$.

\begin{ex}[Simple Affine \VOA{}s]\label{ex:affine}
Given a simple Lie algebra $\g$ of type $X_n$ and $k\in\C$, let $V_\g(k,0)$ denote the corresponding universal affine vertex algebra at level $k$ (\cite{FZ92}, see also \cite{LL04}). If $k\neq -h^\vee$, where $h^\vee$ is the dual Coxeter number of $\g$, then $V_\g(k,0)$ can be endowed with the Sugawara conformal structure and is a \voa{} of central charge $c=k\dim(\g)/(k+h^\vee)$. For $k\in\Ns$, $V_\g(k,0)$ has a maximal proper ideal and the corresponding simple quotient is the \emph{simple affine \voa{}} (or \emph{current algebra} or \emph{Wess--Zumino--Witten model}) $L_\g(k,0)$. We shall denote $L_\g(k,0)$ by $\mathsf{X}_{n,k}$ in the following.

The \voa{} $\mathsf{X}_{n,k}$ is \strat{} and its representation category is the (pseudo-unitary) modular tensor category that is often denoted by $\Rep(\mathsf{X}_{n,k})=(X_n,k)$.

We note that for $X_n$ simply-laced and level $k=1$, $\mathsf{X}_{n,1}$ is isomorphic to the lattice \voa{} $V_{X_n}$ where $X_n$ now denotes the root lattice of $\g$.
\end{ex}

\begin{ex}[Parafermion \VOA{}s]\label{ex:para}
For each simple Lie algebra $\g$ and positive integer $k\in\Ns$, we may obtain a \voa{}
\begin{equation*}
    K(\g,k) \coloneqq \Com_{L_{\g}(k,0)}(\langle\hh\rangle),
\end{equation*}
where $\langle\hh\rangle\cong M_{\hat\hh}(k,0)$ is the (non-full) vertex operator subalgebra of $L_{\g}(k,0)$ generated by a choice of Cartan subalgebra $\mathfrak{h}$ of $\g$. The $K(\g,k)$ are known to be \strat{} \voa{}s \cite{ALY14,DR17} and are referred to as \emph{parafermion \voa{}s}.

In the special case that $\g=\sl_2$, we use the notation $\mathcal{P}(k)\coloneqq K(\mathfrak{sl}_2,k)$ and refer to the algebras as $\Z_k$-parafermion \voa{}s because $\Rep(\mathcal{P}(k))$ possesses an order-$k$ simple current. Following \cite{MP01}, we label the irreducible modules $\mathcal{P}(k,[\ell,m])$ of $\mathcal{P}(k)=\mathcal{P}(k,[0,0])$ by pairs $(\ell,m)$ of integers satisfying $0\leq \ell\leq k$, $0\leq |m|\leq \ell$ and $\ell-m\in 2\Z$. The $\Z_k$-parafermion vertex operator algebra has central charge $c=\frac{2(k-1)}{k+2}$ and the irreducible module $\mathcal{P}(k,[\ell,m])$ has conformal weight
\begin{equation*}
   h_{\ell,m}=\frac{\ell(\ell+2)}{4(k+2)}-\frac{m^2}{4k}.
\end{equation*}
\end{ex}


\subsubsection{Inner Automorphisms}\label{sec:innerauts}

Let $V$ be a conformal vertex algebra (with conformal vector $\omega$). Then an automorphism of $V$ is a vector-space automorphism $g\in\Aut_\C(V)$ satisfying
\begin{equation*}
gY(\cdot,x)g^{-1}=Y(g\cdot,x),\quad g\vac=\vac\quad\text{and}\quad g\omega=\omega,
\end{equation*}
i.e.\ intertwining the vertex operators and fixing the vacuum and conformal vectors. Removing the last condition yields a vertex algebra automorphism.

It follows from the commutator formula that the zero-mode $a_0$ of any vector $a\in V$ acts as a derivation on the vertex operators, i.e.
\begin{equation*}
a_0Y(b,x)c=Y(a_0b,x)c+Y(b,x)a_0c
\end{equation*}
for all $b,c\in V$. This entails that if the exponential $e^{a_0}=\sum_{n=0}^\infty\frac{a_0^n}{n!}$ is well-defined, it defines an automorphism of $V$ (see, e.g., \cite{MN99}). Automorphisms generated by such exponentials are called \emph{inner automorphisms}. If $V$ has a conformal structure, we need to additionally demand that $a_0\omega=0$ so that $e^{a_0}$ fixes the conformal vector $\omega$ in order for $e^{a_0}$ to be an automorphism of the conformal vertex algebra.

\medskip

To avoid convergence issues, let us from now on assume that $V$ is a \voa{}, and in particular has finite-dimensional weight spaces. Also assume that $V$ is of CFT-type. Then, any vector $a\in V_1$ (which, we recall, forms a finite-dimensional Lie algebra under the zero mode) defines a \voa{} automorphism of $V$. Indeed, $a$ having $L_0$-weight $1$ and the skew-symmetry formula imply $a_0\omega=0$. Then, the \emph{inner automorphism group} is the normal subgroup of $\Aut(V)$ defined as
\begin{equation*}
\Inn(V)\coloneqq\langle\{e^{a_0}\mid a\in V_1\}\rangle<\Aut(V).
\end{equation*}
It is clear from the definition that every inner automorphism of the Lie algebra $V_1$ can be extended to an inner automorphism of $V$, and conversely that every inner automorphism of $V$ is an inner automorphism when restricted to $V_1$.


\subsubsection{Cartan Subalgebras}\label{sec:cartan}

In this section, in analogy to Lie algebras, we shall introduce the notion of a Cartan subalgebra of a conformal vertex algebra. Recall that over an algebraically closed field of characteristic zero, a Cartan subalgebra of a Lie algebra can be equivalently defined as a maximal toral subalgebra.

\medskip

Recall that we fix the base field $\C$. A toral subalgebra of a Lie algebra is a subalgebra such that all its elements act semisimply on the Lie algebra in the adjoint action. Now, let $V$ be a vertex algebra. Of a toral subalgebra of $V$ we shall demand that all its elements $a$ act semisimply on $V$ via the zero-mode $a_0\in\End_\C(V)$. In the Lie algebra setting, the skew-symmetry property implies that a toral subalgebra is automatically abelian. In order to replicate this important feature for a vertex algebra $V$, we shall want to ensure that the zero-mode is skew-symmetric. Finally, for a \voa{} or more generally a conformal vertex algebra, it is natural to demand some compatibility with the conformal structure. This leads to the following definition:
\begin{defi}[Toral and Cartan Subalgebra]\label{defi:va_cartan}
Let $V$ be a conformal vertex algebra. A subspace $\hh$ of $V$ is called a \emph{toral subalgebra} of $V$ if
\begin{enumerate}
\item $h_0h'\in\hh$, $h_1h'\in\C\vac$ and $h_nh'=0$ for all $n\geq2$,
\item $h_0$ acts semisimply on $V$,
\item $L_0h=h$ and $L_nh=0$ for all $n\geq 1$
\end{enumerate}
for all $h,h'\in\hh$.

A maximal toral subalgebra $\hh$ of $V$ is called \emph{Cartan subalgebra} of $V$.
\end{defi}

Some remarks are warranted:
\begin{rem}\label{rem:cartan}
\begin{enumerate}[wide]
\item The (vertex algebraic) skew-symmetry formula for the zero-mode
\begin{equation*}
b_0a=-\sum_{i=0}^\infty(-1)^i\frac{L_{-1}^i}{i!}a_ib
\end{equation*}
for $a,b\in V$ as well as $h_1h'\in\C\vac$ and $h_nh'=0$ for all $n\geq2$ imply that the zero-mode is (properly) skew-symmetric when restricted to $\hh$, i.e.\ $h_0h'=-h'_0h$ for all $h,h'\in\hh$. (Hence, the zero-mode defines a Lie bracket on $\hh$ since in any vertex algebra the zero-mode is a derivation, which with skew-symmetry implies the Jacobi identity.)
\item But then, the same argument as for toral subalgebras of Lie algebras shows that $h_0h'=0$ for all $h,h'\in\hh$, i.e.\ that $H$ is actually an abelian Lie algebra under the zero mode.
\item The skew-symmetry formula also implies that the one-mode $\langle h,h'\rangle\vac\coloneqq h_1h'$ for $h,h'\in\hh$ defines a symmetric bilinear form on $\hh$.

We remark that this bilinear form may be degenerate. We shall see in \autoref{sec:asslat} that if $V$ is a \strat{} \voa{} and $\hh$ is a maximal toral subalgebra, then this bilinear form is always non-degenerate (and moreover is the restriction of the unique symmetric, invariant bilinear form on $V$ \cite{Li94}).

\item The two previous observations together with item~(1) of the definition then imply that $\hh$ generates a vertex subalgebra of $V$ that is isomorphic to a Heisenberg vertex algebra associated with the abelian Lie algebra $\hh$.

If the bilinear form $\langle\cdot,\cdot\rangle$ on $\hh$ is non-degenerate, then $\hh$ even generates a Heisenberg \voa{}, i.e.\ it carries a conformal structure given by the standard Heisenberg conformal vector $\omega^\hh$,
which is compatible with the conformal structure on $V$ in the sense that $\omega_2\omega^\hh=0$ and $\omega_1\omega^\hh=2\omega^\hh$. (These are essentially the conditions that allow us to define a conformal structure on the commutant of $\hh$ in $V$, see \cite{FZ92,LL04}.)
\item The third property of the definition implies that $h_0\omega=0$ by the skew-symmetry formula and hence $[h_0,L_n]=0$ for all $n\in\Z$ by the commutator formula for all $h\in\hh$.
\end{enumerate}
\end{rem}

Let $V$ and $V'$ be conformal vertex algebras. It is not difficult to see that if $\hh$ is a Cartan subalgebra of $V$ and $\hh'$ is a Cartan subalgebra of $V'$, then $\hh\oplus\hh'$, which is short for $\hh\otimes\C\vac\oplus\C\vac\otimes\hh'$, is a Cartan subalgebra of $V\otimes V'$.

\medskip

Suppose that $V$ is a \voa{} of CFT-type. Then the weight-$1$ space $V_1$ forms a Lie algebra under the zero mode. (Again, the zero-mode acts as a derivation and the CFT-type assumption forces the zero-mode to be skew-symmetric, which implies the Jacobi identity). In that situation, clearly, any toral subalgebra of $V$ (according to \autoref{defi:va_cartan}) is also a toral subalgebra of the Lie algebra $V_1$.

In the next section, we shall see that under suitable regularity assumptions on $V$ the converse is true.

More generally, one can ask:
\begin{prob}
Let $V$ be a conformal vertex algebra. Suppose that the zero-mode equips $V_1$ with the structure of a Lie algebra. Under which conditions are the Cartan subalgebras of $V_1$ exactly the Cartan subalgebras of $V$? 
\end{prob}
Of course, the point of our definition of toral subalgebra is that it is applicable even when $V_1$ is not a Lie algebra.

\medskip

We now discuss an example where the conformal vertex algebra is not necessarily a \voa{} of CFT-type (and $V_1$ not necessarily a Lie algebra).
\begin{ex}[Lattice Vertex Algebras]\label{ex:latticecartan}
Consider the lattice vertex algebra
\begin{equation*}
V_L=\bigoplus_{\alpha\in L}M_{\hat\hh}(1,\alpha)
\end{equation*}
associated with a (not necessarily positive-definite) even lattice $L$, where $\hh=L\otimes_\Z\C$. That is, $V_L$ is a conformal vertex algebra, but not necessarily a \voa{}. Then $\{k(-1)\ee_0\mid k\in L\otimes_\Z\C\}\cong\hh$ is a maximal toral subalgebra of $V_L$.
\end{ex}
\begin{proof}
It is clear that $\hh$ is a toral subalgebra. Let $x\in V$ be an element such that the linear span of $\hh$ and $x$ is also a toral subalgebra. Then $x$ must have eigenvalue~$1$ for $L_0$ and $0$ for any $h_0$, $h\in\hh$. By the definition of $V_L$, this implies that $x\in\hh$. This proves the maximality of $\hh$.
\end{proof}
We refer to the Cartan subalgebra $\hh$ in the above example as the standard Cartan subalgebra of a lattice vertex algebra.

\medskip

In analogy to Lie algebras, one can ask about the uniqueness of Cartan subalgebras of conformal vertex algebras:
\begin{prob}\label{prob:cartanconj}
Let $V$ be a conformal vertex algebra. Are all Cartan subalgebras of $V$ conjugate under (inner) automorphisms of $V$?
\end{prob}
In the next section, we shall give a positive answer to this question for suitably regular \voa{}s by reducing the problem to the corresponding question for the weight-$1$ Lie algebra $V_1$. In general, this problem is much harder. For an arbitrary conformal vertex algebra, even the mere definition of an inner automorphisms is problematic as convergence issues may arise (see, e.g., \cite{MN99}).


\subsubsection{Associated Lattice and Heisenberg Commutant}\label{sec:asslat}

In this section, we describe the Cartan subalgebras and the associated root-space decomposition of suitably regular \voa{}s. This corresponds to identifying a certain dual pair in the \voa{}, which then decomposes as a simple-current extension over it.

\medskip

The following considerations are explained in detail in, e.g., \cite{HM23,Hoe17} and are based on \cite{Mas14}. From now on, let $V$ be a \strat{} \voa{} (in particular of CFT-type). Let $V_1$ be the (complex, finite-dimensional) weight-$1$ Lie algebra, which is reductive by \cite{DM04b}.
\begin{defi}[Heisenberg Commutant, Associated Lattice]\label{defi:asslat}
Let $V$ be a \strat{} \voa{}. Let $\hh$ be a choice of Cartan subalgebra of $V_1$, which generates a Heisenberg vertex operator subalgebra of $V$. We then define the \emph{Heisenberg commutant} $C\coloneqq\Com_V(\hh)$ and consider the double commutant $\Com_V(\Com_V(\hh))\cong V_L$, which is isomorphic to a lattice \voa{} $V_L$ for some positive-definite, even lattice $L$, which we call the \emph{associated lattice} of $V$.
\end{defi}
We shall see in \autoref{prop:cartancartan} below that $\hh$ is already a Cartan subalgebra of $V$. Importantly, it is shown in \cite{Mas14} that the bilinear form (as defined in \autoref{rem:cartan}) on the Cartan subalgebra $\hh$ is non-degenerate.
This allows us to define the (standard) conformal structure on the Heisenberg vertex algebra generated by~$\hh$, which is also the (standard) conformal structure of $V_L$. Moreover, $V_L$ is \strat{} as it is a lattice \voa{}, and so is $C$ by \cite{CKLR19} (see also \cite{CGN21}).
\begin{rem}
$C$ and $L$ in the above decomposition are unique up to isomorphism. Indeed, $C$ is the commutant and $V_L$ the double commutant of a choice of Cartan subalgebra $\hh$ of the finite-dimensional Lie algebra $V_1$. But all Cartan subalgebras are conjugate under inner automorphisms of $V_1$, which then lift to inner automorphisms of $V$, as discussed in \autoref{sec:innerauts}. See also \autoref{cor:cartan_unique}.
\end{rem}

We conclude that $V$ is a (simple-current) conformal extension of the dual pair $C\otimes V_L$ in it. More precisely, using mirror extensions \cite{CKM22,Lin17}:
\begin{prop}\label{prop:asslatdecomp}
Let $V$ be a \strat{} \voa{}. Let $L$ denote the associated lattice and $C$ the Heisenberg commutant of $V$. Then $V$ is a simple-current extension of the \strat{} \voa{} $C\otimes V_L$,
\begin{equation*}
V=\bigoplus_{\alpha+L\in A}C^{\tau(\alpha+L)}\otimes V_{\alpha+L}
\end{equation*}
for some subgroup $A<L'/L$, with $\Rep(V_L\vert V)\cong\mathcal{C}(A)$ the corresponding full subcategory of $\Rep(V_L)\cong\mathcal{C}(L'/L)$, some pointed full subcategory $\Rep(C\vert V)$ of $\Rep(C)$ and some ribbon-reversing equivalence $\tau\colon\Rep(V_L\vert V)\to\Rep(C\vert V)$.
\end{prop}
Here, $\Irr(C)=\{C^i\}$ denotes the finitely many irreducible $C$-modules up to isomorphism (with $C^0\cong C$). Moreover, $\mathcal{C}(A)$ is the pointed ribbon fusion category (with positive categorical dimensions) in the notation of \cite{EGNO15,JS93} (see also \autoref{sec:cat}) associated with the abelian group $A$ equipped with the (possibly degenerate) quadratic form inherited from $L'/L$.

As a word of warning, we point out that not all simple-current extensions of $C\otimes V_L$ will have $C\otimes V_L$ as a dual pair or $C$ as its Heisenberg commutant.
We shall discuss a particular counterexample in \autoref{rem:wrongext}.

We also remark that in the above decomposition, the standard Cartan subalgebra of the lattice \voa{} $V_L$ (see \autoref{ex:latticecartan}) can be identified with the original choice of Cartan subalgebra in \autoref{defi:asslat}.

\begin{ex}
We mention two extreme cases: $V_1=\{0\}$ if and only if the associated lattice is trivial if and only if $V=C$.

On the other hand, $V=V_L$ is a lattice \voa{} (with associated lattice $L$) if and only if the Heisenberg commutant $C$ is trivial.
\end{ex}

\medskip

We return with our discussion to Cartan subalgebras of $V$. Based on the above decomposition into modules for the dual pair $C\otimes V_L$, it is easy to see that any Cartan subalgebra of $V_1$ is actually a Cartan subalgebra of $V$:
\begin{prop}\label{prop:cartancartan}
Let $V$ be a \strat{} \voa{}. Then Cartan subalgebras of $V_1$ are exactly the Cartan subalgebras of $V$.
\end{prop}
The statement also holds for (not necessarily maximal) toral subalgebras.
\begin{proof}
First, let $\hh$ be a Cartan subalgebra of $V_1$. We consider the statement of \autoref{prop:asslatdecomp}. By the decomposition stated there, $h_0$ acts semisimply on $V$ for all $h\in\hh$. Indeed, the vectors in $C^{\tau(\alpha+L)}\otimes M_{\hat\hh}(1,\alpha)\subset C^{\tau(\alpha+L)}\otimes V_{\alpha+L}$ have eigenvalue $\langle h,\alpha\rangle$ for $h_0,h\in\hh$. Overall, it is not difficult to see that $\hh$ is a toral subalgebra of $V$.

Now, $\hh$ is by definition a Cartan subalgebra of the Lie algebra $V_1$. Hence, $\hh$ is also a maximal toral subalgebra of $V$, as any further elements of the toral subalgebra would have to be in $V_1$, which contradicts $\hh$ being a Cartan subalgebra of $V_1$.

In summary, Cartan subalgebras of the Lie algebra $V_1$ are always Cartan subalgebras of the \voa{} $V$.

\smallskip

Conversely, let $\hh$ be a Cartan subalgebra of $V$. Then $\hh$ is a toral subalgebra of $V_1$. Suppose for the sake of contradiction that $\hh$ is not a maximal toral subalgebra of~$V_1$. That is, $\hh$ can be properly extended to a Cartan subalgebra $\hh'$ of $V_1$. But then, by \autoref{prop:asslatdecomp}, $\hh'$ is also a toral subalgebra of $V$, contradicting the maximality of $\hh$ as a toral subalgebra of $V$. Hence, $\hh$ was already a maximal toral subalgebra of $V_1$.

In total, for a \strat{} \voa{}, Cartan subalgebras of the Lie algebra $V_1$ are exactly the same as Cartan subalgebras of the \voa{} $V$.
\end{proof}

As already alluded to above, the above proposition implies:
\begin{cor}\label{cor:cartan_unique}
Let $V$ be a \strat{} \voa{}. Then all Cartan subalgebras of $V$ are conjugate under inner automorphisms of $V$.
\end{cor}
\begin{proof}
The claim follows since the same statement is true for the Lie algebra $V_1$ and inner automorphisms of $V_1$ lift to inner automorphisms of $V$ (see \autoref{sec:innerauts}).
\end{proof}
An important question is whether \autoref{cor:cartan_unique} also holds in situations where $V_1$ is not a Lie algebra (see \autoref{prob:cartanconj}).

\medskip

Another way to characterize toral and Cartan subalgebras of conformal vertex algebras $V$ is in the language of \cite{Mor21}. Essentially, the definitions there are more restrictive than ours but coincide, as soon as we can apply \autoref{prop:asslatdecomp}.
\begin{prop}
Let $V$ be a \strat{} \voa{} and $\hh$ some choice of Cartan subalgebra of $V$ (or equivalently of $V_1)$.

If the Heisenberg commutant $C$ is also positive, the pair $(V,\hh)$
satisfies Assumption~(A) in the sense of \cite{Mor21}.

If $(V,\hh)$ satisfies Assumption~(A), then $\hh$ is a Cartan subalgebra of $V_1$.
\end{prop}
\begin{proof}
The first assertion follows by applying \autoref{prop:asslatdecomp} and then comparing with the definitions in \cite{Mor21}. The second statement follows more directly and is already asserted in \cite{Mor21}.
\end{proof}
Hence, \autoref{prop:asslatdecomp} establishes the existence of the rather intricate structures in \cite{Mor21} in any \strat{} \voa{}.


\subsection{Physics}

Before moving past the preliminaries, we (telegraphically) describe various ways in which the mathematical concepts reviewed above arise in physics.

\medskip

First, we recall that a 3d topological quantum field theory (TQFT) is completely specified by a tuple $(\mathcal{C},c)$, where $\mathcal{C}$ is a modular tensor category, and $c\in\mathbb{Q}$ is the chiral central charge (see, e.g., \cite{BK01,Tur10}).
In this correspondence, the simple objects of the modular tensor category describe the anyons (or equivalently, topological line operators) of the TQFT, the tensor product describes their fusion, and so on and so forth. For example, TQFTs whose associated modular tensor categories are the pointed categories $\mathcal{C}(D)$ for some metric space $D$ can be realized by abelian Chern--Simons theories \cite{WZ92}. Indeed, letting $L$ be any even lattice with $L'/L\cong D$ (see \autoref{sec:lat}), we may choose a basis $\{\lambda_i\}$ of $L$ and define $K_{ij}=\langle\lambda_i,\lambda_j\rangle$, from which we obtain a Lagrangian description of $\mathcal{C}(D)$ given by
\begin{equation*}
    \mathcal{L}=\sum_{i,j} \frac{K_{ij}}{4\pi} a^ida^j,
\end{equation*}
where the $a^i$ are $\mathrm{U}(1)$-gauge fields. Another example of a 3d topological quantum field theory is non-abelian Chern--Simons theory, specified by the choice of a simple gauge Lie algebra $\g$ and a non-negative integer level $k$.

In this work, by chiral algebra we will mean the algebraic structure formed by the \emph{holomorphic}  local operators of a 2d conformal field theory (CFT) \cite{Zam85,MS89}. Mathematically, we take vertex operator algebras as our axiomatization of chiral algebras \cite{Bor86,FLM88}. We often focus on rational conformal field theories (RCFTs) with a unique vacuum, in which case the chiral algebra is expected to be described by a \strat{} vertex operator algebra, \autoref{defi:strongrationality}. For example, the lattice \voa{}s $V_L$ described in the previous subsection arise as chiral algebras of certain rational points in the conformal manifold of $\rk(L)$ free compact bosons, which is often referred to as a Narain moduli space. The current algebras $\mathsf{X}_{r,k}$ arise as the chiral algebras of Wess--Zumino--Witten models. 

It is often useful to think of a chiral algebra $V$ as forming a theory in its own right. This theory is not, strictly speaking, two dimensional: rather, one should think of the chiral algebra as living on the boundary of a three dimensional bulk TQFT \cite{Wit89,EMSS89}. One often refers to $V$ as a \emph{relative} theory. In the situation that $V$ is a \strat{} vertex operator algebra, its category of representations $\Rep(V)$ admits the structure of a modular tensor category, and hence defines the anyonic content of a 3d TQFT \cite{MS89}: the chiral algebra $V$ can be thought of as living at the boundary of $\mathcal{T}_V\coloneqq (\Rep(V),c(V))$, as depicted in \autoref{fig:3d2d}. The simplest example is that the abelian Chern--Simons theory associated with a lattice $L$ defines the bulk of the lattice chiral algebra $V_L$. Similarly, non-abelian Chern--Simons theory with $\mathfrak{g}=\mathsf{X}_r$ and level $k$ furnishes the bulk of the current algebra $\mathsf{X}_{r,k}$.

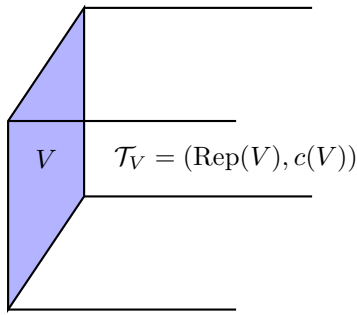
\begin{figure}[ht]
\begin{tikzpicture}
\filldraw[black,fill=blue!30,  thick] (-6,0) -- (-6,-1-1.5) -- (-5,-1) -- (1-6,1.5)    -- cycle;
\draw[black, thick](-2,-1)--(-5,-1);
\draw[black,  thick] (-6,0)--(-3,0);
\draw[black,  thick] (-3,-2.5) -- (-6,-1-1.5);
\draw[black, thick] (-2,1.5)--(1-6,1.5);
\node[] at (-.5-4+1.5,-.5) {$\mathcal{T}_V=(\Rep(V),c(V))$};
\node[] at (-2.5-3,-.5) {$V$};
\end{tikzpicture}
\caption{A boundary chiral algebra $V$ and the corresponding bulk 3d TQFT $\mathcal{T}_V$.}
\label{fig:3d2d}
\end{figure}

The three-dimensional perspective makes clear what the defining data of an RCFT is \cite{KS11}. Let us assume for ease of exposition that the gravitational anomaly (i.e.\ the left-moving minus the right-moving central charge) vanishes, though everything we say can be straightforwardly generalized. An RCFT with vanishing gravitational anomaly may be thought of as the specification of a left-moving chiral algebra $V$, a right-moving chiral algebra $\overline{W}$ with the same central charge, and a Lagrangian algebra object of $\Rep(V)\boxtimes\overline{\Rep(W)}$, which can be thought of more physically as a topological boundary condition of the 3d TQFT $\mathcal{T}_V\otimes\overline{\mathcal{T}_W}$.  Indeed, by placing this TQFT on an interval with the topological boundary $\mathcal{H}$ imposed at one end, and the chiral algebra boundary condition $V\otimes \overline{W}$ imposed at the other, one obtains an RCFT (with vanishing gravitational anomaly) by dimensionally reducing along the interval direction, as shown in \autoref{fig:3dRCFT}. This furnishes a physical interpretation of the description of rational conformal field theories developed starting in \cite{FRS02}.

We have suggestively used the symbol $\mathcal{H}$ to label the topological boundary condition because the different choices of $\mathcal{H}$ precisely label the different modular invariant ways of combining representations $V_i$ of $V$ and $\overline{W}_j$ of $\overline{W}$ into a consistent $S^1$ Hilbert space,
\begin{equation*}
    \mathcal{H} = \bigoplus_{i,j}M_{ij}\,V_i \otimes \overline{W}_{j},
\end{equation*}
where the $M_{ij}$ are non-negative integers with $M_{00}=1$. We represent this RCFT by the triple $(V,W,\mathcal{H})$. Conversely, it is known that any RCFT can be decomposed into a triple $(V,W,\mathcal{H})$, but this decomposition is non-unique. For example, one need not always choose $V$ and $W$ to be the \emph{maximal} left- and right-moving chiral algebras; one could instead choose them to be conformal subalgebras of the maximal left- and right-moving chiral algebras.  When $M_{ij}$ pairs the representations $V_i$ of $V$ with the representations $\overline{W}_j$ of $\overline{W}$ in a one-to-one manner, we say that the RCFT is \emph{diagonal} with respect to the chiral algebras $V$ and $W$.
We note that by ``unfolding'' along the boundary condition, $\mathcal{H}$ may instead be thought of as a topological interface between $\mathcal{T}_V$ and $\mathcal{T}_W$. The RCFT $(V,W,\mathcal{H})$ is diagonal if and only if this topological interface is invertible.

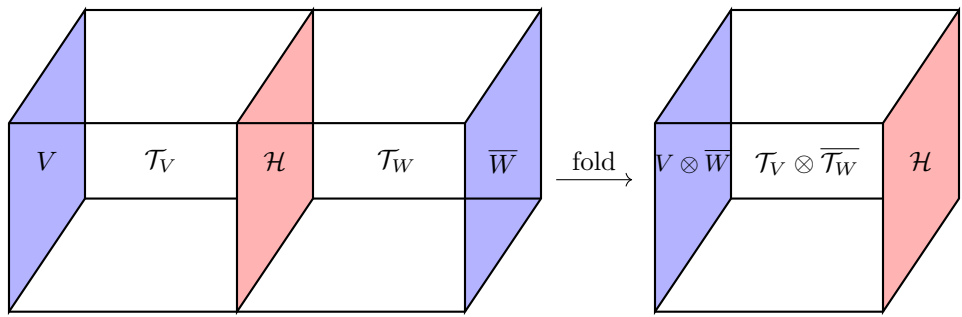
\begin{figure}[ht]
\begin{center}
\begin{tikzpicture}
\filldraw[black,fill=blue!30,  thick] (-6,0) -- (-6,-1-1.5) -- (-5,-1) -- (1-6,1.5) -- cycle;
\filldraw[black,fill=blue!30,  thick] (-6+6,0) -- (-6+6,-1-1.5) -- (-5+6,-1) -- (1-6+6,1.5) -- cycle;
 \draw[black, thick](1,-1)--(-5,-1);
\filldraw[black,fill=red!30,  thick] (-3,0) -- (1-3,1.5) -- (1-3,-1) -- (-3,-1-1.5) -- cycle;
\draw[black,thick] (-6,0) -- (0,0);
\draw[black,thick] (0,-2.5) -- (-6,-1-1.5);
\draw[black,thick] (1,1.5) -- (1-6,1.5);
\node[] at (-.5-.4,-.5) {$\mathcal{T}_{W}$};
\node[] at (-.5-4+.5,-.5) {$\mathcal{T}_V$};
\node[] at (-2.5,-.5) {$\mathcal{\mathcal{H}}$};
\node[] at (-2.5-3,-.5) { $ V$};
\node[] at (-2.5-3+6,-.5) { $\overline{W}$};

\tikzstyle{s}=[draw,decorate,->]
\path[s] (1.2,-.75) to node[anchor=south,shift={(0,0)}] {fold} (2.2,-.75);
\draw[black, thick](1+6-.5,-1)--(-5+6+3-.5,-1);
\filldraw[black,fill=blue!30,  thick] (-3+6-.5,0) -- (1-3+6-.5,1.5) -- (1-3+6-.5,-1) -- (-3+6-.5,-1-1.5) -- cycle;
\draw[black,  thick] (-6+6+3-.5,0)--(+6-.5,0);
\draw[black,  thick] (0+6-.5,-2.5) -- (-6+6+3-.5,-1-1.5);
\draw[black, thick] (1+6-.5,1.5)--(1-6+6+3-.5,1.5);
\node[] at (-.5+6-.5-.5,-.5) {$\mathcal{T}_V\otimes \overline{\mathcal{T}_W}$};
\node[] at (-2.5+6-.5,-.5) {\small $V\otimes \overline{W}$};

\filldraw[black,fill=red!30,  thick] (-3+6-.5+3,0) -- (1-3+6-.5+3,1.5) -- (1-3+6-.5+3,-1) -- (-3+6-.5+3,-1-1.5) -- cycle;
\node[] at (-2.5+6-.5+3,-.5) { $\mathcal{H}$};
\end{tikzpicture}
\end{center}
\caption{Two different 3d representations of the RCFT $(V,W,\mathcal{H})$, which are related by folding.}
\label{fig:3dRCFT}
\end{figure}

Finally, we comment on the physical interpretation of the decomposition in \autoref{prop:asslatdecomp}. Suppose that a \strat{} chiral algebra $V$ has a continuous global symmetry group of rank $r$. One may choose a Cartan subalgebra of the Lie algebra of this symmetry group, and consider the spin-1 Noether currents associated with its generators, which will form a kind of \emph{free sector} of $V$. That is, these Noether currents may be bosonized: their operator algebra is the same as that of $r$ chiral free bosons $\phi^i(z)$. In fact, more is true: for certain vectors ${\lambda}=(\lambda_1,\dots,\lambda_r)$, the vertex operators $V_{{\lambda}}(z)\coloneqq e^{i{\lambda}\cdot \phi(z)}$ will also participate in the chiral algebra $V$. The collection of all such vectors $\lambda$ for which $V_{{\lambda}}(z)$ arises in $V$ define the associated lattice $L$ of $V$, and we learn that $V$ admits a lattice vertex operator subalgebra $V_L\subset V$. The lattice $L$ is part of the charge lattice of $V$, with the full charge lattice being contained inside the dual $L'$.

On the other hand, by considering all of the operators of $V$ which commute with the operators appearing in $V_L$ (i.e.\ the coset of $V$ by its free sector $V_L$), one obtains an \emph{interacting sector} of $V$, which we have been calling the Heisenberg commutant~$C$. The full chiral algebra $V$ is not quite a trivially decoupled tensor product $C\otimes V_L$ of its free and interacting sectors, but rather decomposes into finitely many representations of it, as in \autoref{prop:asslatdecomp}.


\section{Bulk Genus}\label{sec:bulkgenus}

In this section, we explore  a definition of \voa{} genus  due to \cite{Hoe03} which naturally generalizes the description of lattice genera given in part~\ref{item:latgen2} of \autoref{defi:latticegenus}. Mathematically, two \strat{} \voa{}s belong to the same genus if they have the same central charge and their representation categories are ribbon equivalent. Physically, one declares two chiral algebras to belong to the same genus if they arise on the boundary of the same bulk 3d TQFT. We refer to this notion of genus as \emph{bulk genus}. 

After laying out the definition and its basic properties in \autoref{subsec:bulkdefprop}, we describe an effective method for classifying isomorphism classes of strongly rational vertex operator algebras within a fixed bulk genus by means of the gluing principle in \autoref{subsec:glue}. We then discuss and provide evidence in favor of two important conjectures related to bulk genera, namely reconstruction (\autoref{subsec:reconstruction}) and finiteness (\autoref{subsec:finiteness}). We finish with examples in \autoref{subsec:bulkex}.


\subsection{Definition and Properties}\label{subsec:bulkdefprop}

Recall that if $V$ is assumed to be \strat{}, then $\Rep(V)$ can be given the structure of a modular tensor category \cite{Hua08b}. 

\begin{defi}[Bulk Genus]\label{defi:bulkgenus}
Let $V$ and $V'$ be \strat{} \voa{}s. Then $V$ and $V'$ are in the same \emph{bulk genus} if they have the same central charges, $c(V)=c(V')$, and their representation categories are ribbon
equivalent, $\Rep(V)\cong\Rep(V')$.
\end{defi}
\begin{defiph}[Bulk Genus]\label{defiph:bulkgenus}
Two rational chiral algebras $V$ and $V'$ are in the same \emph{bulk genus} if they arise as gapless chiral boundary conditions of the same bulk 3d topological quantum field theory.
\end{defiph}
This defines an equivalence relation, and we denote the bulk genus (or equivalence class) of a \voa{} $V$ by $\bgen(V)$.

Physically, the choice of a ribbon equivalence $\phi\colon\Rep(V)\to\Rep(V')$ corresponds to the choice of an invertible topological interface $\mathcal{I}$ between the 3d TQFTs $(\Rep(V),c(V))$ and $(\Rep(V'),c(V'))$ that support $V$ and $V'$, respectively, on their boundaries.\footnote{A topological interface $\mathcal{I}$ between theories $\mathcal{T}$ and $\mathcal{T}'$ is said to be invertible if there exists another topological interface $\mathcal{I}'$ between $\mathcal{T}'$ and $\mathcal{T}$ such that the fusion of $\mathcal{I}$ with $\mathcal{I}'$ is the trivial surface.} Thus, if $V$ and $V'$ belong to the same bulk genus, then we may form the ``sandwich'' picture in \autoref{fig:bulkboundary}. By squeezing this sandwich, one obtains a \emph{diagonal} 2d RCFT with vanishing gravitational anomaly, and Hilbert space given by
\begin{equation*}
\mathcal{H} = \bigoplus_{M\in\Irr(V)}M \otimes \overline{\phi(M)}.
\end{equation*}
Conversely, it is known \cite{FRS02,KS11} that any diagonal 2d RCFT with vanishing gravitational anomaly can be inflated into a 3d sandwich as in \autoref{fig:bulkboundary}. So, we obtain the following physical characterization of bulk equivalence. 

\begin{figure}[ht]
\begin{center}
\begin{tikzpicture}
\draw[black, thick](2,-1)--(-6,-1);
\draw[black, thick,fill=blue!30] (1,0) -- (2,1.5) -- (2,-1) -- (1,-1-1.5) -- cycle;
\filldraw[black,fill=red!30,  thick] (-3,0) -- (1-3,1.5) -- (1-3,-1) -- (-3,-1-1.5) -- cycle;
\draw[black,  thick,fill=blue!30] (-7,0) -- (1-7,1.5) -- (1-7,-1) -- (-7,-1-1.5) -- cycle;
\draw[black,  thick] (-7,0)--(1,0);
\draw[black,  thick] (1,-2.5) -- (-7,-1-1.5);
\draw[black, thick] (2,1.5)--(1-7,1.5);
\node[] at (1.5,-.5) {$V$};
\node[] at (.5-7,-.5) {$\overline{V'}$};
\node[] at (-.5,-.5) {$(\Rep(V),c(V))$};
\node[] at (-.5-4,-.5) {$(\Rep(V'),c(V'))$};
\node[] at (-2.5,-.5) {$\mathcal{I}$};
\end{tikzpicture}
\end{center}
\caption{Two chiral algebras $V$ and $V'$ belong to the same bulk genus if their bulk TQFTs, $(\Rep(V),c(V))$ and $(\Rep(V'),c(V'))$ respectively,  can be separated by an \emph{invertible} topological interface~$\mathcal{I}$.}
\label{fig:bulkboundary}
\end{figure}
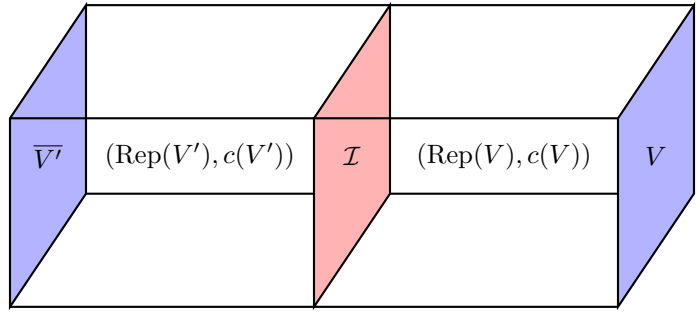

\begin{propph}\label{prop:physicalcharacterizationbulkequivalence}
Two rational chiral algebras $V$ and $V'$ belong to the same bulk genus if and only if there exists a diagonal 2d RCFT with vanishing gravitational anomaly, and with $V$ as its maximal left-moving chiral algebra and $\overline{V'}$ as its maximal right-moving chiral algebra.
\end{propph}

\medskip

In the following, we describe in what sense the above notion of bulk genus generalizes that of lattice genus (for positive-definite, even lattices). Recall that
\begin{equation*}
i\colon L\mapsto V_L
\end{equation*}
defines an injective map from isometry classes of positive-definite, even lattices  to isomorphism classes of \strat{} \voa{}s \cite{Bor86,FLM88}. The image of the map $i$ is exactly the lattice \voa{}s.

Moreover, there are (surjective) projection maps $\gen$ and $\bgen$ passing from a lattice or \voa{}, respectively, to the corresponding genus.

\begin{prop}\label{prop:bulkcomp}
The map $j$ from lattice genera of positive-definite, even lattices to bulk genera of \strat{} \voa{}s defined by
\begin{equation*}
j\colon \gen(L)\mapsto \bgen(V_L)=\bgen(i(L))
\end{equation*}
is well-defined and injective.
\end{prop}
\begin{proof}
To see that the map $j$ is well-defined, suppose that the positive-definite, even lattices $L$ and $M$ are in the same genus so that the discriminant forms $L'/L\cong M'/M$ are isometric and $\rk(L)=\rk(M)$. Then the corresponding lattice \voa{}s $V_L$ and $V_M$ have the same central charges $\rk(L)=\rk(M)$ and their representation categories $\Rep(V_L)\cong\mathcal{C}(L'/L)$ and $\Rep(V_M)\cong\mathcal{C}(M'/M)$ are ribbon equivalent (see \autoref{sec:cat}).

Similarly, two different lattice genera will either have different ranks, resulting under the map $j$ in different central charges, or have different discriminant forms, resulting in different modular tensor categories by the classification of pointed braided fusion categories \cite{JS93,EGNO15}  (see \autoref{sec:cat}). Hence, $j$ is injective.
\end{proof}

By definition, the embeddings $i$ and $j$ are intertwined by the projection maps $\gen$ and $\bgen$ passing from a positive-definite, even lattice or a \strat{} \voa{}, respectively, to the corresponding genus, i.e.\ the diagram
\begin{equation}\label{eq:voagen}
\begin{tikzcd}[ampersand replacement=\&]
\text{lattices}\arrow[twoheadrightarrow]{rr}{\gen}\arrow[hookrightarrow]{dd}[swap]{i} \&\& \text{lattice genera}\arrow[hookrightarrow]{dd}{j}\\\\
\text{VOAs}\arrow[twoheadrightarrow]{rr}{\bgen} \&\& \text{bulk genera}
\end{tikzcd}
\end{equation}
commutes.

\autoref{prop:bulkcomp} and \eqref{eq:voagen} state that each lattice genus defines a unique \voa{} genus and that the lattices in that genus map to \voa{}s in the corresponding \voa{} genus. Crucially, lattices are in the same genus if and only if the corresponding lattice \voa{}s are in the same bulk genus.

These properties, i.e.\ \autoref{prop:bulkcomp} and \eqref{eq:voagen} should be satisfied by any reasonable notion of \voa{} genus that generalizes lattice genera.

\medskip

We can give a different characterization of the map $j$. By definition, the genus of a positive-definite, even lattice can be viewed as a pair $(D,r)$ where $D$ is a metric group (up to isometry) and $r\in\N$. We should note, however, that not all such pairs actually correspond to lattice genera. For instance, we recall from the discussion in \autoref{sec:lat} that it is necessary but not sufficient that $r=\sign(D)\pmod{8}$.

Similarly, the bulk genus of a \strat{} \voa{} can be viewed as a pair $(\mathcal{C},c)$ of a modular tensor category $\mathcal{C}$ (up to ribbon equivalence) and some $c\in\Q$. Again, not all such pairs actually correspond to \voa{} genera. The analog of the signature condition is that it is necessary that $c=c(\mathcal{C})\pmod{4}$, where $c(\mathcal{C})\in\Q/8\Z$ is the chiral central charge of $\Cat$ defined in \autoref{sec:cat} \cite{DLN15}.

It is not difficult to see that the chiral central charge of $\mathcal{C}(D)$ is equal to the signature of the metric group $D$, $c(\mathcal{C}(D))=\sign(D)\in\Z_8$.
It follows that the map $j$ can be written as
\begin{equation*}
j\colon (D,r)\mapsto (\mathcal{C}(D),r)
\end{equation*}
for lattice genera $(D,r)$.

\medskip

It might seem surprising that the ``mod 8'' condition in the case of lattices has become a ``mod 4'' condition for vertex operator algebras. We bridge this gap with the following proposition, which recovers the ``mod 8'' condition in the case of pseudo-unitary representation categories, such as the pointed ones $\Cat(D)$. (Recall that, for us, pseudo-unitarity includes the unique choice of spherical pivotal structure that makes the categorical dimensions all positive. Moreover, recall that for the modular tensor category $\Rep(V)$, the ribbon or pivotal structure is fixed by the construction in \cite{Hua08b}.)

\begin{prop}\label{prop:pseudounitarymod8}
    If $V$ is a strongly rational vertex operator algebra with $\Rep(V)$ a pseudo-unitary modular tensor category, then $c(V)=c(\Rep(V))\pmod{8}$.
\end{prop}
In particular, in the lattice case, since $\Rep(V_L)\cong\Cat(L'/L)$ is a pseudo-unitary pointed modular tensor category, we recover for it the ``mod 8'' condition obeyed by lattices.

To facilitate the proof of this proposition, we require a technical lemma. Let $S$ and $T$ be the genus-1 modular data of $V$ \cite{Zhu96}, i.e.\ the matrices satisfying 
\begin{equation}\label{eqn:modulartransformationcharacters}
\mathrm{ch}_{V^i}(-1/\tau) = \sum_j S_{ij}\mathrm{ch}_{V^j}(\tau),\quad\mathrm{ch}_{V^i}(\tau+1)=\sum_j T_{ij}\mathrm{ch}_{V^j}(\tau),
\end{equation}
where $\mathrm{ch}_{V^i}(\tau)$ is the $q$-character of the irreducible module $V^i$ of $V$, defined as 
\begin{equation*}
\mathrm{ch}_{V^i}(\tau) = \tr_{V^i}q^{L_0-c/24}.
\end{equation*}
Note that $(ST)^3=C$, where $C_{ij}=\delta_{ij^\ast}$ is the charge conjugation matrix and $j^\ast$ indexes the object dual to the one indexed by $j$ in $\mathcal{C}$. We recall from \cite{DLN15} that $S$ and $T$ are related to the twist matrix $\tilde{t}$ and the normalized $S$-matrix $s$ of $\Rep(V)$ defined in \autoref{sec:cat} as
\begin{equation*}
S=\pm s\quad\text{and}\quad T= e^{-2\pi\i c(V)/24}\tilde{t}.
\end{equation*}

\begin{lem}\label{lem:signSmatrix}
Let $V$ be a strongly rational vertex operator algebra, with $\Rep(V)$ not necessarily pseudo-unitary. Then $S=+s$ if and only if $c(V) = c(\Rep(V))\pmod{8}$, and $S=-s$ if and only if $c(V)+4=c(\Rep(V))\pmod{8}$.
\end{lem}

\begin{proof}
Let $S=\epsilon s$, where $\epsilon=\pm 1$. Following the proof of Theorem 3.10 of \cite{DLN15}, we calculate that 
\begin{equation*}
C=(ST)^3=\epsilon (s\tilde{t}e^{-2\pi\i c/24})^3 = \epsilon \frac{p^+}{D} e^{-2\pi \i c/8} C=\epsilon e^{\frac{2\pi\i}{8}(c(\Rep(V))-c(V))}C,
\end{equation*}
where we have used \eqref{eqn:modularrelations} and \eqref{eqn:chiralcentralcharge}. The lemma immediately follows.
\end{proof}

We are now ready to give the proof of \autoref{prop:pseudounitarymod8}.

\begin{proof}[Proof of \autoref{prop:pseudounitarymod8}]
If $\Rep(V)$ is pseudo-unitary, then the entries $s_{0i}=d_i/D$ of the normalized $S$-matrix are all positive numbers, where the $d_i>0$ are the categorical dimensions of the simple objects. If $c(V)\neq c(\Rep(V))\pmod{8}$, it must be the case that $c(V)+4=c(\Rep(V))\pmod{8}$, in which case \autoref{lem:signSmatrix} says that $S=-s$. Then
\begin{equation*}
    \mathrm{ch}_V(-1/\tau) = -\sum_{j}s_{0j}\mathrm{ch}_{V^j}(\tau).
\end{equation*}
If we evaluate this equation at $\tau=i$, we find that the left-hand side is a positive number, while the right-hand side is negative, and thus reach a contradiction. It follows that the ``mod 4'' condition can be strengthened to a ``mod 8'' condition.
\end{proof}

\medskip

We pause to make the following  observation:
clearly, there are bulk genera that do not come from lattice genera, i.e.\ that are not in the image of the map $j$. For example, when a bulk genus $(\mathcal{C},c)$ involves a modular tensor category $\mathcal{C}$ that is not pointed, then it cannot support any lattice vertex operator algebras. The simplest example is $(\mathcal{C},c) = \bigl((G_2,1),\sfrac{14}5\bigr)$, which contains a unique \voa{}, the simple affine \voa{} $\mathsf{G}_{2,1}$. Perhaps more surprisingly, there are even pointed bulk genera, i.e.\ bulk genera  of the form $(\mathcal{C}(D),r)$, that nonetheless do not contain any lattice vertex operator algebras, and hence are also not in the image of $j$. We provide an example below.
\begin{ex}[Commutant B]\label{ex:commB}
Consider the isometry $\nu\in\Aut(\Lambda)$ of the Leech lattice $\Lambda$ of Frame shape $1^82^8$. The coinvariant lattice 
\begin{equation*}
\Lambda_\nu=(\Lambda^\nu)^\perp = \{ \lambda \in \Lambda \mid \nu(\lambda)=\lambda\}^\perp
\end{equation*}
is in the lattice genus $(2_{\II}^{+8},8)$, which is an edge case in the sense of \cite{Nik80}, meaning that the lattice rank is minimal for the given discriminant form. The isometry $\nu$ now acts fixed-point freely on $\Lambda_\nu$ and all lifts $\hat\nu$ to an automorphism of the corresponding lattice \voa{} $V_{\Lambda_\nu}$ have order~$2$ and are conjugate.
The \fpvosa{} $V_{\Lambda_\nu}^{\hat\nu}$ has a pointed representation category and is in the bulk genus $(\mathcal{C}(2_{\II}^{+10}),8)$ \cite{Moe16}. By \cite{Nik80}, there can be no lattice \voa{} in that genus; the rank would have to be at least $10$.
The \voa{} $V_{\Lambda_\nu}^{\hat\nu}$ appears prominently as the Heisenberg commutant of the hyperbolic genus B (see \autoref{ex:schellekens2} and \autoref{sec:holhyp}).
\end{ex}

Furthermore, a \voa{} genus that is in the image of $j$, and hence must contain at least one lattice \voa{}, may also contain non-lattice \voa{}s. For example, the moonshine module $V^\natural$ belongs to the bulk genus $(\mathcal{C},c) = (\Vect,24)$; in fact, most of the \voa{}s in this bulk genus are non-lattice, as we review in \autoref{sec:holo}. 


\subsection{Gluing Principle and Classification}\label{subsec:glue}

The following method, emphasized in \cite{SR23} (see also \cite{GHM16}), is useful for classifying \voa{}s in an arbitrary bulk genus $(\mathcal{C},c)$.

Recall that we denote the irreducible modules $\Irr(V)$ of an arbitrary strongly rational vertex operator algebra $V$ as $V^i$, with $V^0\coloneqq V$, and use $\rho(V^i)$ to denote the lowest eigenvalue of $L_0$ acting on $V^i$. Recall further that $\overline{\mathcal{C}}$ is the ribbon-reversed modular tensor category, as defined in \autoref{sec:cat}.

Given a bulk genus $(\mathcal{C},c)$, fix a strongly rational vertex operator algebra $\tilde{V}$ in a bulk genus of the form $(\overline{\mathcal{C}},\tilde{c})$ and a ribbon-reversing equivalence $\phi\colon\Rep(\tilde{V})\to\mathcal{C}$.
\begin{prop}[Gluing Principle]\label{prop:gluingprinciple}
Any strongly rational vertex operator algebra $V$ in the bulk genus $(\mathcal{C},c)$ which satisfies $\rho(V^{\phi(i)})>-\rho(\tilde{V}^i)$ for all $i\neq 0$ can be expressed as $V\cong \Com_{W}(\iota\tilde{V})$ for some \strat{}, holomorphic \voa{} $W$ of central charge $c+\tilde{c}$ and some primitive embedding $\iota\colon\tilde{V}\to W$. 
\end{prop}

\begin{rem}
In particular, if $\tilde{V}$ is positive, then any positive $V$ in the bulk genus $(\mathcal{C},c)$ can be obtained as $V\cong \Com_W(\iota\tilde{V})$ for some holomorphic $W$ of central charge $c+\tilde{c}$ and some primitive embedding $\iota\colon\tilde{V}\to W$.
\end{rem}
\begin{proof}
Any $V$ in the bulk genus $(\mathcal{C},c)$ satisfying $\rho(V^{\phi(i)})>-\rho(\tilde{V}^i)$ can be ``glued'' to $\tilde{V}$ using the ribbon-reversing equivalence $\phi$ to produce a \strat{}, holomorphic \voa{} $W$ of central charge $c+\tilde{c}$,
\begin{equation*}
W\coloneqq\bigoplus_{i\in\Irr(\tilde{V})} V^{\phi(i)}\otimes \tilde{V}^i
\end{equation*}
\cite{CKM22,Lin17}. The condition on the conformal weights is necessary to ensure that $W$ is of CFT-type. It follows from the construction that $\tilde{V}$ embeds into $W$ primitively via some embedding~$\iota$ and that $V\cong \Com_W(\iota\tilde{V})$. 
\end{proof}

Thus, to enumerate (say, positive) \voa{}s in a genus $(\mathcal{C},c)$, the gluing principle tells us that all one needs is
\begin{enumerate}
\item a single positive, \strat{} \voa{} $\tilde{V}$ in a genus of the form $(\overline{\mathcal{C}},\tilde{c})$ (preferably with $\tilde{c}$ as small as possible),
\item the list of all \strat{}, holomorphic \voa{}s of central charge $c+\tilde{c}$ that support $\tilde{V}$ as a primitive subalgebra and
\item all possible primitive embeddings of $\tilde{V}$ into those holomorphic \voa{}s. 
\end{enumerate}
It is usually not too difficult to find a ``seed'' \voa{} $\tilde{V}$; for example, if $\overline{\mathcal{C}}$ is a Chern--Simons type modular tensor category, then one can take $\tilde{V}$ to be a simple affine \voa{} (or current algebra), see \autoref{ex:affine}. As one can glean from the discussion in \autoref{sec:holo}, item~(2) is also tractable when $c+\tilde{c}\leq 24$, because \strat{}, holomorphic vertex operator algebras are classified through central charge $24$, up to \autoref{conj:moonshineuniqueness} (see \autoref{sec:holo}). Finally, in fortuitous circumstances, item (3) is also achievable, e.g., when $\tilde{V}$ is an affine \voa{}, because in this case, determining embeddings of $\tilde{V}$ reduces to a problem of determining embeddings of ordinary Lie algebras, for which an abundance of tools exist.

\medskip

The gluing principle is most effective at enumerating positive vertex operator algebras. This is sometimes good enough, because one can often exclude the existence of vertex operator algebras violating positivity. The intuition that a physicist might want to appeal to is that if $\mathcal{C}$ is a \emph{unitary} modular tensor category, then any (\strat{}) vertex operator algebras it supports must also be unitary, and hence be positive. However, to the best of our knowledge, such a result has not yet been obtained; in fact, it is not even clear that it should be true, though no counter-examples are known. Instead, the following is an easily-checked sufficient condition on $\mathcal{C}$ for the absence of strongly rational vertex operator algebras $V$ with $\Rep(V)\cong \mathcal{C}$ that violate positivity. 

Consider a modular tensor category $\mathcal{C}$ with simple objects indexed by $I=\{0,\dots,\rk(\mathcal{C})-1\}$ and with normalized $S$-matrix $s$. Let $J\subset I$ be a set different
from $\{0\}$ satisfying (1) $\theta_j=\theta_{j'}$ for $j,j'\in J$ and (2) $j^\ast\in J$ if $j\in J$, where $j^\ast$ is the unique index for which $(s^2)_{jj^\ast}=1$ (i.e.\ the index of the dual object). We call a set $J$ that satisfies these conditions a \emph{valid indexing set}, and say that a modular tensor category is \emph{essentially positive} if its normalized $S$-matrix $s$ has the property that, for any valid indexing set $J$, there exists an $i\in I$ such that $s_{ij}+s_{ij^\ast}<0$ for all $j\in J$. (Note that $s_{ij}+s_{ij^\ast}$ is a real number because $s_{ij^\ast}=\overline{s_{ij}}$.)

\begin{prop}\label{prop:essentiallypositive}
If $V$ is a strongly rational vertex operator algebra in a bulk genus $(\mathcal{C},c)$ with $\mathcal{C}$ pseudo-unitary and  essentially positive, then $V$ is positive.
\end{prop}

\begin{proof}
Let $V$ be a strongly rational vertex operator algebra in the essentially positive bulk genus $(\mathcal{C},c)$ with irreducible modules $V^i$, $i\in I$, that have conformal weights $h_i\coloneqq\rho(V^i)\in\Q$. We study the character vector
\begin{equation*}
\mathrm{ch}_i(\tau) = \tr_{V^i}q^{L_0-c/24} = \sum_{h\in h_i+\N} \dim(V^i_h) q^{h-c/24},
\end{equation*}
$i\in I$. By combining \autoref{lem:signSmatrix} and \autoref{prop:pseudounitarymod8}, we learn that the genus-1 $S$-matrix $S$ of $V$ defined in \eqref{eqn:modulartransformationcharacters} is precisely equal to the normalized $S$-matrix $s$ of $\mathcal{C}\cong\Rep(V)$ (as opposed to $S=-s$).

Recall that the effective central charge (see, e.g., \cite{GN03}) of $V$ is defined as $\tilde{c}=c-24h_\text{min}\geq c$ where $h_\text{min}=\min\{h_i\mid i\in J\}\leq 0$. Suppose, for the sake of contradiction, that $V$ is not positive. Then $J=\{j\in I\mid c-24h_j=\tilde{c}\}=\{j\in I\mid h_j=h_\text{min}\}\neq \{0\}$ is a valid indexing set. By the assumption of essential positivity, there exists an $i\in I$ such that $S_{ij}+S_{ij^\ast}<0$ for all $j\in J$. A version of the Cardy formula \cite{Car86} then asserts, as a consequence of the modular transformation properties of the $q$-characters of $V$, that the asymptotic density of states in $V^i$ grows exponentially for large $h$ as 
\begin{equation*}
\dim(V^i_h)\sim \sum_{j\in J} \overline{S_{ij}}\dim(V^j_{h_j})\exp\left(2\pi \sqrt{\frac{\tilde{c}h}{6}}\right).
\end{equation*}
Using the fact that $j^\ast\in J$ if $j\in J$, we can rewrite this as
\begin{align*}
\dim(V^i_h)&\sim \sum_{j\in J} \frac{\overline{S_{ij}}\dim(V^j_{h_j})+\overline{S_{ij^\ast}}\dim(V^{j^\ast}_{h_{j^\ast}})}{2}\exp\left(2\pi \sqrt{\frac{\tilde{c}h}{6}}\right) \\
&=\sum_{j\in J}\frac{S_{ij}+S_{ij^\ast}}{2}\dim(V^j_{h_j})\exp\left(2\pi \sqrt{\frac{\tilde{c}h}{6}}\right)<0
\end{align*}
where we use that $\dim(V^j_{h_j})=\dim(V^{j^*}_{h_{j^*}})$ to go from the first line to the second line, and subsequently that $S_{ij}+S_{ij^\ast}<0$ by the assumption of essential positivity to deduce that the sum is negative. This contradicts the fact that $\dim(V^i_h)$ is the dimension of a vector space. Hence, $V$ is positive.
\end{proof}


\subsection{Reconstruction}\label{subsec:reconstruction}

Before moving on, we record two important conjectures in the theory of \strat{} vertex operator algebras, in this subsection and the next. An important question is to decide if a given pair $(\mathcal{C},c)$ of a modular tensor category $\mathcal{C}$ and $c\in\Q$, satisfying the obvious constraints discussed in \autoref{subsec:bulkdefprop}, actually corresponds to a \voa{} genus. This would be the analog of the result in \cite{Nik80} for metric groups and lattices, which we reviewed in \autoref{sec:lat}.

\begin{conj}[Reconstruction]\label{conj:reconstruction}
   For every modular tensor category $\mathcal{C}$, there is a sufficiently large $c\in\Q$ for which the bulk genus $(\mathcal{C},c)$ exists. That is, every modular tensor category is of the form $\mathrm{Rep}(V)$ for some \strat{} vertex operator algebra $V$. 
\end{conj}
\begin{rem}
The claim that every symbol $(\mathcal{C},c)$ with $c=c(\mathcal{C})\pmod{4}$ corresponds to a bulk vertex operator algebra genus is certainly false. For example, consider the modular tensor category $\Cat\coloneqq(A_1,5)_{\sfrac12}$ generated by objects in $(A_1,5)$ labeled by an integer $\mathrm{SU}(2)$ spin.
In \cite{Ray23}, it was found that there is no positive, \strat{} \voa{} of central charge $\sfrac{8}{7}$ and representation category $\overline{\Cat}$. Furthermore, one can check that $\overline\Cat$ is an essentially positive modular tensor category, so that there are no vertex operator algebras with representation category $\overline{\Cat}$ that violate positivity by \autoref{prop:essentiallypositive}. Thus, the symbol $(\overline{(A_1,5)_{\sfrac12}},\sfrac{8}{7})$ does not correspond to a bulk vertex operator algebra genus. This, however, does not contradict \autoref{conj:reconstruction} as the genus $(\overline{(A_1,5)_{\sfrac12}},\sfrac{8}7+8n)$ exists whenever $n\in\Ns$. For example, it turns out that the unique simple-current extension of $\mathsf{A}_{1,5}\mathsf{E}_{7,1}$ has central charge $64/7$ and representation category $\overline{(A_1,5)_{\sfrac12}}$.
\end{rem}
As a potentially easier first target, it would even be interesting to prove the following variation of \autoref{conj:reconstruction}. 

\begin{conj}\label{conj:weakreconstruction}
    If there is a positive strongly rational vertex operator algebra $V$ with $\Rep(V)\cong \mathcal{C}$, then there is a positive strongly rational vertex operator algebra $\tilde{V}$ with $\Rep(\tilde{V})\cong \overline{\mathcal{C}}$.
\end{conj}
In particular, establishing this variation would show that the choice of a ``seed'' \voa{} $\tilde{V}$ to use as input to the gluing principle always exists. One then obtains the following corollary of \autoref{conj:weakreconstruction}.
\begin{cor}\label{cor:holemb}
    Assuming \autoref{conj:weakreconstruction}, it follows that every positive \strat{} \voa{} can be embedded primitively into a \strat{}, holomorphic \voa{}, from which it can be obtained as a commutant (or coset). 
\end{cor}
\autoref{conj:reconstruction} can be proved for certain classes of modular tensor categories $\mathcal{C}$. For example, the quantum group categories $(X_n,k)$ are realized by the simple affine \voa{}s $\mathsf{X}_{n,k}$. (Physically, a Chern--Simons theory always admits a chiral Wess--Zumino--Witten boundary condition.) Moreover, when $(X_n,k)$ admits no non-trivial condensable algebras (see, e.g., \cite{Kon14} for the definition of condensable algebra), $\overline{(X_n,k)}$ can also be shown to be saturated by a \strat{} \voa{} (see, e.g., the paragraph before Section~3 of \cite{SY90}).\footnote{The proof is essentially to show that the current algebra $\mathsf{X}_{n,k}$ can always be embedded into a holomorphic vertex operator algebra $V$. The assumption that $(X_n,k)$ admits no non-trivial condensable algebras implies that $\mathsf{X}_{n,k}$ admits no non-trivial conformal extensions, and hence that any embedding of $\mathsf{X}_{n,k}$ is primitive. Thus, the coset $V/\mathsf{X}_{n,k}$ is guaranteed to have representation category $\overline{(X_n,k)}$.} Further, any pointed modular tensor category is the representation category of some lattice vertex operator algebra. (Physically, an abelian Chern--Simons theory can support free chiral bosons on its boundary.) Finally, the twisted Drinfeld double of any finite group is known to satisfy reconstruction \cite{EG22} if one assumes the widely-believed conjecture that the fixed-point subalgebra of a \strat{} \voa{} with respect to a finite subgroup $G$ of its automorphism group is again \strat{} (see \cite{CM16} for the proof in the case that $G$ is solvable).


\subsection{Finiteness}\label{subsec:finiteness}

All the data collected to date, some of which we review in the next subsection in examples, also suggests the following conjecture.
\begin{conj}[Bulk Genera Finiteness]\label{conj:bulkfiniteness}
Every bulk genus contains only finitely many isomorphism classes of \voa{}s.
\end{conj}
Using the gluing principle,  it is actually possible to show that an infinite class of bulk genera admit only finitely many positive vertex operator algebras, if one assumes that there are only finitely many holomorphic vertex operator algebras of a fixed central charge.
\begin{prop}\label{prop:partialfiniteness}
    Let $\mathcal{C}$ be the ribbon reverse of a modular tensor category of the form $(X_{r_1}^{(1)},k_1)\boxtimes\dots\boxtimes(X_{r_n}^{(n)},k_n)$. Assuming that there are only finitely many \strat{}, holomorphic \voa{}s in each central charge, it follows that any bulk genus of the form $(\mathcal{C},c)$ contains finitely many positive strongly rational vertex operator algebras.
\end{prop}

\begin{proof}
Because $\overline{\mathcal{C}}$ is a product of Chern--Simons modular tensor categories, we can obtain a  positive seed \voa{} $\tilde{V}$ with $\Rep(\tilde{V})\cong \overline{\mathcal{C}}$ to use as input into the gluing principle by taking $\tilde{V} = \mathsf{X}_{r_1,k_1}^{(1)}\dots\mathsf{X}_{r_n,k_n}^{(n)}.$ Call $\tilde{c}$ the central charge of $\tilde{V}$.

By \autoref{prop:gluingprinciple}, every positive \voa{} $V$ in the genus $(\mathcal{C},c)$ is of the form $\Com_{W}(\iota\tilde{V})$ for some \strat{}, holomorphic \voa{} $W$ of central charge $c+\tilde{c}$, and some primitive embedding $\iota\colon\tilde{V}\hookrightarrow W$. By assumption, the number of such $W$ is finite. On the other hand, we claim the number of primitive embeddings $\iota$ is finite as well. To see this, note that $\iota$ is completely determined by where it sends the weight-1 subspace $\tilde{V}_1$, because $\tilde{V}_1$ strongly generates $\tilde{V}$. Now, $\iota(\tilde{V}_1)\subset W_1$ defines a semi-simple subalgebra of a reductive Lie algebra, of which it is known there are only finitely many up to inner automorphisms in $W$.

Thus, we learn that there are only finitely many ways to form the commutants which enter the gluing principle, and therefore that the number of positive \strat{} \voa{}s in the genus $(\mathcal{C},c)$ is finite.
\end{proof}

\autoref{prop:partialfiniteness} is conditioned on the finiteness of holomorphic \voa{}s, so it is natural to ask how reasonable this assumption is. As we shall see below and in \autoref{sec:holo}, it is known that $(\Vect,c)$ is finite when $c\leq 16$, but not known when $c\geq 24$, as there could be infinitely many ``fake'' moonshine modules (cf.\ \autoref{ex:fake}).


\subsection{Examples}\label{subsec:bulkex}

In general, it is a very difficult problem to classify all isomorphism classes of \strat{} \voa{}s $V$ in a given bulk genus. However, there are some situations where the problem becomes tractable.

The following example is a natural starting point (in particular, in view of \autoref{cor:holemb}) and shall be described in greater detail in \autoref{subsec:holobulk}.
\begin{ex}[Holomorphic \VOA{}s]\label{ex:schellekens}
Recall that a strongly rational vertex operator algebra $V$ is holomorphic if $\Rep(V)\cong\Vect$, i.e.\ if the unique irreducible $V$-module is $V$ itself (so that $V$ is in particular positive). \Strat{}, holomorphic \voa{}s only exist for central charge $c\in8\N$ \cite{Zhu96}. The corresponding bulk genera have been determined through central charge $c\leq 24$, up to an open problem about the uniqueness of the moonshine module,
{\allowdisplaybreaks
\begin{align*}
(\Vect,0)&=\{\C\vac\},\\
(\Vect,8)&=\{V_{E_8}\},\\
(\Vect,16)&=\{V_{E_8^2},V_{D_{16}^+}\},\\
(\Vect,24)&=\{71\text{ Schellekens VOAs}\}.
\end{align*}
}%
At central charge $24$, known as the Schellekens \voa{}s \cite{Sch93}, there are $70$ \voa{}s $V$ with $V_1\neq\{0\}$
and additionally the moonshine module (or monster \voa{}) $V^\natural$ with $V^\natural_1=\{0\}$. While the classification for $V_1\neq\{0\}$ has recently been rigorously established in the mathematical literature, the classification for $V_1=\{0\}$, the moonshine uniqueness problem, is still open (see \autoref{conj:moonshineuniqueness}).

Of these $71$ holomorphic \voa{}s, only $24$ correspond to lattices, the so-called Niemeier lattices. Hence, the genus $(\Vect,24)$, and in fact all genera $(\Vect,c)$ with $c\geq24$ contain \voa{}s that are not lattice \voa{}s. By contrast, the bulk genera for $c=0,8,16$ only contain lattice \voa{}s.

The landscape of holomorphic \voa{}s becomes quite unwieldy at central charges beyond $24$.
\end{ex}

\begin{ex}[Vertex Operator Algebras with Two Irreducible Modules]
The previous example dealt with bulk genera which support \voa{}s with only a single irreducible module. Here, we summarize analogous known results about vertex operator algebras with exactly two irreducible modules, i.e.\ about bulk genera $(\mathcal{C},c)$ with  $\rk(\mathcal{C})=2$. 

For simplicity, we assume that $\mathcal{C}$ is a \emph{unitary} (and in particular pseudo-unitary) modular tensor category. Then, by the classification of modular tensor categories with low rank \cite{RSW09}, it follows that
\begin{equation*}
    \mathcal{C}\in\left\{ (A_1,1), \ (E_7,1), \ (G_2,1),\ (F_4,1)\right\}.
\end{equation*}
By studying the $S$-matrices of these categories, one can easily verify that they are all essentially positive. For example, for $(A_1,1)$, the normalized $S$-matrix is 
\begin{equation*}
    s=\frac{1}{\sqrt{2}}\left(\begin{array}{cc} 1 & 1 \\ 1 & -1\end{array}\right).
\end{equation*}
The only valid indexing set is $J=\{1\}$, and $s_{1,1}=-1/\sqrt{2}<0$. Hence, by \autoref{prop:essentiallypositive}, all strongly rational vertex operator algebras $V$ with $\Rep(V)$ equivalent to one of these categories must necessarily be positive.

Using \autoref{prop:pseudounitarymod8}, we learn that vertex operator algebras $V$ with $\Rep(V)\cong \mathcal{C}$ taken from the above list must have central charge $c$ congruent to $1$, $7$, $14/5$ or $26/5$ modulo $8$, respectively. On the other hand, the absence of vector-valued cusp forms\footnote{A vector-valued modular form for $\mathrm{SL}_2(\Z)$ is a cusp form if the $q$-expansions of its components involve only positive powers of $q$.} transforming with respect to the modular data of these categories rules out the existence of vertex operator algebras with negative central charge. So let $(\mathcal{C},c)$ be the bulk genus $\left((A_1,1),1\right)$ for concreteness. Identical arguments work for the bulk genera $\left((E_7,1),7\right)$, $\left((G_2,1),14/5\right)$ and $\left((F_4,1),26/5\right)$.

Considerations involving vector-valued modular forms \cite{SR23,Ray23} reveal that any two positive, \strat{} vertex operator algebras in the bulk genus $\left((A_1,1),1\right)$ have the same character. In particular, any such \voa{} has character vector equal to that of $\mathsf{A}_{1,1}$. On the other hand, \cite{MNS21} shows that $\mathsf{A}_{1,1}$ is the unique strongly rational vertex operator algebra with character vector equal to $\mathsf{A}_{1,1}$. Thus, $\mathsf{A}_{1,1}$ is the unique positive \voa{} in the bulk genus $\left((A_1,1),1\right)$, and hence, by the essential positivity of the $(A_1,1)$ category, it is the only \voa{} in this bulk genus. Similar comments apply to the theories $\mathsf{E}_{7,1}$, $\mathsf{G}_{2,1}$, and $\mathsf{F}_{4,1}$. These chiral algebras are often called Mathur--Mukhi--Sen theories after the authors of \cite{MMS88}, who showed that their character vectors are distinguished by considerations in the theory of modular  differential equations. 

By deploying the gluing principle, using the Mathur--Mukhi--Sen theories as the ``seed'' vertex operator algebras $\tilde{V}$ and noting that $\overline{(A_1,1)}\cong (E_7,1)$ and $\overline{(G_2,1)}\cong (F_4,1)$, one can obtain classification results for bulk genera of vertex operator algebras with two irreducible modules and higher central charges. In particular, \cite{SR23} classified all positive vertex operator algebras $V$ in bulk genera $(\mathcal{C},c)$ with $0\leq c<25$ and $\mathcal{C}$ unitary with $\rk(\mathcal{C})=2$. We emphasize that, by using \autoref{prop:essentiallypositive}, the result of \cite{SR23} remains valid even after relaxing the positivity assumption on $V$.
\end{ex}

\begin{ex}[Bulk Genera for $2_{\II}^{+2}$]\label{ex:bulk4}
The positive, \strat{} \voa{}s $V$ in the bulk genera $(\Cat(2_{\II}^{+2}),c)$ with $c\leq24$ were determined in \cite{HM23}, assuming \autoref{conj:moonshineuniqueness}. Here, $2_{\II}^{+2}$ denotes the metric group given by the group $\Z_2\times\Z_2$ with the quadratic form taking values $0$, $0$, $0$ and $1/2$. Physically, $\mathcal{C}(2_{\II}^{+2})$ corresponds to the toric code topological order, and it is also the (untwisted) Drinfeld double $\mathcal{D}(\Z_2)$. These bulk genera only exist for $c\in8\Ns$, corresponding to the signature of $2_{\II}^{+2}$. Such \voa{}s are interesting because they appear as even parts of \strat{}, holomorphic \svoa{}s. For example, there are exactly $969$ positive, \strat{} \voa{}s in the bulk genus $(\Cat(2_{\II}^{+2}),24)$ assuming \autoref{conj:moonshineuniqueness}.

One can show that the results of \cite{HM23} continue to hold even after relaxing the assumption that $V$ is positive. The category $\mathcal{C}(2_{\II}^{+2})$ is unfortunately not essentially positive, because there is no $i\in I\coloneqq\{0,\dots,3\}$ such that $s_{ij}^\ast+s_{ij^\ast}^\ast<0$ for $j$ in the valid indexing sets $J=\{0,2\}$ and $J=\{0,3\}$. However, using logic identical to that in the proof of \autoref{prop:essentiallypositive}, one can argue that the only way that $V$ could violate positivity in this bulk genus is if $h_0=h_2=0$ with $h_1,h_3>0$, \emph{or} if $h_0=h_3=0$ with $h_1,h_2>0$. On the other hand, if one uses the data in \cite{Ray23} to compute the possible vector-valued modular forms $\{f_i(\tau)\}_{i\in I}$ that could serve as the character vector of such a \voa{}, then one finds that each $\{f_i(\tau)\}_{i\in I}$ either has negative coefficients or a component which identically vanishes, which is a contradiction. 
Thus, all \strat{} \voa{}s in the bulk genus $(\mathcal{C}(2_{\II}^{+2}),24)$ are necessarily positive.

By splitting off free fermions, one can deduce from this result also a classification of the \strat{} \voa{}s in the bulk genera $(\sfrac{47}{2},(B_7,1))$, $(23,\Cat(4_7^{+1}))$, $(\sfrac{45}{2},(B_5,1))$, etc. For $c\leq 23$ these results are unconditional.
\end{ex}

Similar classification results to the previous few examples have been derived using the gluing principle for most bulk genera $(\mathcal{C},c)$ with $c\leq 24$ and $\mathcal{C}$ of rank (i.e.\ number of simple objects) at most $4$ and pseudo-unitary (i.e.\ with positive categorical dimensions) \cite{HM23,SR23,Ray23}. See especially Table~1 of \cite{Ray23} for a glossary of these results. Again, we emphasize that all such $\mathcal{C}$ (except for $\mathcal{C}(2_{\II}^{+2})$, $(G_2,1)\boxtimes (F_4,1)$, and $(A_1,1)\boxtimes (E_7,1)$, which can be treated separately as in the previous example) are essentially positive, and hence the positivity assumption on \voa{}s can be relaxed in all of these results.

As in the holomorphic case, the number of \strat{} \voa{}s in a genus $(\mathcal{C},c)$ with fixed $\mathcal{C}$ is generally expected to grow quite rapidly with the central charge $c$, such that explicit enumeration of \voa{}s becomes quickly impractical if the central charge is taken to be too large. We will make this expectation sharper when we come to the mass formula in \autoref{sec:hypgenus}.


\section{Hyperbolic Genus}\label{sec:hypgenus}

In this section, we study a second definition  of \voa{} genus due to \cite{Mor21} which naturally generalizes part~\ref{item:latgen3} of \autoref{defi:latticegenus}, and which we call \emph{hyperbolic genus}. Two \voa{}s belong to the same hyperbolic genus if they become isomorphic after tensoring with $V_{\II_{1,1}}$, the lattice vertex algebra associated with the unique, even unimodular lattice $\II_{1,1}$ of signature $(1,1)$. We derive an alternative (and arguably more tractable) characterization of the hyperbolic genus of a \voa{} (see \autoref{thm:hypcomm}) which appeals to its Heisenberg commutant and associated lattice (or, to a physicist, its free and interacting sectors). We use this alternative characterization to show in \autoref{cor:hypbulk} that the hyperbolic genus is a finer equivalence relation than the bulk genus, and to study a mass formula for \voa{}s.


\subsection{Definition and Properties}

Recall that $\II_{1,1}$ denotes the unique even, unimodular lattice of signature $(1,1)$. The corresponding lattice vertex algebra $V_{\II_{1,1}}$ \cite{Bor86,FLM88} (see \autoref{ex:lattice}) is not a \voa{} but rather a conformal vertex algebra, as its $L_0$-weight spaces are infinite-dimensional and the weights are not bounded from below.

The following definition makes sense not just for \voa{}s but for any conformal vertex algebra.
\begin{defi}[Hyperbolic Genus, Preliminary]\label{defi:hypgenprel}
Let $V$ and $V'$ be two conformal vertex algebras. Then $V$ and $V'$ are in the same \emph{hyperbolic genus} if $V\otimes V_{\II_{1,1}}\cong V'\otimes V_{\II_{1,1}}$ as conformal vertex algebras.
\end{defi}
This clearly defines an equivalence relation, and we denote the hyperbolic genus (or equivalence class) of a conformal vertex algebra $V$ by $\hgen(V)$.

Again, we verify that this definition is a reasonable generalization of the notion of lattice genus (cf.\ \autoref{prop:bulkcomp}).
\begin{prop}\label{prop:hypercomp}
The map $j$ from lattice genera of even lattices to hyperbolic genera of conformal vertex algebras defined by
\begin{equation*}
j\colon \gen(L)\mapsto \hgen(V_L)=\hgen(i(L))
\end{equation*}
is well-defined and injective.
\end{prop}
Compared to \autoref{prop:bulkcomp} for the bulk genus, we dropped the assumption that the lattices are positive-definite and correspondingly that the conformal vertex algebras are (\strat{}) \voa{}s, i.e.\ we extended the domain and codomain of the map $j$.
\begin{proof}
Suppose $L$ and $M$ are even lattices in the same genus, i.e.\ $L\oplus\II_{1,1}\cong M\oplus\II_{1,1}$. Then $V_L\otimes V_{\II_{1,1}}\cong V_{L\oplus\II_{1,1}}\cong V_{M\oplus\II_{1,1}}\cong V_M\otimes V_{\II_{1,1}}$ so that $V_L$ and $V_M$ are in the same hyperbolic genus. Hence, $j$ is well-defined.

Now, suppose that the lattice vertex algebras $V_L$ and $V_M$ are in the same hyperbolic genus, i.e.\ $V_L\otimes V_{\II_{1,1}}\cong V_M\otimes V_{\II_{1,1}}$. Then $V_{L\oplus\II_{1,1}}\cong V_{M\oplus\II_{1,1}}$ and therefore $L\oplus\II_{1,1}\cong M\oplus\II_{1,1}$ so that the lattices $L$ and $M$ are in the same genus. Hence, $j$ is injective.
\end{proof}

As for the bulk genus in \autoref{sec:bulkgenus}, we obtain the following commutative diagram:
\begin{equation}\label{eq:voahgen}
\begin{tikzcd}[ampersand replacement=\&]
\text{lattices}\arrow[twoheadrightarrow]{rr}{\gen}\arrow[hookrightarrow]{dd}[swap]{i} \&\& \text{lattice genera}\arrow[hookrightarrow]{dd}{j}\\\\
\text{conf.\ VAs}\arrow[twoheadrightarrow]{rr}{\hgen} \&\& \text{hyp.\ genera}
\end{tikzcd}
\end{equation}

The above proposition shows that two lattice vertex algebras in the same hyperbolic genus correspond to lattices in the same lattice genus. The same question can be asked for the associated lattices (see \autoref{sec:asslat}) of non-lattice vertex algebras. Indeed, it is already known from \cite{Mor21} (see, e.g., the beginning of Section 4.2 of op.\ cit.) that the associated lattices of \voa{}s in the same hyperbolic genus are in the same lattice genus. We shall make this relationship more precise in \autoref{thm:hypcomm} below.

\medskip

We point out that the (preliminary) version of hyperbolic genus presented in \autoref{defi:hypgenprel} differs slightly from the one given in \cite{Mor21}. While we believe that the above definition is the more natural one, it is easier to work with the slightly more restrictive version given in \cite{Mor21}, as we shall explain in the following.

For simplicity, assume from now on that $V$ and $V'$ are \strat{} \voa{}s. Then, in particular, all Cartan subalgebras of $V$ are conjugate, as are those of $V'$.

Let $\hh$ and $\hh'$ be any choices of Cartan subalgebras of $V$ and $V'$, respectively. Moreover, we take the standard choice of Cartan subalgebra $\hh_{\II_{1,1}}$ of the lattice vertex algebra $V_{\II_{1,1}}$. Then $\hh\oplus\hh_{\II_{1,1}}$ and $\hh'\oplus\hh_{\II_{1,1}}$ are Cartan subalgebras of $V\otimes V_{\II_{1,1}}$ and $V'\otimes V_{\II_{1,1}}$, respectively. Now, in \cite{Mor21} it is additionally demanded that the isomorphism $V\otimes V_{\II_{1,1}}\cong V'\otimes V_{\II_{1,1}}$ map these Cartan subalgebras into each other, leading to the following modified definition.
\begin{defi}[Hyperbolic Genus]\label{defi:hypgen}
Let $V$ and $V'$ be \strat{} \voa{}s. Then $V$ and $V'$ are in the same \emph{hyperbolic genus} if for some (and hence for any) choices of Cartan subalgebras $\hh$ and $\hh'$ of $V$ and $V'$, respectively, there is an isomorphism $V\otimes V_{\II_{1,1}}\cong V'\otimes V_{\II_{1,1}}$ of conformal vertex algebras mapping $\hh\oplus\hh_{\II_{1,1}}$ to $\hh'\oplus\hh_{\II_{1,1}}$.
\end{defi}
Clearly, this definition (which still defines an equivalence relation) is more restrictive than the above one, i.e.\ two \strat{} \voa{}s that are in the same genus in the above sense are also in the same genus in the weaker sense of \autoref{defi:hypgenprel}, simply by forgetting about the Cartan subalgebras. The converse would be true if one could show that all Cartan subalgebras of $V\otimes V_{\II_{1,1}}\cong V'\otimes V_{\II_{1,1}}$ were conjugate.

\begin{conj}
Let $V$ be a \strat{} \voa{} and $V_L$ the conformal vertex algebra associated with an even lattice $L$. Then all Cartan subalgebras of the conformal vertex algebra $V\otimes V_L$ are conjugate.
\end{conj}

In this text, by hyperbolic genus, we shall mean the one according to \autoref{defi:hypgen}, unless otherwise noted, and shall denote it by $\hgen$ from now on.

It is not difficult to see that \autoref{prop:hypercomp} and the diagram~\eqref{eq:voahgen}, in both cases restricted to positive-definite lattices and \strat{} \voa{}s, still hold for this version. The only non-trivial assertion is that if $L$ and $M$ are positive-definite, even lattices in the same genus, then we can find an isomorphism $V_L\otimes V_{\II_{1,1}}\cong V_M\otimes V_{\II_{1,1}}$ that moreover maps $\hh_L\oplus\hh_{\II_{1,1}}$ to $\hh_M\oplus\hh_{\II_{1,1}}$ for the standard choices of Cartan subalgebras $\hh_L$ and $\hh_M$ of $V_L$ and $V_M$, respectively. But the isometry $L\oplus\II_{1,1}\cong M\oplus\II_{1,1}$ extends to the complexifications on both sides, which are naturally isometric to the Cartan subalgebras $\hh_L\oplus\hh_{\II_{1,1}}$ and $\hh_M\oplus\hh_{\II_{1,1}}$, proving the claim.

\begin{remph}\label{remph:currentcurrent}
We briefly comment on an equivalent formulation of hyperbolic genus which was described in Theorem~6.2 of \cite{Mor23}, and which is more natural to physicists than \autoref{defi:hypgen}. We assume again for simplicity that $V$ and $V'$ are strongly rational chiral algebras and write $\mathcal{T}_R^{\mathrm{circ}}$ for the $(c_L,c_R)=(1,1)$  conformal field theory of a compact free boson with radius $R$. In op.\ cit.\ it was shown that $V$ and $V'$ belong to the same hyperbolic genus if and only if one can pass from $V\otimes \mathcal{T}_R^{\mathrm{circ}}$ to $V'\otimes \mathcal{T}_R^{\mathrm{circ}}$ by performing a current-current deformation, i.e.\ by deforming by exactly marginal operators formed out of currents.

In the case that $V$ is part of the worldsheet of a string theory (and similarly for~$V')$, the hyperbolic equivalence of $V$ and $V'$ implies that, once one dimensionally reduces on a circle, the two string theories live on the same connected component of the resulting vacuum moduli space.
\end{remph}


\subsection{Alternative Characterization}\label{subsec:equivchar}

Next, we give a novel characterization of the hyperbolic genus in the context of \strat{} \voa{}s, one which is often easier to work with than \autoref{defi:hypgen}.  This definition will be based on the associated lattice and Heisenberg commutant introduced in \autoref{sec:asslat}. One virtue of this alternative characterization is that it will allow us to demonstrate, in \autoref{cor:hypbulk}, that two vertex operator algebras belonging to the same hyperbolic genus also must belong to the same bulk genus. 

Let $V$ and $W$ be \strat{} \voa{}s. We recall from \autoref{prop:asslatdecomp} that, depending on choices $\hh$ and $\hh'$, respectively, of Cartan subalgebras, they decompose as
\begin{align*}
V&=\bigoplus_{\alpha+L\in A}C^{\tau(\alpha+L)}\otimes V_{\alpha+L},\\
W&=\bigoplus_{\alpha+K\in B}D^{\sigma(\alpha+K)}\otimes V_{\alpha+K}.
\end{align*}
Here, $L$ and $K$ are the (positive-definite, even) associated lattices of $V$ and $W$, and $C$ and $D$ are the (\strat{}) Heisenberg commutants of $V$ and $W$, respectively. Furthermore, $A<L'/L$ and $B<K'/K$ encode the modules of $V_L$ and $V_K$ arising in $V$ and $W$, respectively, and we denote the corresponding full subcategories by $\Rep(V_L\vert V)\coloneqq\mathcal{C}(A)\subset \Rep(V_L)\cong\mathcal{C}(L'/L)$ and $\Rep(V_K\vert W)\coloneqq\mathcal{C}(B)\subset \Rep(V_K)\cong\mathcal{C}(K'/K)$. Similarly, there are pointed full subcategories $\Rep(C\vert V)$ of $\Rep(C)$ and $\Rep(D\vert W)$ of $\Rep(D)$, and finally ribbon-reversing equivalences $\tau\colon\Rep(V_L\vert V)\to\Rep(C\vert V)$ and $\sigma\colon\Rep(V_K\vert W)\to\Rep(D\vert W)$, respectively. We fix this notation for the remainder of the section. 

We recall that an isometry of lattices $L\oplus\II_{1,1}\to K\oplus\II_{1,1}$ naturally induces
an isometry of metric groups $(L\oplus\II_{1,1})'/(L\oplus\II_{1,1})\to (K\oplus\II_{1,1})'/(K\oplus\II_{1,1})$ and hence also an isometry $L'/L\to K'/K$. 
On the level of vertex algebras, this says that an isomorphism of conformal vertex algebras $\phi\colon V_L\otimes V_{\II_{1,1}}\to V_K\otimes V_{\II_{1,1}}$ induces a ribbon equivalence $\bar\phi\colon\Rep(V_L)\to\Rep(V_K)$ (cf.\ \autoref{cor:hypbulk} and \autoref{sec:altproof}).

Similarly, any \voa{} automorphism $\theta\colon C\to D$ naturally induces a ribbon equivalence $\bar\theta\colon\Rep(C)\to\Rep(D)$ on the corresponding modular tensor categories.

We are now in a position to give an equivalent definition of hyperbolic genus.
\begin{defi}[Hyperbolic Genus, Alternative]\label{defi:commgenus}
Let $V$ and $W$ be \strat{} \voa{}s, with choices $\hh$ and $\hh'$, respectively, of Cartan subalgebras and the corresponding decomposition in \autoref{prop:asslatdecomp}. Then $V$ and $W$ are in the same \emph{hyperbolic genus} if:
\begin{enumerate}[label=(\alph*)]
\item\label{item:gen1} There exists a \voa{} isomorphism $\theta\colon C\xrightarrow{\sim} D$ of their Heisenberg commutants.
\item\label{item:gen2} There exists an isomorphism $\phi\colon V_L\otimes V_{\II_{1,1}}\xrightarrow{\sim} V_K\otimes V_{\II_{1,1}}$ of conformal vertex algebras (which we can and will assume to map $\hh\oplus\hh_{\II_{1,1}}$ to $\hh'\oplus\hh_{\II_{1,1}}$), so that in particular the associated lattices $L$ and $K$ are in the same genus.
\item\label{item:gen3} The induced ribbon equivalences $\bar\theta\colon\Rep(C)\to\Rep(D)$ and $\bar\phi\colon\Rep(V_L)\to \Rep(V_K)$ are such that:
\begin{enumerate}[label=(\roman*)]
\item $\bar\theta$ restricts to an equivalence $\vartheta\colon\Rep(C\vert V)\to\Rep(D\vert W)$.
\item $\bar\phi$ restricts to an equivalence $\varphi\colon\Rep(V_L\vert V)\to\Rep(V_K\vert W)$.
\item $\sigma\circ\varphi=\vartheta\circ\tau$.
\end{enumerate}
\end{enumerate}
\end{defi}

The last point in the definition means that the diagram
\begin{equation*}
\begin{tikzcd}[ampersand replacement=\&]
\Rep(V_L\vert V)\arrow[rightarrow]{r}{\tau}\arrow[rightarrow]{d}[swap]{\varphi}\&\Rep(C\vert V)\arrow[rightarrow]{d}{\vartheta}\\
\Rep(V_K\vert W)\arrow[rightarrow]{r}{\sigma}\&\Rep(D\vert W)
\end{tikzcd}
\end{equation*}
commutes, where the vertical arrows are ribbon equivalences and the horizontal ones ribbon-reversing equivalences. We also see that $A\cong B$ as subgroups of $L'/L\cong K'/K$.

The first main theorem of this section is the following result that shows that this definition is equivalent to previous one.
\begin{thm}\label{thm:hypcomm}
Two \strat{} \voa{}s $V$ and $W$ belong to the same hyperbolic genus as in \autoref{defi:commgenus} if and only if they are in the same hyperbolic genus according to \autoref{defi:hypgen}.
\end{thm}
\begin{proof}
Suppose that $V$ and $W$, with choices of Cartan subalgebras $\hh$ and $\hh'$, respectively, are in the same genus according to \autoref{defi:commgenus}. Then, applying the automorphism $(\theta\otimes\phi)^{-1}\colon D\otimes V_{\alpha+K}\otimes V_{\II_{1,1}}\to C\otimes V_{\alpha+L}\otimes V_{\II_{1,1}}$, the corresponding ribbon equivalence allows us to rewrite $W\otimes V_{\II_{1,1}}$ as
{\allowdisplaybreaks
\begin{align*}
W\otimes V_{\II_{1,1}}&=\bigoplus_{\alpha+K\in B}D^{\sigma(\alpha+K)}\otimes V_{\alpha+K}\otimes V_{\II_{1,1}}\\
&\cong\bigoplus_{\alpha+K\in B}C^{\bar\theta^{-1}(\sigma(\alpha+K))}\otimes V_{\bar\phi^{-1}(\alpha+K)}\otimes V_{\II_{1,1}}\\
&=\bigoplus_{\alpha+L\in \bar\phi^{-1}(B)}C^{\bar\theta^{-1}(\sigma(\bar\phi(\alpha+L)))}\otimes V_{\alpha+L}\otimes V_{\II_{1,1}}\\
&=\bigoplus_{\alpha+L\in A}C^{\tau(\alpha+L)}\otimes V_{\alpha+L}\otimes V_{\II_{1,1}}\\
&=V\otimes V_{\II_{1,1}}.
\end{align*}
}%
This shows that $V\otimes V_{\II_{1,1}}$ and $W\otimes V_{\II_{1,1}}$ are isomorphic as simple-current extensions of $C\otimes V_{\alpha+L}\otimes V_{\II_{1,1}}\cong D\otimes V_{\alpha+K}\otimes V_{\II_{1,1}}$, and hence, by an elementary argument (cf.\ Proposition~5.3 in \cite{DM04b}), they are isomorphic as conformal vertex algebras. Moreover, recall that the decompositions of $V$ and $W$ depended on the choice of Cartan subalgebras $\hh$ and $\hh'$ of $V$ and $W$, respectively. By assumption, the above isomorphism maps $\hh\oplus\hh_{\II_{1,1}}$ to $\hh'\oplus\hh_{\II_{1,1}}$. Therefore, $V$ and $W$ are in the same hyperbolic genus according to \autoref{defi:hypgen}.

\smallskip

We now prove the converse statement. Suppose that $V$ and $W$, with the same choices of Cartan subalgebras as before, are in the same hyperbolic genus as in  \autoref{defi:hypgen}, i.e.\ that $V\otimes V_{\II_{1,1}}\cong W\otimes V_{\II_{1,1}}$ and that this isomorphism maps $\hh\oplus\hh_{\II_{1,1}}$ to $\hh'\oplus\hh_{\II_{1,1}}$. Then
\begin{equation*}
C=\Com_{V\otimes V_{\II_{1,1}}}(\hh\oplus\hh_{\II_{1,1}})\cong \Com_{W\otimes V_{\II_{1,1}}}(\hh'\oplus\hh_{\II_{1,1}})=D,
\end{equation*}
say via the isomorphism $\theta$, proving \ref{item:gen1}. Similarly, by taking the double commutants, we obtain
\begin{equation*}
V_L\otimes V_{\II_{1,1}}\cong V_K\otimes V_{\II_{1,1}}
\end{equation*}
as conformal vertex algebras, say via the isomorphism $\phi$, which moreover maps $\hh\oplus\hh_{\II_{1,1}}$ to $\hh'\oplus\hh_{\II_{1,1}}$ by construction. This proves proves \ref{item:gen2}.
In particular, $L\oplus\II_{1,1}\cong K\oplus\II_{1,1}$, i.e.\ $L$ and $K$ are in the same lattice genus. Now, consider as before
\begin{align*}
W\otimes V_{\II_{1,1}}=\bigoplus_{\alpha+K\in B}D^{\sigma(\alpha+K)}\otimes V_{\alpha+K}\otimes V_{\II_{1,1}}
\end{align*}
and via the ribbon equivalence induced by $(\theta\otimes\phi)^{-1}$,
\begin{align*}
(\bar\theta\otimes\bar\phi\otimes\id)^{-1}(W\otimes V_{\II_{1,1}})&=\bigoplus_{\alpha+K\in B}C^{\bar\theta^{-1}(\sigma(\alpha+K))}\otimes V_{\bar\phi^{-1}(\alpha+K)}\otimes V_{\II_{1,1}}\\
&=\bigoplus_{\alpha+L\in \bar\phi^{-1}(B)}C^{\bar\theta^{-1}(\sigma(\bar\phi(\alpha+L)))}\otimes V_{\alpha+L}\otimes V_{\II_{1,1}},
\end{align*}
which is assumed to be isomorphic to
\begin{equation*}
V\otimes V_{\II_{1,1}}=\bigoplus_{\alpha+L\in A}C^{\tau(\alpha+L)}\otimes V_{\alpha+L}\otimes V_{\II_{1,1}}
\end{equation*}
under an isomorphism of conformal vertex algebras preserving $\hh\oplus\hh_{\II_{1,1}}$ and hence $C$ and $V_L\otimes V_{\II_{1,1}}$ setwise. In other words, $V\otimes V_{\II_{1,1}}$ and $(\bar\theta\otimes\bar\phi\otimes\id)^{-1}(W\otimes V_{\II_{1,1}})$ are isomorphic as simple-current extensions of $C\otimes(V_L\otimes V_{\II_{1,1}})$,
and this isomorphism maps $\hh\oplus\hh_{\II_{1,1}}$ into itself.

Restricting this isomorphism to $C$ and to $V_L\otimes V_{\II_{1,1}}$ implies that there are automorphisms $f\in\Aut(C)$ and $g\in\Aut(V_L\otimes V_{\II_{1,1}})$, the latter fixing $\hh\oplus\hh_{\II_{1,1}}$ setwise, inducing ribbon auto-equivalences 
$\bar{f}$ of $\Rep(C)$ and $\bar{g}$ of $\Rep(V_L)$ such that $(\bar{f}\otimes\bar{g}\otimes\id)(V_L\otimes V_{\II_{1,1}})=(\bar\theta\otimes\bar\phi\otimes\id)^{-1}(W\otimes V_{\II_{1,1}})$. This finally tells us that we may modify the above $\theta$ and $\phi$, namely to $\theta\otimes f$ and $\phi\otimes g$, such that \ref{item:gen3} holds.
\end{proof}

One immediate application of this alternative characterization is that it allows us to show that hyperbolic equivalence is finer than bulk equivalence, as conjectured in \cite{Mor21} (and proved in the special case when one of the \voa{}s is assumed to be holomorphic). This is the second main result of this section.
\begin{thm}\label{cor:hypbulk}
If $V$ and $W$ are \strat{} \voa{}s belonging to the same hyperbolic genus, then they belong to the same bulk genus.
\end{thm}
\begin{proof}
Assumptions \ref{item:gen1} and \ref{item:gen2} in \autoref{defi:commgenus} assert the ribbon equivalence
\begin{equation*}
\Rep(C\otimes V_L)\cong\Rep(C)\boxtimes\mathcal{C}(L'/L)\cong\Rep(D)\boxtimes\mathcal{C}(K'/K)\cong\Rep(D\otimes V_K),
\end{equation*}
and \ref{item:gen3} states that the \voa{}s $V$ and $W$ correspond to condensable algebra objects (see, e.g., \cite{Kon14} for the definition of condensable algebra) in $\Rep(C\otimes V_L)$ and $\Rep(D\otimes V_K)$, respectively, which are related under this ribbon equivalence. But then, the categories $\Rep(V)$ and $\Rep(W)$ of local modules of these algebras \cite{HKL15} must be ribbon equivalent as well.
\end{proof}

\begin{rem}
We also sketch a more direct proof of \autoref{cor:hypbulk} in \autoref{sec:altproof}, i.e.\ without using the alternative characterization of the hyperbolic genus established in \autoref{thm:hypcomm}. This proof is considerably harder on a technical level, because we need to deduce from the isomorphism $V\otimes V_{\II_{1,1}}\cong W\otimes V_{\II_{1,1}}$ a ribbon equivalence $\Rep(V)\cong\Rep(W)$, the problem being that $V\otimes V_{\II_{1,1}}$ and $W\otimes V_{\II_{1,1}}$ are only conformal vertex algebras. For the latter, a lot of the general theory is not thoroughly developed.

By contrast, in order to prove the alternative characterization in \autoref{defi:commgenus}, which then almost immediately implies \autoref{cor:hypbulk}, we only need that an isomorphism of \emph{lattice} conformal vertex algebras $V_L\otimes V_{\II_{1,1}}\to V_K\otimes V_{\II_{1,1}}$ induces a ribbon equivalence $\Rep(V_L)\to\Rep(V_K)$.

Moreover, we use that the representation category of a lattice vertex algebra $V_L$ is the pointed modular tensor category $\mathcal{C}(L'/L)$ associated with the discriminant form $L'/L$, regardless of whether $L$ is positive-definite or not (see the remark after \autoref{ex:lattice}).
\end{rem}

The converse of \autoref{cor:hypbulk} is false. A counterexample is provided by the Schellekens \voa{}s discussed in \autoref{sec:holo}. Concretely, while the Leech lattice \voa{} $V_\Lambda$ and the moonshine module $V^\natural$ are in the same bulk genus, they are not hyperbolically equivalent. Indeed, the Heisenberg commutant of $V^\natural$ is $V^\natural$, while that of $V_\Lambda$ is trivial.


\subsection{Unpointed Case}

The statement of \autoref{thm:hypcomm} simplifies considerably if we assume from the outset that the representation category of $V$ is \emph{unpointed} (or \emph{perfect}), i.e.\ that it has no non-trivial simple currents.
Because, by \autoref{cor:hypbulk}, $V$ and $W$ can only be in the same hyperbolic genus if they are in the same bulk genus, we further assume that both $V$ and $W$ have ribbon equivalent representation categories. For simplicity, we also assume pseudo-unitarity, but we do not believe that this is essential. 
\begin{cor}\label{cor:unpointed}
Let $V$ and $W$ be two \strat{} \voa{}s with $\Rep(V)\cong\Rep(W)$ unpointed and pseudo-unitary. Then $V$ and $W$ belong to the same hyperbolic genus if and only if they have the same central charge and isomorphic Heisenberg commutants.
\end{cor}
Before we state the proof, we recall from \cite{Nik80} that for any isometry of even lattices $L'/L\to K'/K$, there is an isometry $L\oplus\II_{1,1}\to K\oplus\II_{1,1}$ that induces the former, as described in \autoref{subsec:equivchar}.
Consequently, for any ribbon equivalence $\bar\phi\colon\Rep(V_L)\to\Rep(V_K)$, which descends to a lattice isometry $L'/L\to K'/K$, there is an isomorphism of conformal vertex algebras $\phi\colon V_L\otimes V_{\II_{1,1}}\to V_K\otimes V_{\II_{1,1}}$ that induces it.
\begin{proof}
The forward direction is immediate. We treat the reverse direction. To this end, suppose that $V$ and $W$ have the same central charge, let $C$ be the Heisenberg commutant of both $V$ and $W$, and call their associated lattices $L$ and $K$, respectively. We aim to show that the conditions in \autoref{defi:commgenus} are met. 

First, we prove that since $V$ has no non-trivial simple currents, it follows that $\Rep(C)\cong\Rep(V)\boxtimes\mathcal{C}(\overline{L'/L})$ and that
\begin{equation*}
V=\bigoplus_{\alpha+L\in L'/L}C^{(0,\tau(\alpha+L))}\otimes V_{\alpha+L}
\end{equation*}
where $\tau\colon\mathcal{C}(L'/L)\to\mathcal{C}(\overline{L'/L})$ is a ribbon-reversing equivalence. To see this, first note that, by virtue of $C\otimes V_L$ being a dual pair in $V$,
\begin{align*}
\Rep(V)&\cong(\Rep(C)\boxtimes\mathcal{C}(L'/L))^{\mathrm{loc}}_{\mathcal{A}},\\
\Rep(C)&\cong(\Rep(V)\boxtimes\mathcal{C}(\overline{L'/L}))^{\mathrm{loc}}_{\mathcal{B}}
\end{align*}
for the condensable algebra $\mathcal{A}$ in $\Rep(C)\boxtimes\mathcal{C}(L'/L)$ determined by $A<L'/L$ and the ribbon-reversing equivalence $\tau$ and for some condensable algebra $\mathcal{B}$ in $\Rep(V)\boxtimes\mathcal{C}(\overline{L'/L})$ \cite{FFRS06}.
This implies that
\begin{align*}
\dim(\Rep(V))&=\frac{\dim(\Rep(C))|L'/L|}{|A|^2},\\
\dim(\Rep(C))&=\frac{\dim(\Rep(V))|L'/L|}{\dim(\mathcal{B})^2}
\end{align*}
so that
\begin{equation*}
|L'/L|^2=|A|^2\dim(\mathcal{B})^2.
\end{equation*}
Now, Corollary~3.14 in \cite{YY21} (see also \cite{CKM17}, Section 4.3) states that
\begin{equation*}
|\Irr(V)_\text{sc}|=\frac{|L'/L|}{|A|}\frac{|\Irr(C)_\text{sc}|}{|A|},
\end{equation*}
where, e.g., $\Irr(V)_{\mathrm{sc}}$ refers to the set of irreducible simple current modules of $V$. We note that both factors are greater than $1$ because $\mathcal{C}(A)\cong\Rep(V_L\vert V)\subset\Rep(V_L)\cong\mathcal{C}(L'/L)$ and $\mathcal{C}(\overline{A})\cong\Rep(C\vert V)\subset\Rep(C)$ are full subcategories. The assumption $|\Irr(V)_\text{sc}|=1$ then implies $|\Irr(C)_\text{sc}|=|A|=|L'/L|$ and in particular $A=L'/L$. Combining this with the previous equation, we derive that $\dim(\mathcal{B})=\pm1$. Pseudo-unitarity rules out the possibility of the minus sign, and we can conclude that $\mathcal{B}$ must be the trivial algebra object. Then $\Rep(C)\cong \Rep(V)\boxtimes \mathcal{C}(\overline{L'/L})$. Moreover, $\Rep(V_L\vert V)\cong \mathcal{C}(L'/L)$ and $\Rep(C\vert V)\cong\mathcal{C}(\overline{L'/L})$, proving the claim. 

Analogously, it follows that
\begin{equation*}
V'=\bigoplus_{\alpha+K\in K'/K}C^{(0,\sigma(\alpha+K))}\otimes V_{\alpha+K}
\end{equation*}
where $\Rep(C)\cong\Rep(V')\boxtimes\mathcal{C}(\overline{K'/K})$ and $\sigma\colon\mathcal{C}(K'/K)\to\mathcal{C}(\overline{K'/K})$ is a ribbon-reversing equivalence. Here, $\Rep(V_K\vert V')\cong \mathcal{C}(K'/K)$ and $\Rep(C\vert V')\cong\mathcal{C}(\overline{K'/K})$. In particular, the same irreducible modules of $C$ arise in both $V$ and $V'$, so that we may take $\theta$ to be the identity isomorphism.

Furthermore, noting that
\begin{equation}\label{eq:cancel}
\begin{split}
\Rep(V)\boxtimes \mathcal{C}(\overline{L'/L})&\cong\Rep(C)\\
&\cong\Rep(V')\boxtimes\mathcal{C}(\overline{K'/K)}\\
&\cong\Rep(V)\boxtimes\mathcal{C}(\overline{K'/K}),
\end{split}
\end{equation}
and recalling that $\Rep(V)\cong\Rep(V')$ has no non-trivial simple currents, it follows that both $\mathcal{C}(\overline{L'/L})$ and $\mathcal{C}(\overline{K'/K})$ are equivalent to the maximal pointed subcategory of $\Rep(C)$, and thus $\mathcal{C}(L'/L)\cong\mathcal{C}(K'/K)$. Since $V$ and $V'$ have the same central charge, this also implies that the lattices $L$ and $K$ are in the same genus.

Finally, using the argument based on \cite{Nik80} stated above, we may choose an isomorphism $\phi\colon V_L\otimes V_{\II_{1,1}}\to V_K\otimes V_{\II_{1,1}}$
that induces $\sigma^{-1}\circ\tau\colon\mathcal{C}(L'/L)\to\mathcal{C}(K'/K)$. This proves that the conditions in \autoref{defi:commgenus} are satisfied.
\end{proof}
In view of \eqref{eq:cancel}, we remark that in general (pseudo-unitary) modular tensor categories do \emph{not} satisfy a cancellation property. Indeed, this already fails for metric groups where, in terms of their Jordan decomposition \cite{CS99}, $2_{\II}^{+2}.2_t^a\cong 2_0^{+2}.2_t^a$ for any odd $2$-adic Jordan component $2_t^a$, but $2_{\II}^{+2}\ncong 2_0^{+2}$. The statement then carries over to the corresponding pointed modular tensor categories.

We also state the following immediate consequence, which we shall discuss in more detail in \autoref{sec:holo}.
\begin{cor}\label{cor:holhyp}
Two \strat{}, holomorphic \voa{}s belong to the same hyperbolic genus if and only if they have the same central charge and isomorphic Heisenberg commutants.
\end{cor}

\begin{ex}[Moonshine-Like \VOA{}s]\label{ex:fake}
Let $V$ and $V'$ be two \strat{} \voa{}s without spin-1 currents, i.e.\ satisfying $V_1=V'_1=\{0\}$. Then $V$ and $V'$ are equal to their respective Heisenberg commutants and thus belong to the same hyperbolic genus if and only if $V$ is isomorphic to $V'$.

In particular, the moonshine module $V^\natural$, and any potential ``fake'' copies of the moonshine module, i.e.\ \strat{}, holomorphic \voa{}s $V$ of central charge~$24$ with $V_1=\{0\}$ and $V\ncong V^\natural$, each live in their own hyperbolic genus (see \autoref{sec:holo} and in particular \autoref{conj:moonshineuniqueness}).
\end{ex}


\subsection{Mass Formula}\label{subsec:mass}

It is interesting to ask whether one can obtain statistics on the \voa{}s which live in a given hyperbolic genus. In the setting of lattices, such statistics come in the form of the Smith--Minkowski--Siegel mass formula and its variations (see, e.g., \cite{CS99} and also \autoref{sec:lat}). Indeed, a remarkable result of classical lattice theory is that for a positive-definite and, say, even lattice $K$, the quantity
\begin{equation*}
    \mass(L) = \sum_{M\in\gen(L)} \frac{1}{|\Aut(M)|}
\end{equation*}
is often computable, even when the list of lattices in the genus is unknown. Here, the sum runs through the finitely many isomorphism classes of lattices $M$ in the genus of $L$. Such mass formulae have played an important role in the classification of integral lattices,  especially for unimodular lattices of rank~$24$. For example, one can demonstrate that a list of known lattices exhausts all of the lattices in a particular genus by showing that the mass of the known lattices is equal to the ``true'' mass.

\medskip

In \cite{Mor21}, the author studies an analogous quantity in the setting of (\strat{}) \voa{}s,
\begin{equation*}
\mass(V) \coloneqq \sum_{W\in\hgen(V)}\frac{1}{|G_W|}
\end{equation*}
where $\hgen(V)$ is the hyperbolic genus of $V$, and $G_W$ is a subgroup of the automorphism group of the associated lattice $M$ of $W$ defined as follows. Fix a Cartan subalgebra $\hh$ of $W$, and consider the subgroup $\Aut(W,\hh)\subset \Aut(W)$ of automorphisms of $W$ which map $\hh$ into itself. Such automorphisms induce automorphisms of the associated lattice, i.e.\ there is a map $\Aut(W,\hh)\to\Aut(M)$, and we define $G_W$ to be the image of this map (which is finite since $\Aut(M)$ is). 

It is straightforward to see that the mass furnishes a lower bound on the number of \voa{}s in the corresponding hyperbolic genus, 
\begin{equation*}
    |\hgen(V)|\geq\mass(V).
\end{equation*}
It is clearly of interest to compute it in concrete examples. This is facilitated by the following result.

We write $\tilde{L}\coloneqq L\oplus\II_{1,1}$ and $\tilde{V}\coloneqq V\otimes V_{\II_{1,1}}$ for brevity. Given a \strat{} \voa{} $V$ with a choice of Cartan subalgebra $\hh$, we consider the Cartan subalgebra $\hh\oplus\hh_{\II_{1,1}}$ of $\tilde{V}$ and, like above, define $G_{\tilde{V}}$ as the subgroup of $\Aut(\tilde{L})$ induced by $\Aut(\tilde{V},\hh\oplus\hh_{\II_{1,1}})$, the difference being that neither $\Aut(\tilde{L})$ nor $G_{\tilde{V}}$ are necessarily finite.
\begin{prop}[Theorem 4.16 of \cite{Mor21}]\label{thm:massformula}
Let $V$ be a \strat{} \voa{} with associated lattice $L$, and assume that $[\mathrm{Aut}(\tilde{L}):G_{\tilde{V}}]$ is finite. The mass of the hyperbolic genus of $V$ can be expressed in terms of the mass of the genus of $L$ as 
\begin{equation*}
    \mass(V)=\mass(L)[\Aut(\tilde{L}):G_{\tilde{V}}].
\end{equation*}
\end{prop}

\begin{rem}
In \cite{Mor21}, the author worked with vertex algebras, as opposed to strongly rational vertex operator algebras. In this more general setting, the finiteness of the group-theoretical index $[\mathrm{Aut}(\tilde{L}):G_{\tilde{V}}]$ needs to be added as an assumption (in addition to the positive-definiteness of the associated lattice $L$, which is automatic for strongly rational vertex operator algebras). However, we will show in \autoref{cor:massconstant} below that this group-theoretical index is always finite when $V$ is a strongly rational vertex operator algebra, and so this assumption can be relaxed in \autoref{thm:massformula}.
\end{rem}

Thus, the mass of a hyperbolic genus of \voa{}s is as computable as the mass of a genus of lattices, assuming that one has control over the group index serving as constant of proportionality. To this end, we turn to characterizing the group $G_V$ for a conformal vertex algebra $V$, which we shall later take to be either a \strat{} \voa{} $V$ or that \strat{} \voa{} tensored with $V_{\II_{1,1}}$, i.e.\ $\tilde{V}$ (cf.\ \cite{Hoe17,HM23}).

\medskip

For now, let $V$ be an arbitrary conformal vertex algebra of the form
\begin{equation*}
    V = \bigoplus_{\alpha +L \in A}C^{\tau(\alpha+L)}\otimes V_{\alpha +L},
\end{equation*}
where $C$ is a \strat{} \voa{} with $C_1=\{0\}$, $L$ is a (possibly indefinite) even lattice, $A<L'/L$ and $\tau\colon\mathcal{C}(A)\to\Rep(C\vert V)$ is a ribbon-reversing equivalence. We choose the Cartan subalgebra $\hh\coloneqq\hh_L$ of $V$, where $\hh_L$ is the standard Cartan subalgebra of $V_L$.
Let $\overline{\Aut}(L)$ be the image of the natural map 
\begin{equation}\label{eqn:muL}
\mu_L\colon\Aut(L)\to\Aut(\mathcal{C}(L'/L)),
\end{equation}
where $\Aut(\mathcal{C}(L'/L))$ is the group of ribbon auto-equivalences of $\mathcal{C}(L'/L)$, and let $\overline{\Aut}(L, A)$ be the subgroup of auto-equivalences in $\overline{\Aut}(L)$ that restrict to auto-equivalences of $\mathcal{C}(A)=\Rep(V_L\vert V)$. Similarly, let $\overline{\Aut}(C)$ be the image of the natural map
\begin{equation*}
\Aut(C)\to\Aut(\Rep(C)),
\end{equation*}
where $\Aut(\Rep(C))$ is the group of ribbon auto-equivalences of $\Rep(C)$, and let $\overline{\Aut}(C,V)$ be the subgroup of ribbon auto-equivalences in $\overline{\Aut}(C)$ that restrict to ribbon auto-equivalences of $\Rep(C\vert V)$. Finally, call $\tau^\ast\colon\Aut(\Rep(C\vert V))\to \Aut(\mathcal{C}(A))$ the map defined by $\bar{\theta}\mapsto \tau^{-1}\circ \bar{\theta}\circ \tau$.

For a conformal vertex algebra $V$ given as simple-current extension like above (and with the choice of Cartan subalgebra $\hh$), our result characterizes the isometries $\phi$ of $L$ in the subgroup $G_V<\Aut(L)$. We argue below that they are precisely those isometries satisfying:
\begin{enumerate}
    \item the induced map $\bar{\phi}\in\Aut(\mathcal{C}(L'/L))$ restricts to a ribbon auto-equivalence of $\mathcal{C}(A)$,
    \item there exists an automorphism $\theta\in\Aut(C)$ such that the induced map $\bar{\theta}\in\Aut(\Rep(C))$ restricts to a ribbon auto-equivalence of $\Rep(C\vert V)$,
    \item $\tau^{-1}\circ \bar{\theta}\circ \tau = \bar{\phi}$. 
\end{enumerate}
More succinctly: 
\begin{prop}\label{prop:GVcharacterization}
Let $V$ be a conformal vertex algebra as above. Then the group $G_V<\Aut(L)$ is given by 
    \begin{equation*}
        G_V=\mu_L^{-1}\left( \overline{\Aut}(L,A)\cap \tau^\ast\overline{\Aut}(C,V) \right).
    \end{equation*}
\end{prop}
\begin{proof}
We start by showing that $G_V$ is contained in the group on the right-hand side. Consider an automorphism $\varphi$ of $V$ in $\Aut(V,\hh)$. Because $\varphi$ maps the Cartan subalgebra $\hh\subset V$ to itself, it restricts to an automorphism $\theta\otimes\Phi\in\Aut(C\otimes V_L)$, where $\theta\in\Aut(C)$ and $\Phi\in\Aut(V_L)$. In turn, $\Phi$ (which still fixes $\hh$ setwise) induces an automorphism $\phi\in\Aut(L)$, and the assignment $\varphi \mapsto \phi$ is precisely the map $\Aut(V,\hh)\to G_V< \Aut(L)$. Because $\theta\otimes \Phi$ lifts to an automorphism of $V$, it follows that $\theta\otimes \Phi$ must permute those $C\otimes V_L$-modules that appear in the decomposition of $V$ among themselves. It is straightforward to see that this happens if and only if $\bar{\Phi}=\bar{\phi}\in\Aut(\mathcal{C}(L'/L))$ and $\bar{\theta}\in\Aut(\Rep(C))$ satisfy the conditions (1) to (3) listed right before the statement of the proposition. This in turn means that $\phi$ belongs to the group on the right-hand side of the assertion.

Next, we show the reverse inclusion. Consider an automorphism $\phi\in\Aut(L)$ belonging to the group on the right-hand side. That is, there exists an automorphism $\theta\in\Aut(C)$ such that $\bar\theta$ satisfies the conditions (1) to (3). Let $\Phi$ be a lift of $\phi$ to an automorphism of $V_L$. Then, by a suitable generalization (to conformal vertex algebras) of Theorem~2.1 in \cite{Shi07}, the automorphism $\theta\otimes\Phi$ lifts to an automorphism of $V$ in $\Aut(V,\hh)$ that induces the automorphism $\phi$ of $L$. Thus, $\phi$ belongs to $G_V$.
\end{proof}

We return to the situation where $V$ is a \strat{} \voa{} with a choice of Cartan subalgebra $\hh$ and corresponding associated lattice $L$, and moreover consider $\tilde{V}=V\otimes V_{\II_{1,1}}$ with the Cartan subalgebra $\hh\oplus\hh_{\II_{1,1}}$ and the corresponding lattice $\tilde{L}\coloneqq L\oplus\II_{1,1}$.
We assume furthermore that $V$ is holomorphic, which simplifies the following result.
\begin{cor}\label{cor:massconstant}
If $V$ is a \strat{}, holomorphic \voa{}, then 
\begin{align*}
[\Aut(L):G_V]&=[\overline{\Aut}(L):\overline{\Aut}(L)\cap \tau^\ast \overline{\Aut}(C)], \\
[\Aut(\tilde{L}):G_{\tilde{V}}]&=[\Aut(\Rep(C)):\overline{\Aut}(C)].
\end{align*}
In particular, $[\mathrm{Aut}(L):G_V]$ and  $[\mathrm{Aut}(\tilde{L}):G_{\tilde{V}}]$ are finite when $V$ is a strongly rational vertex operator algebra.
\end{cor}
\begin{proof}
Let us consider the first identity. Because $V$ is holomorphic, $A=L'/L$ and hence $\overline{\Aut}(L,A)=\overline{\Aut}(L)$. Similarly, $\overline{\Aut}(C,V)=\overline{\Aut}(C)$. Then, we claim that the map $\mu_L$ in \eqref{eqn:muL} induces a map $\bar{\mu}_L$ of (left) cosets,
\begin{equation*}
\bar{\mu}_L\colon\Aut(L)/G_V\to \overline{\Aut}(L)/(\overline{\Aut}(L)\cap\tau^\ast\overline{\Aut}(C)).
\end{equation*}
Indeed, by \autoref{prop:GVcharacterization}, this map is well-defined and injective. It is also surjective because $\overline{\Aut}(L)$ is by definition the image of $\mu_L$. The identity follows.

To prove the second identity, we note that the same argument above shows that $[\Aut(\tilde{L}):G_{\tilde{V}}]=[\overline{\Aut}(\tilde{L}):\overline{\Aut}(\tilde{L})\cap \tau^\ast \overline{\Aut}(C)]$. On the other hand, it follows from the result of \cite{Nik80} stated before the proof of \autoref{cor:unpointed} that $\overline{\Aut}(\tilde{L})=\Aut(\mathcal{C}(\tilde{L}'/\tilde{L}))\cong \Aut(\Rep(C))$ and hence that $\overline{\Aut}(\tilde L)\cap\tau^\ast\overline{\Aut}(C)=\tau^\ast\overline{\Aut}(C)$. The second identity follows.
\end{proof}
This corollary in particular expresses the constant of proportionality in the mass formula entirely in terms of the Heisenberg commutant $C$ that, since $V$ is holomorphic, characterizes the hyperbolic genus of $V$ (see \autoref{cor:holhyp}). Moreover, as we shall see in \autoref{subsec:holomass}, this reformulation of the index reveals that it has already been computed in key cases of interest.

\begin{remph}
Physically, we can interpret the constant of proportionality as follows. Consider a hyperbolic genus of chiral CFTs. All of the theories in this hyperbolic genus possess the same interacting sector $C$. The bulk 3d TQFT to which $C$ is attached has a group $\Aut(\Rep(C))$ of invertible surface operators. Only a subgroup $\overline{\Aut}(C)$ of these invertible surface operators can actually terminate topologically on the gapless chiral boundary defined by $C$. The constant of proportionality is the index of this subgroup.
\end{remph}

Examples of mass formulae for hyperbolic genera shall be discussed in \autoref{subsec:holomass} below (see also \autoref{tab:Schellekensgenera}). We will also state in \autoref{subsec:SiegelWeil} a (conjectural) generalization of the mass formula to a Siegel--Weil identity in the context of holomorphic vertex operator algebras. That is, for holomorphic vertex operator algebras, we will not only be able to work out the weighted sum of the number $1$ within a hyperbolic genus, as in the mass formula, but further the weighted sum of the partition functions. We will describe a holographic interpretation of this Siegel--Weil identity.


\subsection{Examples}

We discuss some examples of hyperbolic genera. Often, the rather concrete characterization in \autoref{defi:commgenus} is more useful than the one given in \autoref{defi:hypgen}.

As a starting point, we briefly consider the holomorphic case discussed in \autoref{ex:schellekens}. Again, we refer to \autoref{sec:holhyp} for a more detailed treatment.
\begin{ex}[Holomorphic \VOA{}s]\label{ex:schellekens2}
The idea to organize \strat{}, holomorphic \voa{}s into families labeled by their Heisenberg commutants was put forward by Höhn in \cite{Hoe17}. By \autoref{cor:holhyp}, this exactly describes the decomposition of the bulk genera into hyperbolic genera.

We saw that the bulk genera $(\Vect,0)$, $(\Vect,8)$ and $(\Vect,16)$ only contain lattice \voa{}s. Hence, their Heisenberg commutants are trivial and each of these bulk genera defines just one hyperbolic genus.

For central charge $c=24$ one finds that the $71$ Schellekens \voa{}s decompose into $12$ hyperbolic genera, labeled by the letters A to L (see \autoref{table:12}). For example, genus~A are the $24$ Niemeier lattice \voa{}s and the Heisenberg commutant for genus~B was already described in \autoref{ex:commB}.

The hyperbolic genus~$L$ only contains the moonshine module $V^\natural$. There could be further \strat{}, holomorphic \voa{}s $V$ of central charge $24$ with $V_1=\{0\}$, but we showed in \autoref{ex:fake} that they each have to be in their own hyperbolic genus. The moonshine uniqueness conjecture (see \autoref{conj:moonshineuniqueness}) posits that such hyperbolic genera do not exist.
\end{ex}

\begin{ex}[Bulk Genus $((G_2,1),\sfrac{94}{5})$]
Consider the bulk genus $(\mathcal{C},c)$ with modular tensor category $\mathcal{C}=(G_2,1)$ and central charge $c=\sfrac{94}{5}$. From \cite{SR23}, coupled with \autoref{prop:essentiallypositive}, we know that there are eight \strat{} \voa{}s in this bulk genus. Since $\mathcal{C}$ does not have any non-trivial simple currents, \autoref{cor:unpointed} tells us that we only need to check that the Heisenberg commutants of two vertex operator algebras are isomorphic to decide if they are hyperbolically equivalent.

The six \voa{}s (see \autoref{subsec:holobulk} for the notation)
\begin{align*}
&\mathbf{S}(\mathsf{E}_{6,1}^4)\big/\mathsf{F}_{4,1}, \quad \mathbf{S}( \mathsf{A}_{11,1}\mathsf{D}_{7,1}\mathsf{E}_{6,1})\big/\mathsf{F}_{4,1}, \quad \mathbf{S}(\mathsf{D}_{10,1}\mathsf{E}_{7,1}^2)\big/\mathsf{F}_{4,1},\\
&\mathbf{S}(\mathsf{A}_{17,1}\mathsf{E}_{7,1})\big/\mathsf{F}_{4,1}, \quad \mathsf{E}_{8,1}^3\big/\mathsf{F}_{4,1}\cong\mathsf{E}_{8,1}^2\mathsf{G}_{2,1}, \quad \mathbf{S}(\mathsf{D}_{16,1}\mathsf{E}_{8,1})\big/\mathsf{F}_{4,1}
\end{align*}
that come from the coset of a Niemeier lattice \voa{} by $\mathsf{F}_{4,1}$ belong to a single hyperbolic genus. Indeed, the Heisenberg commutant of all these \voa{}s is $C=L(\sfrac45,0)\oplus L(\sfrac45,3)$, an extension of the (discrete series) Virasoro minimal model for $(p,q)=(5,6)$. $C$ has representation category $\Rep(C)\cong(G_2,1)\boxtimes(E_6,1)$, and in each case it is paired with a lattice \voa{} $V_L$ with $\Rep(V_L)\cong(A_2,1)=\overline{(E_6,1)}=\Cat(3^{-1})$ and $\rk(L)=18$. Indeed, there are exactly six lattices $L$ in the genus $\II_{18,0}(3^{-1})$, for example $E_8^2A_2$.

On the other hand, the two remaining \voa{}s
\begin{equation*}
\mathbf{S}(\mathsf{C}_{8,1}\mathsf{F}_{4,1}^2)\big/\mathsf{F}_{4,1}, \quad \mathbf{S}(\mathsf{E}_{7,2}\mathsf{B}_{5,1}\mathsf{F}_{4,1})\big/\mathsf{F}_{4,1}
\end{equation*}
likely together form their own hyperbolic genus. For example, their Heisenberg commutants will both be of the form $\mathrm{Ex}(C'/P)$, where $C'$ is the Heisenberg commutant of genus B (see \autoref{table:12} and \autoref{ex:commB}), and $P\coloneqq K(F_4,1)$ is the parafermion \voa{} (see \autoref{ex:para}) obtained by taking the commutant in $\mathsf{F}_{4,1}$ of its Cartan subalgebra \cite{DW11}. The only subtlety is that $P$ may be embedded into $C'$ differently in the two cases, potentially leading to different Heisenberg commutants. In principle, the hyperbolic equivalence of these two vertex operator algebras could be checked by an application of the mass formula, but we leave this to future work. 
\end{ex}


\section{Holomorphic Vertex Operator Algebras}\label{sec:holo}

In this section, we apply some of the machinery developed throughout the text to the special case of \strat{}, holomorphic vertex operator algebras, which were already treated briefly in \autoref{ex:schellekens} and \autoref{ex:schellekens2}.


\subsection{Bulk Genera}\label{subsec:holobulk}

We begin by reviewing what is known about the organization of \strat{}, holomorphic vertex operator algebras into bulk genera. Recall that, by definition, a holomorphic vertex operator algebra $V$ is a (simple and rational) \voa{} with trivial representation category, i.e.\ $\Rep(V)\cong\Vect\cong \mathcal{C}(1)$. This actually implies that the central charge $c$ lies in $8\N$ \cite{Zhu96}, analogously to positive-definite, even, unimodular lattices only existing if the rank $r$ is in $8\N$. Thus, \strat{}, holomorphic vertex operator algebras belong to bulk genera of the form $(\mathcal{C},c) = (\Vect,8n)$ for $n\geq 0$. 

For central charges $c=0$, $8$ and $16$, exactly the lattice genera are reproduced, i.e.\ all \voa{}s are lattice \voa{}s \cite{DM04}:
\begin{align*}
(\Vect,0)&=\{\C\vac\},\\
(\Vect,8)&=\{V_{E_8}\},\\
(\Vect,16)&=\{V_{(E_8)^2},V_{D_{16}^+}\}.
\end{align*}
In other words, each bulk genus decomposes into exactly one hyperbolic genus, which has trivial Heisenberg commutant.

At central charge $24$, there are exactly $70$ \strat{}, holomorphic \voa{}s $V$ with $V_1\neq\{0\}$ up to isomorphism, which we refer to as the Schellekens \voa{}s \cite{Sch93}. Each of these $70$ \voa{}s is uniquely determined by the simple affine \voa{} $\langle V_1\rangle\cong\mathsf{X}^{(1)}_{r_1,k_1}\otimes\dots\otimes\mathsf{X}^{(n)}_{r_n,k_n}$,
and we write $\mathbf{S}(\mathsf{X}^{(1)}_{r_1,k_1}\dots\mathsf{X}^{(n)}_{r_n,k_n})$ for the unique \strat{}, holomorphic \voa{} of central charge $c=24$ with this affine structure. The classification of this class of \voa{}s was only recently made mathematically rigorous (see, e.g., \cite{Hoe17,MS23,MS21,HM22,LM22}), even though the original work of Schellekens dates back several decades. Of these $70$ \voa{}s, only $24$ correspond to lattices (the so-called Niemeier lattices);
hence, the genus $(\Vect,24)$, and in fact all genera $(\Vect,c)$ with $c\geq24$, contain \voa{}s that are not lattice \voa{}s. That is, these bulk genera decompose into more than one hyperbolic genus (see \autoref{sec:holhyp} below).

The remaining gap to fill in the classification of the bulk genus $(\Vect,24)$ is to characterize the \voa{}s with $V_1=\{0\}$. These are substantially harder to study, because one does not have the crutch of continuous global symmetries (or affine \voa{}s), which were essential to the study of the Schellekens vertex operator algebras. Conjecturally, there is only one such \voa{} with $V_1=\{0\}$ \cite{FLM88}.

\begin{conj}[Moonshine Uniqueness]\label{conj:moonshineuniqueness}
The moonshine module $V^\natural$ is the unique \strat{}, holomorphic \voa{} $V$ of central charge $c=24$ (i.e.\ in the bulk genus $(\mathcal{C},c)=(\Vect,24)$) with $V_1=\{0\}$.
\end{conj}

The landscape of \strat{}, holomorphic \voa{}s becomes quite unwieldy at central charges beyond $24$. For example, when $c=32$, one can show, using refinements of the Smith--Minkowski--Siegel mass formula (see, e.g., \cite{CS99}), that there are more than one billion lattice \voa{}s alone \cite{Kin03}. Thus, enumeration of isomorphism classes of \voa{}s in bulk genera of the form $(\Vect,c)$ with $c>24$ is infeasible, though it is natural to ask whether they can be classified up to some coarser notion of equivalence, like hyperbolic equivalence. We explore this idea further in subsequent sections. 


\subsection{Hyperbolic Genera}\label{sec:holhyp}

We start by describing some generalities that will allow us to determine how bulk genera of \strat{}, holomorphic vertex operator algebras split into hyperbolic genera. 

When a \voa{} $V$ is \strat{} and holomorphic, the decomposition in \autoref{prop:asslatdecomp} in terms of the associated lattice $L$ and Heisenberg commutant $C$ simplifies to
\begin{equation*}
V=\bigoplus_{\alpha+L\in L'/L}C^{\tau(\alpha+L)}\otimes V_{\alpha+L}
\end{equation*}
for some ribbon-reversing equivalence $\tau\colon\Rep(V_L)\to\Rep(C)$. In particular, since $\Rep(V)\cong\Vect$ is (trivially) pointed and pseudo-unitary, the same is true for $\Rep(C)$ by Proposition~4.4 and 4.5 in \cite{HM23}. In other words, there is a metric group $R_C$ such that $\Rep(C)\cong\mathcal{C}(R_C)$. Moreover, there is an anti-isometry $\tau\colon L'/L\to R_C$ inducing the ribbon-reversing equivalence of the same name. Finally, note that $C$ has vanishing weight-$1$ space $C_1=\{0\}$ since otherwise $L$ could not be the associated lattice (see also \autoref{prop:hypgenhol} below).

\begin{rem}\label{rem:wrongext}
We issue a word of warning. Suppose for simplicity that $C$ is positive. We recall \cite{Moe16,EMS20a} that holomorphic extensions of $C\otimes V_L$, like $V$, correspond bijectively to the self-dual, isotropic subgroups $I$ of $R_C\times L'/L$. The isotropic subgroup for $V$ is $I=\{(\tau(x),x)\,|\,x\in A\}<K'/K\times R_C$. Indeed, any subgroup of this form with some (other) anti-isometry $\tau\colon L'/L\to R_C$ is also self-dual and isotropic, and the corresponding extension is a holomorphic \voa{} in which $C$ is a Heisenberg commutant and $C\otimes V_L$ a dual pair.

However, not all holomorphic simple-current extensions have the latter property. For example, consider the dual pair $C\otimes V_L$ appearing in the hyperbolic genus F (see \autoref{table:12} below). Both $C$ and $V_L$ have the representation category $\mathcal{C}(5^{+6})$. The underlying metric group $5^{+6}$ has signature $0\pmod{8}$ and has a self-dual, isotropic subgroup $J$. Correspondingly, both $C$ and $V_L$ can \emph{individually} be extended to holomorphic \voa{}s of central charge $16$ and $8$, respectively, which then both must be lattice \voa{}s, as we wrote in \autoref{subsec:holobulk}. In other words, the simple-current extension $\tilde{V}$ of $C\otimes V_L$ corresponding to the self-dual, isotropic subgroup $\tilde{I}\coloneqq J\times J<K'/K\times R_C$ is some Niemeier lattice \voa{}. Clearly, neither is $C$ the Heisenberg commutant of $\tilde{V}$, nor is $C\otimes V_L$ a dual pair in $\tilde{V}$.
\end{rem}

\medskip

Based on \autoref{cor:holhyp} and the preceding discussion, we can label any hyperbolic genus of holomorphic \voa{}s, say $\hgen(V)$ with representative $V$, by the pair $(C,c)$ consisting of the central charge $c\in 8\N$ of $V$ and the Heisenberg commutant $C$ of $V$ up to isomorphism. To emphasize that this hyperbolic genus sits inside the bulk genus $(\Vect,c)$, we shall write
\begin{equation*}
(C,c)_{\Vect}
\end{equation*}
for the hyperbolic genus. These two invariants must satisfy a number of properties, which we summarize in the following proposition. In fact, we shall see that any such pair $(C,c)_{\Vect}$ with these properties corresponds to a hyperbolic genus. Here, for a metric group $D$, $\overline{D}$ denotes the same abelian group with the quadratic form multiplied by $-1$.
\begin{prop}\label{prop:hypgenhol}
There is a bijection between hyperbolic genera of \strat{}, holomorphic \voa{}s and pairs $(C,c)_{\Vect}$ where $c\in 8\N$ and $C$ is an isomorphism class of \strat{} \voa{}s such that
\begin{enumerate}
\item\label{item:prop1} the weight-$1$ Lie algebra vanishes $C_1=\{0\}$,
\item the modular tensor category $\Rep(C)$ is pseudo-unitary and pointed, i.e.\ ribbon equivalent to $\Rep(C)\cong\mathcal{C}(R_C)$ for some metric group $R_C$,
\item there is a (non-empty) lattice genus for the symbol $\II_{c-c(C),0}(\overline{R_C})$, implying in particular that $c(C)\leq c$ and $c(C)\in\Z$,
\item\label{item:prop4} there is a lattice $L$ in $\II_{c-c(C),0}(\overline{R_C})$ and an anti-isometry $\tau\colon L'/L\to R_C$ such that $C^{\tau(\alpha+L)}$ has lowest $L_0$-weight
\begin{equation*}
\rho(C^{\tau(\alpha+L)})>-\min_{v\in\alpha+L}\langle v,v\rangle/2
\end{equation*}
for all non-zero $\alpha+L\in L'/L$.
\end{enumerate}
The bijection maps the hyperbolic genus $\hgen(V)$ to the central charge $c$ and the Heisenberg commutant $C$ of $V$ up to isomorphism.
\end{prop}

We remark that condition~\eqref{item:prop4} is trivially satisfied if the Heisenberg commutant is positive, which would also imply the pseudo-unitarity of $\Rep(C)$.
\begin{proof}
Given a \strat{}, holomorphic \voa{} $V$, representing the hyperbolic genus $\hgen(V)$, we associate with it the pair $(C,c)_{\Vect}$ as specified above. This map is well-defined, i.e.\ only depends on the genus $\hgen(V)$, and injective by \autoref{cor:holhyp}. It is not difficult to see that $C$ must satisfy the stated properties.

\smallskip

Conversely, let $(C,c)_{\Vect}$ be a pair with these properties. We prove that there is a \strat{}, holomorphic \voa{} $V$ that is mapped to $(C,c)_{\Vect}$. Let $L$ be a lattice in the genus $\II_{c-c(C),0}(\overline{R_C})$ and $\tau\colon L'/L\to R_C$ an anti-isometry, which lifts to a ribbon-reversing equivalence $\tau\colon\Rep(V_L)\to\Rep(C)$, such that \eqref{item:prop4} holds. We define the simple-current extension
\begin{equation*}
V\coloneqq\bigoplus_{\alpha+L\in L'/L}C^{\tau(\alpha+L)}\otimes V_{\alpha+L},
\end{equation*}
which is a ($\Z$-graded) \voa{} because $\Rep(C)$ is pseudo-unitary, as is $\Rep(V_L)$ (cf.\ \cite{CKL20,CKM17,CKM22,HM23}). Item~\eqref{item:prop4} precisely guarantees that $V$ is ($\N$-graded and) of CFT-type, and then, since $V$ and $C$ are \strat{}, so is the simple-current extension $V$.
By construction, the central charge of $V$ is $c(C)+c(V_L)=c(C)+(c-c(C))=c$.

It remains to show that $C$ is the Heisenberg commutant of $V$ or equivalently, since $C$ and $V_L$ form a dual pair in $V$ by the mirror extension, that $L$ is the associated lattice of $V$. We consider $\hh\coloneqq\C\vac\otimes\{k(-1)\ee_0\mid k\in L\otimes_\Z\C\}$, which is clearly a toral subalgebra of $V$. We show that $\hh$ is in fact a Cartan subalgebra, i.e.\ a maximal toral subalgebra, of $V$. First, note that any element $v$ in a toral subalgebra of $V$ extending $\hh$ must satisfy $h_0v=0$ for all $h\in\hh$. Since the bilinear form on $L$ and its extension to $\hh\cong L\otimes_\Z\C$ are non-degenerate, this shows that $v\in C\otimes M_{\hat\hh}(1,0)\subset C\otimes V_L$. Moreover, as $C_1=\{0\}$ and because $v$ must have $L_0$-weight $1$, it follows that $v\in M_{\hat\hh}(1,0)_1=\hh$.

Finally, the conformal vector of $V$ is given by the sum $\omega^V=\omega^C+\omega^{\hh}$ of the conformal vectors of $C$ and of $V_L$. Note that $\omega^{\hh}$ is also the conformal vector of the Heisenberg \voa{} $\langle\hh\rangle=M_{\hat\hh}(1,0)\subset V_L$. Then
\begin{equation*}
C=\Com_V(V_L)=\ker(\omega^\hh_0)=\Com_V(\hh)
\end{equation*}
is the Heisenberg commutant of $V$ \cite{FZ92},
proving the assertion.
\end{proof}

In the following, we shall label hyperbolic genera of \strat{}, holomorphic \voa{}s by pairs $(C,c)_{\Vect}$ satisfying properties \eqref{item:prop1} to \eqref{item:prop4} above.

\begin{cor}\label{cor:increasingc}
If $(C,c)_{\Vect}$ is a hyperbolic genus of \strat{}, holomorphic \voa{}s, then so is $(C,c+8n)_{\Vect}$ for all $n\in\N$.
\end{cor}
\begin{proof}
If $V$ belongs to the hyperbolic genus $(C,c)_{\Vect}$, then tensoring with the lattice \voa{} $V_{E_8}=\mathsf{E}_{8,1}$ for $E_8$ increases the central charge by~$8$, but leaves the Heisenberg commutant unchanged. Then, \autoref{prop:hypgenhol} shows that $V\otimes V_{E_8}^{\otimes n}$ belongs to the hyperbolic genus $(C,c+8n)_{\Vect}$.
\end{proof}

\medskip

Next, we describe how the Schellekens \voa{}s (see \autoref{subsec:holobulk}) divide into hyperbolic genera. The idea to organize the \strat{} \voa{}s of central charge $c=24$ into families labeled by their Heisenberg commutants was put forward in \cite{Hoe17}. By \autoref{cor:holhyp}, this exactly describes the decomposition of the bulk genus $(\Vect,24)$ into hyperbolic genera.

Once again, we work under the assumption of \autoref{conj:moonshineuniqueness}. (If not, all that would change is that we would end up with ``fake'' copies of the hyperbolic genus labeled by the letter L below.) One of the main statements of \cite{Hoe17} is that all Heisenberg commutants $C$ of the Schellekens \voa{}s are either the moonshine module $V^\natural$ (or ``fake'' copies thereof) or come from the Leech lattice~$\Lambda$ in the sense that they are of the form
\begin{equation}\label{eqn:SchellekensHeisenbergCommutant}
C\cong V_{\Lambda_\nu}^{\hat\nu}
\end{equation}
for some automorphism $\nu\in\Aut(\Lambda)$. Here, $\Lambda_\nu=(\Lambda^\nu)^\bot$ denotes the coinvariant lattice, the orthogonal complement of the invariant lattice $\Lambda^\nu$, and $\hat\nu$ is a lift of $\nu$ (restricted to $\Lambda_\nu$) to $\Aut(V_{\Lambda_\nu})$ . Since $\nu$ acts fixed-point freely on $\Lambda_\nu$, all its lifts are standard and conjugate (see, e.g., \cite{EMS20b}). 

This statement was shown in \cite{Hoe17} case by case using Schellekens' classification result, but it can be seen conceptually based on the results in \cite{ELMS21,MS23} (see also Section~4.2 in \cite{HM22}). If one follows the latter approach, it is an easy task to determine all conjugacy classes $\nu\in\Aut(\Lambda)$ such that $C=V_{\Lambda_\nu}^{\hat\nu}$ (which is always positive) satisfies the conditions in \autoref{prop:hypgenhol}. In fact, the only non-trivial criterion is that there is a (non-empty) lattice genus
\begin{equation*}
\II_{c-c(C),0}(\overline{R_C})=\II_{\rk(\Lambda^\nu),0}(\overline{R_C}).
\end{equation*}
The shape of the metric group $R_C$ is described in \cite{Lam20} (with some partial results in \cite{Moe21}), and in particularly nice (and not so rare) cases it is of the form $R_C\cong\Lambda'_\nu/\Lambda_\nu\times H$ where $H$ is the metric group associated with a twisted Drinfeld double $\mathcal{D}_\omega(\Z_n)\cong\mathcal{C}(H)$ where $n=|\nu|$.

This way, one is left with exactly $11$ conjugacy classes $\nu\in\Aut(\Lambda)$, and correspondingly $11$ Heisenberg commutants, labeled by the letters A to K, which we list in \autoref{table:12}. Moreover, the moonshine module $V^\natural$, which is its own Heisenberg commutant (since its associated lattice is trivial) and does not come from the Leech lattice in the above sense, is labeled by the letter $L$.

\begin{table}[ht]\caption{The $12$ Heisenberg commutants (or hyperbolic genera) appearing in the Schellekens \voa{}s and the corresponding associated lattice genera.}
\begin{tabular}{c|l|r|l||l|r||r}
\multicolumn{2}{l|}{Heisenberg comm.} & $c$ & $\Aut(\Lambda)$ & Lattice genus & No.\ &  VOAs\ \\\hline\hline
A & $\mathcal{C}(1)$                              &  0 & $\sAA$ & $\gAA$ & 24 & 24\\
B & $\mathcal{C}(2_{\II}^{+10})$                  &  8 & $\sBB$ & $\gBB$ & 17 & 17\\
C & $\mathcal{C}(3^{-8})$                         & 12 & $\sCC$ & $\gCC$ &  6 &  6\\
D & $\mathcal{C}(2_{\II}^{-10}4_{\II}^{-2})$      & 12 & $\sDD$ & $\gDD$ &  2 &  9\\
E & $\mathcal{C}(2_6^{+2}4_{\II}^{-6})$           & 14 & $\sEE$ & $\gEE$ &  5 &  5\\
F & $\mathcal{C}(5^{+6})$                         & 16 & $\sFF$ & $\gFF$ &  2 &  2\\
G & $\mathcal{C}(2_{\II}^{+6}3^{-6})$             & 16 & $\sGG$ & $\gGG$ &  2 &  2\\
H & $\mathcal{C}(7^{+5})$                         & 18 & $\sHH$ & $\gHH$ &  1 &  1\\
I & $\mathcal{C}(2_1^{+1}4_5^{-1}8_{\II}^{-4})$   & 18 & $\sII$ & $\gII$ &  1 &  1\\
J & $\mathcal{C}(2_{\II}^{+4}4_{\II}^{-2}3^{-5})$ & 18 & $\sJJ$ & $\gJJ$ &  1 &  2\\
K & $\mathcal{C}(2_{\II}^{-2}4_{\II}^{-2}5^{+4})$ & 20 & $\sKK$ & $\gKK$ &  1 &  1\\\hline
L & $\mathcal{C}(1)$                              & 24 &        & $\gLL$ &  1 &  1\\
\end{tabular}
\label{table:12}
\end{table}

Finally, by \autoref{cor:holhyp} we know that:
\begin{prop}
The bulk genus $(\Vect,24)$ decomposes into exactly 12 hyperbolic genera (ignoring ``fake'' copies of the moonshine module), each labeled by its Heisenberg commutant, in turn labeled by a letter A to L.
\end{prop}

We explain the entries of \autoref{table:12}. The first column contains the name of the hyperbolic genus, i.e.\ the Heisenberg commutant. The next two columns describe the bulk genus of the Heisenberg commutant, followed by the element $\nu\in\Aut(\Lambda)$ defining it via \eqref{eqn:SchellekensHeisenbergCommutant}. The fifth column lists the corresponding lattice genus, followed by the number of lattices in the genus. In the last column, we list the number of holomorphic \voa{}s in the hyperbolic genus.

We remark that for most of the $12$ hyperbolic genera, the number of \voa{}s coincides with the number of lattices in the lattice genus $\II_{c-c(C),0}(\overline{R_C})$, i.e.\ for each lattice $L$ in that genus there is exactly one Schellekens \voa{} with that associated lattice, i.e.\ with the dual pair $C\otimes V_L$. In other words, the choice of anti-isometry $\tau\colon L'/L\to R_C$ described at the beginning of this section is often irrelevant in the holomorphic case. Only for genera D and J, do we need to take care of the anti-isometry.

As explained in \autoref{ex:fake}, each ``fake'' copy of the moonshine module, i.e.\ a \strat{}, holomorphic \voa{} $V$ of central charge~$24$ with $V_1=\{0\}$ and $V\ncong V^\natural$, must define its own hyperbolic genus (whose entry in \autoref{table:12} would look identical to the one for L).


\subsection{Examples with Central Charge 32}\label{subsec:c=32}

We commented in \autoref{subsec:holobulk} that there are over a billion \strat{}, holomorphic vertex operator algebras with central charge $c=32$. That estimate came from counting lattice vertex operator algebras \cite{Kin03}. However, all these holomorphic vertex operator algebras based on even, unimodular lattices belong to the same hyperbolic genus, and hence those order billion theories all collapse to a single hyperbolic equivalence class, namely $(\C\vac,32)_{\Vect}$ with trivial Heisenberg commutant $C\cong\C\vac$.

In this spirit, one might hope that the classification of $c=32$, \strat{}, holomorphic vertex operator algebras, i.e.\ the bulk genus $(\Vect,32)$, becomes tractable, provided one is content with classifying theories only up to hyperbolic equivalence. In that case, one is lead to ask what other choices of $C$ lead to a valid hyperbolic genus of the form $(C,32)_{\Vect}$. Although we do not attempt to be complete in the present paper, we offer a few examples.

\medskip

In general, for any central charge, \autoref{prop:hypgenhol} tells us that in order to obtain candidates for Heisenberg commutants inside holomorphic \voa{}s, we need to produce \strat{} \voa{}s $C$ with two main properties: $C$ must have a pointed representation category and $C_1=\{0\}$. In the following, inspired by the success of this approach for central charge $c=24$, we shall construct such \voa{}s $C$ as \fpvosa{}s $V^G$.

A large class of \voa{}s $V^G$ that is expected to be \strat{} with $\Rep(V^G)$ pointed is when $V$ itself is \strat{} and has a pointed representation category and when $G<\Aut(V)$ is a finite, cyclic (or, under some additional conditions, abelian) group whose corresponding permutation action on the irreducible $h$-twisted $V$-modules is trivial for all $h\in G$.
Moreover, when $G$ is cyclic, in order for $V^G_1$ to be trivial, the finite-dimensional Lie algebra $V_1$ must be abelian. Indeed, $V_1$ must be reductive and if $V_1$ had a non-trivial semisimple part, $V^G_1$ could not be zero \cite{Kac90}.

We consider two types of examples.
\begin{enumerate}[wide]
\item We let $V=V_L$ be a lattice \voa{}, which has the pointed representation category $\Rep(V_L)\cong\mathcal{C}(L'/L)$, and consider a cyclic group $G=\langle\hat\nu\rangle$ of automorphisms, where $\hat\nu$ is some lift of a lattice isometry $\nu\in\Aut(L)$. Then it is shown in \cite{Lam20} that $\Rep(V_L^{\hat\nu})$ is pointed if and only if the induced action of $\nu$ on $L'/L$ is trivial. Moreover, $(V_L^{\hat\nu})_1$ is trivial if and only if $L$ has no norm-2 vectors and $\nu$ acts fixed-point freely on $(V_L)_1\cong L\otimes_\Z\C$.

If $\nu\in\Aut(L)$ is fixed-point free and acts trivially on $L'/L$, then $L$ can be primitively embedded into a unimodular lattice $N$ (of possibly larger rank) and $\nu$ lifts to an isometry of $N$ such that $L=N_\nu=(N^\nu)^\bot$ is the coinvariant lattice \cite{Lam20}. Hence, we are looking at lattice orbifolds of the form $V_{N_\nu}^{\hat{\nu}}$ for some unimodular lattice $N$ such that the coinvariant lattice $N_\nu$ (but not necessarily $N$ itself) has no norm-$2$ vectors.
\item We let $V$ be a \strat{}, holomorphic \voa{} with $V_1$ abelian (or even zero) and consider a cyclic (or more generally abelian) group $G=\langle g\rangle$ of automorphisms. Then $\Rep(V^g)$ is a twisted Drinfeld double $\mathcal{D}_\omega(G)$ \cite{EMS20a,DNR21b}, which is often pointed (e.g., when $G$ is cyclic) and we demand that $g$ act fixed-point freely on $V_1$.
\end{enumerate}

\subsubsection*{From the Leech Lattice}
The Leech lattice $\Lambda$ is the unique positive-definite, even, unimodular lattice of rank~$24$ without norm-$2$ vectors. We consider potential Heisenberg commutants of the form $C\cong V_{\Lambda_\nu}^{\hat\nu}$ for some automorphism $\nu\in\Aut(\Lambda)$, as in \eqref{eqn:SchellekensHeisenbergCommutant}. As just explained, these \voa{}s are \strat{} (and also positive), have a pointed (and pseudo-unitary) representation category and satisfy $C_1=\{0\}$. It remains to verify that (3) of \autoref{prop:hypgenhol} holds.

By \autoref{cor:increasingc}, any of the $11$ Heisenberg commutants of this form already appearing in the Schellekens \voa{}s (see \autoref{table:12}) will still be the Heisenberg commutant of a hyperbolic genus $(C,32)_{\Vect}$ of a \strat{}, holomorphic \voa{}. But, as the corresponding lattice genera in \autoref{prop:hypgenhol} have a rank that is 8 larger, it is much easier for them to exist. We therefore expect that most (if not all) of the automorphisms $\nu\in\Aut(\Lambda)$ yield a hyperbolic genus $(V_{\Lambda_\nu}^{\hat\nu},32)_{\Vect}$ in the bulk genus $(\Vect,32)$.

For simplicity, let us assume that $\nu\in\Aut(\Lambda)$ does not exhibit order doubling when lifted to an automorphism of $V_\Lambda$. This holds for $152$ of the $160$ algebraic conjugacy classes (i.e.\ conjugacy classes of cyclic subgroups) in $\Co_0\cong\Aut(\Lambda)$. Then, as we described in \autoref{sec:holhyp}, the representation category of $V_{\Lambda_\nu}^{\hat\nu}$ has the simple form \cite{Lam20}
\begin{equation*}
\Rep(C)\cong\mathcal{C}(R_C)\quad\text{with}\quad R_C\cong\Lambda'_\nu/\Lambda_\nu\times H,
\end{equation*}
where $H$ is the metric group describing the twisted Drinfeld double $\mathcal{D}_\omega(\Z_n)\cong\mathcal{C}(H)$ for $n=|\nu|$ and $\omega$ determined by the conformal weight of the unique irreducible $\hat\nu$-twisted $V_\Lambda$-module (called type in \cite{EMS20a}).

It is not difficult to verify that, indeed, in all $152$ cases, there is a (non-empty) lattice genus
\begin{equation*}
\II_{32-c(C),0}(\overline{R_C})=\II_{8+\rk(\Lambda^\nu),0}(\overline{R_C}).
\end{equation*}
This follows from \cite{Nik80} (see \autoref{sec:lat} for details), as the rank of the lattice $32-c(C)=8+\rk(\Lambda^\nu)$ turns out to be always greater than the minimal number $l(R_C)$ of generators of the abelian group $R_C$.

Overall, this yields $152$ (likely $160$) \voa{}s that appear as Heisenberg commutants in the bulk genus $(\Vect,32)$. We expect many of these \voa{}s to be non-isomorphic.
We list some simple examples:
\begin{equation*}
\begin{tabular}{l|r|c||l|r||c}
$\Rep(C)$ & $c$ & $\nu\in\Aut(\Lambda)$ & Lattice genus & No.\ &  $|(C,32)_{\Vect}|$\ \\\hline\hline
$\mathcal{C}(3^{+7})$ & 18 & $1^{-3}3^9$ & $\II_{14,0}(3^{-7})$ & 29 & $\geq29$ \\
$\mathcal{C}(5^{+5})$ & 20 & $1^{-1}5^5$ & $\II_{12,0}(5^{+5})$ & 41 & $\geq41$ \\
$\mathcal{C}(7^{-2})$ & 24 & $1^{-4}7^4$ & $\II_{8,0}(7^{-2})$  &  3 & $\geq3$
\end{tabular}
\end{equation*}
We remark that in this construction the Heisenberg commutants $C$ with central charge $c=24$ correspond exactly to the fixed-point free automorphisms $\nu\in\Aut(\Lambda)$, in which case $\Lambda_\nu=\Lambda$ and the \voa{} $C$ is simply a cyclic orbifold of the holomorphic Leech lattice \voa{}. Some but not all of these will reappear below when considering cyclic orbifolds of the moonshine module (the precise criterion involves the vacuum anomaly $\rho_\nu$; see, e.g., \cite{MS23}).

\subsubsection*{From Extremal Lattices in Rank~32}

We can also apply the method we just described to even, unimodular lattices $N$ of rank~$32$. We restrict to unimodular lattices $N$ that already themselves have the property of possessing no norm-$2$ vectors.
It is shown in \cite{Kin03} that there are at least ten million such lattices.
Hence, we forego a more systematic description, and rather content ourselves with investigating an example.

The $32$-dimensional Barnes--Wall lattice $N=\operatorname{BW}_{32}$ is a positive-definite, even, unimodular lattice without norm-$2$ vectors \cite{BW59}. Its isometry group $\Aut(N)$ can be computed in \texttt{Magma} \cite{Magma} and has $399$ conjugacy classes.

We consider potential Heisenberg commutants $C=V_{N_\nu}^{\hat\nu}$ for $\nu\in\Aut(N)$, now with the Barnes--Wall lattice $N=\operatorname{BW}_{32}$ replacing the Leech lattice. We again restrict for simplicity to the cases where $\nu$ does not have order-doubling on $N$ so that $\Rep(C)\cong\mathcal{C}(R_C)$ with the metric group $R_C\cong N'_\nu/N_\nu\times H$ where $H$ is the metric group associated with some twisted Drinfeld double $\mathcal{D}_\omega(\Z_n)\cong\mathcal{C}(H)$ where $n=|\nu|$. It remains to check (3) of \autoref{prop:hypgenhol}, which is now rather restrictive in comparison. We need to verify that
\begin{equation*}
\II_{32-c(C),0}(\overline{R_C})=\II_{\rk(N^\nu),0}(\overline{R_C})
\end{equation*}
corresponds to a (non-empty) lattice genus.

For $18$ conjugacy classes in $\Aut(N)$, the rank of the lattice $32-c(C)=\rk(N^\nu)$ is greater than the minimal number $l(R_C)$ of generators of the abelian group $R_C$, so that the above lattice genus exists \cite{Nik80}. (There are a few edge cases with the rank being equal to $l(R_C)$, which we would have to examine more closely.) We also (for the moment) ignore those Heisenberg commutants $C$ obtained in this way that have the same central charge and representation category as a Heisenberg commutant obtained from the Leech lattice (as they could potentially be isomorphic). This leaves us with at least ten new and non-isomorphic Heisenberg commutants appearing in the bulk genus $(\Vect,32)$. We list some simple examples:
\begin{equation*}
\begin{tabular}{l|r|c||l|c||c}
$\Rep(C)$ & $c$ & $\nu\in\Aut(N)$ & Lattice genus & No.\ &  $|(C,32)_{\Vect}|$\ \\\hline\hline
$\mathcal{C}(7^{-6})$ & 24 & $1^47^4$ & $\II_{8,0}(7^{-6})$ & 3 & $\geq3$ \\
$\mathcal{C}(2_{\II}^{+4}7^{+5})$ & 26 & $1^22^17^214^1$ & $\II_{6,0}(2_{\II}^{+4}7^{-5})$  &  1 & $\geq1$
\end{tabular}
\end{equation*}

\subsubsection*{From the Moonshine Module}

Another source of hyperbolic genera $(C,32)_{\Vect}$ is the moonshine module. That is, we consider Heisenberg commutants of the form $(V^\natural)^G$ for some subgroup $G<\Aut(V^\natural)=\mathbb{M}$ of the monster group. Then $(V^\natural)_1^G=\{0\}$, and the representation category of $(V^\natural)^G$, a certain twisted Drinfeld double $\mathcal{D}_\omega(G)$ for some $3$-cocycle $\omega$ on $G$ \cite{DNR21b}, is pointed if and only if $G$ is abelian and certain auxiliary $2$-cocycles are all coboundaries \cite{CGR00}.
This condition is trivially satisfied if $G$ is cyclic. In that case, the twisted Drinfeld double $\mathcal{D}_\omega(G)$ is described, e.g., in \cite{Moe16,EMS20a} together with a formula for the class of the $3$-cocycle $\omega$ (called type).

Suppose for instance that $G$ is an abelian group with at most three generators and that the class $[\omega]\in H^3(G,\C^\times)$ is trivial. Then $\mathcal{D}_\omega(G)\cong\mathcal{C}(H)$ is pointed and the corresponding metric group $H$ has at most six generators. Then, the criterion from \cite{Nik80} (see also \autoref{sec:lat}) guarantees that a positive-definite, even lattice of rank $8$ with discriminant form $L'/L\cong\overline{H}$ always exists. Thus, for each such conjugacy class of subgroup $G<\mathbb{M}$, we obtain a hyperbolic genus $((V^\natural)^G,32)_{\Vect}$.

\medskip

We further specialize to the case of cyclic $G<\mathbb{M}$. Then, the representation category of $(V^\natural)^G$ is always pointed, regardless of the value of $\omega$. Hence, for each of the $194$ conjugacy classes $g\in\mathbb{M}$ we obtain a hyperbolic genus $((V^\natural)^g,32)_{\Vect}$.

However, not all of these Heisenberg commutants are new. Indeed, $53$ of the conjugacy classes are non-Fricke ($51$ up to algebraic conjugacy). For those, it was shown in \cite{Car18} (see also \cite{Tui92,Tui95}) that $(V^\natural)^g\cong V_\Lambda^{\hat\nu}$ for some lift of an element $\nu\in\Co_0=\Aut(\Lambda)$, which we treated above.

Hence, it suffices to consider Fricke elements $g\in \mathbb{M}$, i.e.\ the remaining $141$ conjugacy classes in $\mathbb{M}$. For these, it was shown in \cite{PPV16} that the only holomorphic extension of $(V^\natural)^g$ is again $V^\natural$. Therefore, these certainly produce new Heisenberg commutants, though some of them could be isomorphic to one another. Once again, we only list some simple examples (of Fricke type), labeling the elements of the monster group $\mathbb{M}$ as in \cite{CCNPW85}:
\begin{equation*}
\begin{tabular}{l|r|c||l|r||c}
$\Rep(C)$ & $c$ & $g\in\mathbb{M}$ & Lattice genus & No.\ &  $|(C,32)_{\Vect}|$ \\\hline\hline
$\mathcal{C}(1)$ & 24 & 1A & $\II_{8,0}(1)$ & 1 & 1\\
$\mathcal{C}(2_{\II}^{+2})$ & 24 & 2A & $\II_{8,0}(2_{\II}^{+2})$ & 1 & 1\\
$\mathcal{C}(3^{-2})$ & 24 & 3A & $\II_{8,0}(3^{-2})$ & 1 & $\geq1$\\
$\mathcal{C}(9^{+1})$ & 24 & 3C & $\II_{8,0}(9^{-1})$ & 2 & $\geq2$\\
\end{tabular}
\end{equation*}
Examples of lattices in the corresponding lattice genera are the root lattices $E_8$, $D_8$, $A_2E_6$ and $A_8$, respectively.

\begin{rem}
With the examples coming from the Leech lattice \voa{} $V_\Lambda$ and the Moonshine module $V^\natural$ (and assuming the uniqueness of the moonshine module), we have, in fact, classified all \strat{} \voa{}s $C$ of central charge $c=24$ with $C_1=\{0\}$ and whose representation categories are twisted Drinfeld doubles $\mathcal{D}_\omega(\Z_n)$ of cyclic groups.

Indeed, we know that they must have either $V_\Lambda$ or $V^\natural$ (or both) as $\Z_n$-extensions (or else $C_1$ could not be trivial), i.e.\ they are exactly of the form $C\cong(V^\natural)^g$ for $g\in\mathbb{M}$ or $C\cong V_\Lambda^{\hat\nu}$ for $\nu\in\Co_0\cong\Aut(\Lambda)$ fixed-point free. These two constructions overlap exactly for the automorphisms considered in \cite{Car18}, i.e.\ when $g$ is non-Fricke (and a corresponding condition on $\nu$).

We have further argued that all of these \voa{}s $C$ are Heisenberg commutants of the bulk genus $(\Vect,32)_{\Vect}$ and therefore of all bulk genera $(\Vect,32+8k)_{\Vect}$ with $k\in\N$.
\end{rem}

We discuss some examples in more detail. The \voa{}s $(V^\natural)^g$ are often related to sporadic groups. For instance, if we suppose $G=\Z_2$ and that $\omega$ is in the trivial class, by bosonizing the results of \cite{HM23}, it follows that there are exactly two \strat{} \voa{}s $C$ with $C_1=\{0\}$, central charge $c(C)=24$ and $\Rep(C)\cong\mathcal{D}(\Z_2)$ (assuming the moonshine uniqueness conjecture \autoref{conj:moonshineuniqueness}), corresponding to the elements $2\mathrm{A}$ (Fricke) and $2\mathrm{B}$ (non-Fricke) in $\mathbb{M}$. They are 
\begin{align*}
(V^\natural)^{2\mathrm{A}}&\cong \textsl{V}\mathbb{B}^\natural \otimes  L(\sfrac12) \oplus \textsl{V}\mathbb{B}^\natural(\sfrac32)\otimes L(\sfrac12,\sfrac12), \\
(V^\natural)^{2\mathrm{B}} &\cong V_{\Lambda}^+,
\end{align*}
where $\textsl{V}\mathbb{B}^\natural$ is the baby monster \voa{} \cite{Hoe95} and $V_\Lambda^+$ is the charge-conjugation (or $(-1)$-involution) orbifold of the Leech lattice \voa{} $V_\Lambda$.\footnote{More precisely, $\textsl{V}\mathbb{B}^\natural$ is the bosonic subalgebra of the baby monster vertex operator superalgebra studied in \cite{Hoe95}.}
On the other hand, the genus of positive-definite, even, rank-$8$ lattices $L$ with $\mathcal{C}(L'/L)=\mathcal{C}(2_{\II}^{+2}) = \overline{\mathcal{D}(\Z_2)} \cong \mathcal{D}(\Z_2)$ is non-empty, with the root lattice $D_8$ furnishing the only example. So, we find two hyperbolic genera of $c=32$, \strat{}, holomorphic \voa{}s, each containing only one \voa{}, namely $\mathrm{Ex}((V^\natural)^{2\mathrm{X}}\otimes \mathsf{D}_{8,1})$ for $\mathrm{X}=\mathrm{A},\mathrm{B}$.
(See \cite{Hoe17,HM23} for how to enumerate non-isomorphic simple-current extensions.)

\medskip

As another similar example, the genus of positive-definite, even, dimension-$8$ lattices $L$ with $\mathcal{C}(L'/L)=\mathcal{C}(3^{-2})=\overline{\mathcal{D}(\Z_3)}\cong \mathcal{D}(\Z_3)$ is nonempty, with the root lattice $A_2E_6$ giving the only example. There are also exactly two \strat{} \voa{}s $C$ with $C_1=0$, $c(C)=24$ and $\Rep(C) = \mathcal{D}(\Z_3)$, namely those corresponding to the elements $3\mathrm{A}$ (Fricke) and $3\mathrm{B}$ (non-Fricke) in $\mathbb{M}$:
\begin{align*}
(V^\natural)^{\mathrm{3A}}&\cong \textsl{VF}_{24}^\natural \otimes \mathcal{P}(3) \oplus \textsl{VF}_{24}^\natural(\sfrac85)\otimes \mathcal{P}(3,[2,0]),\\
(V^\natural)^{3\mathrm{B}}&\cong V_{\Lambda}^{\hat\nu},
\end{align*}
where $\mathcal{P}(3) = L(\sfrac45)\oplus L(\sfrac45,3)$ is the $\Z_3$-parafermion \voa{} (see \autoref{ex:para}), $\textsl{VF}_{24}^\natural$ is the \voa{} with $\textsl{Fi}_{24}$ symmetry \cite{HLY12} and $\nu$ denotes the automorphism of Frame shape $1^{-12}3^{12}$ in $\Aut(\Lambda)$. Thus, we obtain two hyperbolic genera $((V^\natural)^{3\mathrm{X}},32)_{\Vect}$ for $\mathrm{X}=\mathrm{A},\mathrm{B}$ with corresponding holomorphic \voa{}s $\mathrm{Ex}((V^\natural)^{3\mathrm{X}}\otimes \mathsf{A}_{2,1}\mathsf{E}_{6,1})$.
Probably, as for the order-$2$ examples above, each hyperbolic genus only contains one \strat{} \voa{}.

There is one more conjugacy class of order-$3$ elements in the monster $\mathbb{M}$, the 3C conjugacy class. The corresponding fixed-point subalgebra $(V^\natural)^{3\mathrm{C}}$ has the representation category $\mathcal{C}(9^{+1})$, which is a non-trivially twisted Drinfeld double for $\Z_3$. It is a conformal extension of the form $(V^\natural)^{3\mathrm{C}}=\mathrm{Ex}(\textsl{VT}^\natural \otimes \mathcal{P}(9))$, where $\mathcal{P}(9)$ is the $\Z_9$-parafermion \voa{}, and $\textsl{VT}^\natural$ is a \voa{} defined in \cite{BHLLR21} whose automorphism group contains (and is conjecturally precisely) the Thompson sporadic group. In more detail,  
\begin{align*}
    (V^\natural)^{3\mathrm{C}}&=\textsl{VT}^\natural\otimes\mathcal{P}(9)\oplus \textsl{VT}^\natural(\sfrac{20}{11})\otimes \mathcal{P}(9,[2,0])\oplus \textsl{VT}^\natural(\sfrac{16}{11})\otimes \mathcal{P}(9,[4,0])  \\
    & \quad  \oplus \textsl{VT}^\natural(\sfrac{21}{11})\otimes \mathcal{P}(9,[6,0]) \oplus \textsl{VT}^\natural(\sfrac{13}{11})\otimes \mathcal{P}(9,[8,0]).
\end{align*}
This \voa{} also defines a hyperbolic genus $((V^\natural)^{3\mathrm{C}},32)_{\Vect}$. The corresponding rank-$8$ lattices must have the discriminant form $9^{-1}=\overline{9^{+1}}$. There are exactly two lattices in this genus, among them the root lattice $A_8$. Hence, the hyperbolic genus $((V^\natural)^{3\mathrm{C}},32)_{\Vect}$ contains at least two holomorphic \voa{}s.

\medskip

We briefly mention one more example related to sporadic groups. Namely, there is the hyperbolic genus $((V^\natural)^{5\mathrm{A}},32)_{\Vect}$ for the conjugacy class 5A in $\mathbb{M}$ of Fricke type. The Heisenberg commutant $(V^\natural)^{5\mathrm{A}}\cong \mathrm{Ex}(\textsl{VHN}^\natural\otimes \mathcal{P}(5)^{\otimes2})$ is a certain conformal extension of two $\Z_5$-parafermion theories and a \voa{} $\textsl{VHN}^\natural$ whose automorphism group contains the Harada--Norton sporadic group. The precise extension can be deduced from the information given in \cite{BHLLR21}.


\subsection{Lower Bounds from Mass Formulae}\label{subsec:holomass}

It is interesting to see how the mass formula discussed in \autoref{subsec:mass} works in the context of \strat{}, holomorphic \voa{}s. Indeed, using \autoref{cor:massconstant}, we may compute the masses of all hyperbolic genera of the Schellekens \voa{}s.
Table~3 of \cite{BLS23} contains the quantities $[{\Aut}(\Rep(C)):\overline{\Aut}(C)]$ for all hyperbolic genera of the form $(C,24)_{\Vect}$. Similarly, Table~4 of \cite{Hoe17} contains the lattice masses $\mass(L)$ for each $L$ arising as the associated lattice of a $c=24$, \strat{}, holomorphic \voa{}. By \autoref{thm:massformula} \cite{Mor21}, the product of these two numbers yields the \voa{} masses of each of the hyperbolic genera $(C,24)_{\Vect}$.

We summarize these results in \autoref{tab:Schellekensgenera} (cf.\ \autoref{table:12}). The first column gives the name of the hyperbolic genus, the second the mass of the associated lattice genus. The last column gives the constant of proportionality between the lattice mass and the \voa{} mass (cf.\ \autoref{thm:massformula}). In other words, the mass of the hyperbolic genus is the product of the entries in the last two columns.

\begin{table}[ht]
\caption{Masses of hyperbolic genera of $c=24$, \strat{}, holomorphic \voa{}s (data from \cite{Hoe17,BLS23}).}
\begin{tabular}{cHHH|c|c}
Name & \# & $\rk(L)$ & $L'/L$ & $\mass(L)$ & $[\Aut(\tilde{L}):G_{\tilde{V}}]$ \\\hline \hline
A & 24 & 24 & $1$ & $\frac{131\cdot 283\cdot 593\cdot 617\cdot 691^2\cdot 3617\cdot 43867}{2^45\cdot 3^17\cdot 5^7\cdot 7^4\cdot 11^2\cdot 13^2\cdot 17\cdot 19\cdot 23}$ & $1$ \\
B & $17$ & $16$ & $2_{\II}^{+10}$ & $\frac{17\cdot 43\cdot 127\cdot 691}{2^{31}\cdot 3^7\cdot 5^3\cdot 7^2}$ & $1$ \\
C & $6$ & $12$ & $3^{-8}$ & $\frac{11\cdot 13 \cdot 61}{2^{16}\cdot 3^8\cdot 5^2 \cdot 7}$ & $2$ \\
D & $2$ & $12$ & $2_{\II}^{-10}4_{\II}^{-2}$ & $\frac{31}{2^{22}\cdot 3^7 \cdot 5^2 \cdot 7\cdot 11}$ & $2^{11}\cdot 3 \cdot 17$ \\
E & $5$ & $10$ & $2_2^{+2}4_{\II}^{+6}$ & $\frac{17}{2^{12}\cdot 3^5\cdot 7}$ & $2$ \\
F & $2$ & $8$ & $5^{+6}$ & $\frac{13}{2^{12}\cdot3^2\cdot 5^2}$ & $2$ \\
G & $2$ & $8$ & $2_{\II}^{+6}3^{-6}$ & $\frac{1}{2^{11}3^2}$ & $1$ \\
H & $1$ & $6$ & $7^{-5}$ & $\frac{1}{2^5\cdot 3^2\cdot 5 \cdot 7}$ & $2$ \\
I & $1$ & $6$ & $2_5^{-1}4_1^{+1}8_{\II}^{-4}$ &$\frac{1}{2^9\cdot 3\cdot 5}$ &$2$ \\
J & $1$ & $6$ & $2_{\II}^{+4}4_{\II}^{-2}3^{+5}$ & $\frac{1}{2^9\cdot 3^3}$ & $4$ \\
K & $1$ & $4$ & $2_{\II}^{-2}4_{\II}^{-2}5^{+4}$ & $\frac{1}{2^7 \cdot 3}$ & $3$ \\
L & $1$ & $0$ & $1$ & $1$ & $1$
\end{tabular}
\label{tab:Schellekensgenera}
\end{table}

\medskip

\autoref{cor:massconstant} also allows us to perform more explicit checks of \autoref{thm:massformula}. Indeed, we can explicitly perform the sum over \voa{}s in a hyperbolic genus by rewriting it as
\begin{equation}\label{eqn:massrewriting}
\sum_{V\in \hgen(W)} \frac{1}{|G_{V}|} = \sum_{V\in\hgen(W)} \frac{|\overline{\Aut}(L)|}{|\overline{\Aut}(L)\cap\tau^\ast \overline{\Aut}(C)|}\frac{1}{|\Aut(L)|}
\end{equation}
and confirm that the result is the correct integer multiple of the corresponding lattice mass. Fortunately, all the groups appearing on the right-hand side have been computed in \cite{BLS23}. We report this data for the hyperbolic genus D in \autoref{tab:genusD}. The interested reader can insert this data into equation~\eqref{eqn:massrewriting} and confirm that it reproduces the mass reported in \autoref{tab:Schellekensgenera} computed using \autoref{thm:massformula}.

\begin{table}[ht]
\caption{\Strat{}, holomorphic \voa{}s in the hyperbolic genus~D and quantities related to the mass formula in the form of \eqref{eqn:massrewriting}.}
\begin{tabular}{r|l|c|c|c|c}
No. & $V_1$ & $\overline{\Aut}(L)\cap \tau^\ast \overline{\Aut}(C)$ & $L$ & $\Aut(L)$ & $\overline{\Aut}(L)$ \\\hline\hline
2 & $\mathsf{A}_{1,4}^{12}$ & $2^{12}.M_{12}$ & $\sqrt{2}D_{12}$ & $2^{12}.\mathfrak{S}_{12}$ & $2^{12}.\mathfrak{S}_{12}$ \\ 
12 & $\mathsf{B}_{2,2}^6$ & $2^{12}.(\Z_2^6:\mathfrak{S}_5)$ & & \\ 
23 & $\mathsf{B}_{3,2}^4$ & $2^{12}.(\mathfrak{S}_3\wr \mathfrak{A}_4)$ & &  \\
29 & $\mathsf{B}_{4,2}^3$ & $2^{12}(\mathfrak{S}_4 \wr \mathfrak{S}_3)$ & & \\ 
41 & $\mathsf{B}_{6,2}^2$ & $2^{12}.(\mathfrak{S}_6\wr 2)$ & & \\ 
57 & $\mathsf{B}_{12,2}$ & $2^{12}.\mathfrak{S}_{12}$ & &  \\\hline 
13 & $\mathsf{D}_{4,4}\mathsf{A}_{2,2}^4$ & $2^{1+4}.(\mathfrak{S}_3\wr \mathfrak{S}_4).\mathfrak{S}_3$ & $\sqrt{2}E_8\sqrt{2}D_4$ & $G_1$ & $G_2$ \\
$22$ & $\mathsf{C}_{4,2}\mathsf{A}_{4,2}^2$ & $2^{1+4}.(\mathfrak{S}_2\times \mathfrak{S}_5\wr\mathfrak{S}_2).\mathfrak{S}_3$ & & \\
36 & $\mathsf{A}_{8,2}\mathsf{F}_{4,2}$ & $2^{1+4}.(\mathfrak{S}_3\times\mathfrak{S}_9).\mathfrak{S}_3$ & &
\end{tabular}
\label{tab:genusD}
\end{table}

We explain the notation of \autoref{tab:genusD}. The first column refers to the numbering given in \cite{Sch93}. The second column indicates the affine structure of the theory. The third column gives one of the groups necessary for using \autoref{cor:massconstant}. The fourth column gives the associated lattice, the fifth its automorphism group and the sixth its image under the natural map $\mu_L$. Finally, $G_1=((2^{1+4}.\mathfrak{S}_3):\mathfrak{S}_3)\times 2.\mathrm{GO}_8^+(2)$ and $G_2=((2^{1+4}.\mathfrak{S}_3):\mathfrak{S}_3)\times \mathrm{GO}_8^+(2)$.

\medskip

Each of the $12$ Heisenberg commutants $C$ for which $(C,24)_{\Vect}$ is a non-empty hyperbolic genus will lead to non-empty hyperbolic genera $(C,8n)_{\Vect}$ for $n\geq 3$, by \autoref{cor:increasingc}. Using the mass formula, we may obtain lower bounds for the number of \voa{}s in the hyperbolic genus $(C,8n)_{\Vect}$ as a function of $n$, and estimate how the bound grows at asymptotically large central charge. 

We illustrate this for the Heisenberg commutant of the hyperbolic genus H. The only \voa{} in this genus in the Schellekens \voa{} $\mathbf{S}(\mathsf{A}_{6,7})$. Let $C$ denote its Heisenberg commutant so that the hyperbolic genus is $\mathrm{H}=(C,24)_{\Vect}$. We may obtain many more holomorphic \voa{}s by considering the hyperbolic genera $\mathrm{H}_m\coloneqq(C,24+8m)_{\Vect}$ for $m\geq 0$.

The associated lattice $L$ of $\mathbf{S}(\mathsf{A}_{6,7})$ is the unique even lattice in the genus $\gHH$.
Using standard techniques \cite{CS88}, one may compute the mass of the lattice genus $\II_{6+8m,0}(7^{-5})=\gen(L\oplus (E_8)^m)$, which contains the associated lattices of the hyperbolic \voa{} genus $\mathrm{H}_m$, and which therefore by \autoref{thm:massformula} and \autoref{tab:Schellekensgenera} is half of the mass of $\mathrm{H}_m$. Setting $n\coloneqq 6+8m$ and $d\coloneqq 7^5$, one finds 
\begin{equation*}
    \mass(L\oplus (E_8)^m) = 2^{-n}7^{\frac{5(n-5)}{2}}M_7(n-5)M_7(5)\mathrm{std}(n,d) \mathrm{std}_7(n,d)^{-1}.
\end{equation*}
Here, $\mathrm{std}(n,d)$ is the so-called standard mass,
\begin{equation*}
    \mathrm{std}(n,d) = \begin{cases}
        {\displaystyle 2\pi^{-\frac{n(n+1)}{4}}\Biggl(\prod_{j=1}^{n}\Gamma(\tfrac{j}2)\Biggr)\Biggl(\prod_{j=1}^{(n-1)/2}\!\!\zeta(2j)\Biggr)},& n \text{ odd},\\
        \noalign{\vskip3pt}
        {\displaystyle 2\pi^{-\frac{n(n+1)}{4}}\Biggl(\prod_{j=1}^{n}\Gamma(\tfrac{j}2)\Biggr)\Biggl(\prod_{j=1}^{n/2-1}\zeta(2j)    \Biggr)\zeta_{(-1)^{n/2}d}(n/2)}, & n \text{ even},
    \end{cases}
\end{equation*}
which is expressed in terms of the gamma function $\Gamma(z)$, the Riemann zeta function $\zeta(s)$ and a Dirichlet series
\begin{equation*}
    \zeta_D(s) = \prod_{p\text{ prime}}\left(1-\left(\frac{D}{p}\right)\frac{1}{p^s}\right)^{-1},
\end{equation*}
where $\bigl(\frac{D}{p}\bigr)$ is a Kronecker symbol. We have also used the standard $p$-mass, which is defined as
\begin{equation*}
    \mathrm{std}_p(n,d) = \begin{cases}
        \frac12 \left( (1-p^{-2})(1-p^{-4})\dots(1-p^{2-n})(1-\epsilon p^{-n/2})  \right)^{-1} ,& n \text{ even},\\
        \frac12 \left((1-p^{-2})(1-p^{-4})\dots(1-p^{1-n})\right)^{-1},& n \text{ odd},
    \end{cases}
\end{equation*}
where $\epsilon$ is defined as the Kronecker symbol 
\begin{equation*}
    \epsilon = \left(\frac{(-1)^{n/2}d}{p}\right).
\end{equation*}
Finally, let $p$ be an odd prime. If $n$ is odd, then we define
\begin{equation*}
    M_p(n) =
    \frac12\bigl( (1-p^{-2})(1-p^{-4})\dots(1-p^{1-n})\bigr)^{-1},
\end{equation*}
and if $n$ is even, we define 
\begin{equation*}
    M^\pm_p(n) = \frac12 \bigl( (1-p^{-2})(1-p^{-4})\dots(1-p^{2-n})(1\mp p^{-n/2}) \bigr)^{-1}.
\end{equation*}
Using this formula for the lattice mass, we find for small values of $m$ that 
\begin{align*}
    \mass(\mathrm{H}_1) &= \frac{3990690204149401}{7023034368000}\approx 568.2,\\
    \mass(\mathrm{H}_2) &=\frac{117487377559108737529255642264988215913241487}{4325549838586557235200000}\approx 2.7\times 10^{19},\\
    \mass(\mathrm{H}_3)&\approx 1.1\times 10^{53}.
\end{align*}
Recalling that the mass furnishes a lower bound for the number of \voa{}s in the corresponding hyperbolic genus, one sees that the number of \voa{}s grows quite rapidly with central charge. Indeed, using standard estimates for the functions appearing in the mass formula, we find an asymptotic expansion of the form 
\begin{equation*}
    \log\left(\mass(\mathrm{H}_m)\right) \sim  \frac14 n^2 \left(\log n + a \right)+ \frac14 n \left( - \log n+b  \right)+\frac1{24}\left(-\log n +c\right)+\dots
\end{equation*}
where the parameters are 
\begin{align*}
    a=-\frac32 -\log(2\pi),\;\;
    b=1+\log\left(\frac{282475249\pi}{2}\right),\;\;
    c= 1+ \log \left(\frac{2^{18}\zeta_\infty^{24}   M_7(5)^{24}}{7^{300}A^{12}}\right)
\end{align*}
with $\zeta_\infty \coloneqq \prod_{j=1}^\infty \zeta(2j)\approx 1.82102$ and the Glaisher--Kinkelin constant $A\approx 1.28243$. 


\subsection{Siegel--Weil Identity and Ensemble-Averaged Holography}\label{subsec:SiegelWeil}

The mass formula explored in the previous section is a prescription for computing a particular weighted sum over vertex operator algebras in a fixed hyperbolic genus. It is interesting to ask to what extent this result can be generalized.

For example, in the lattice context, it is possible to compute the weighted sum of the theta functions of the lattices in a fixed genus. Recall that the theta function of a positive-definite lattice $L$ is defined as 
\begin{equation*}
    \Theta_L(\tau) = \sum_{\lambda \in L}q^{\langle\lambda,\lambda\rangle/2 }.
\end{equation*}
In the case of positive-definite, even, unimodular lattices of dimension $8n$, the weighted sum over theta functions is determined by the following Siegel--Weil identity (see, e.g., equation~(103) of \cite{Ser73}),
\begin{equation}\label{eqn:siegelweillattice}
    \sum_{L\in\mathrm{gen}(E_8^n)} \frac{\Theta_L(\tau)}{|\Aut(L)|} = M_{n} E_{4n}(\tau),
\end{equation}
where $E_k(\tau)$ is the weight $k$ Eisenstein series for $\mathrm{SL}_2(\Z)$ with constant term normalized to $1$ and $M_n\coloneqq\mathrm{mass}((E_8)^n)$ is the mass of the positive-definite, even, unimodular lattices of dimension $8n$ given in \eqref{eqn:unimodularmass}.

Equation~\eqref{eqn:siegelweillattice} can be rewritten in a more physically suggestive form. First, recall that the vacuum character $\mathrm{ch}_{V_L}(\tau)$ of a lattice vertex operator algebra $V_L$ takes the form 
\begin{equation*}
    \mathrm{ch}_{V_L}(\tau) = \tr_{V_L}q^{L_0-c/24} = \frac{\Theta_L(\tau)}{\eta(\tau)^{\rk(L)}},
\end{equation*}
where $\eta(\tau) = q^{1/24}\prod_{n=1}^\infty (1-q^n)$ is the Dedekind eta function. Moreover, we may re-express the Eisenstein series via the elementary identity (see, e.g., equation~(2.11) of \cite{CD14})
\begin{equation*}
    \sum_{\gamma\in \Gamma_\infty \backslash \mathrm{SL}_2(\Z)} \left(\frac{\epsilon(\gamma)}{\eta(\gamma\tau)}\right)^{k} = \frac{1}{\eta(\tau)^{k}}\sum_{\gamma\in \Gamma_\infty \backslash \mathrm{SL}_2(\Z)}\frac{1}{(c\tau+d)^{\frac k2}}=\frac{E_{k/2}(\tau)}{\eta(\tau)^k}
\end{equation*}
where $\Gamma_\infty = \{\pm \left(\begin{smallmatrix} 1 & n \\ 0 & 1 \end{smallmatrix}\right) \mid n\in\Z\}$ and $\epsilon(\gamma)$ is the ``multiplier system'' for the Dedekind eta function, i.e.
\begin{equation*}
    \eta(\gamma\tau) = \epsilon(\gamma)(c\tau+d)^{1/2}\eta(\tau),\quad\gamma = \left(\begin{smallmatrix} a & b \\ c & d\end{smallmatrix}\right)\in\mathrm{SL}_2(\Z).
\end{equation*}
Putting these two facts together, the Siegel--Weil identity becomes
\begin{equation*}
    \frac{1}{M_n}\sum_{V_L\in \mathrm{hgen}(V_{E_8}^n)}\frac{\mathrm{ch}_{V_L}(\tau)}{|\Aut(L)|}=\sum_{\gamma\in\Gamma_\infty\backslash \mathrm{SL}_2(\Z)}\left(\frac{\epsilon(\gamma)}{\eta(\gamma\tau)}\right)^{8n}.
\end{equation*}

This equation admits a kind of holographic interpretation in the spirit of \cite{MW20,ACHT21}. Indeed, one may think of each summand on the right-hand side as the partition function of a 3d TQFT evaluated on an $\mathrm{SL}_2(\Z)$ black hole, with, e.g., $\gamma=\mathrm{id}$ corresponding to thermal AdS. The  TQFT is roughly an abelian Chern--Simons theory with gauge group $\mathbb{R}^r$ and trivial K-matrix (though see Section~4.4 of \cite{MW20} for why this identification is problematic); the fact that we sum over bulk handle-body geometries with fixed asymptotic boundary is motivated by the standard holographic dictionary. However, instead of being dual to a single 2d CFT, the left-hand side of the equation suggests that what we find instead is an ``ensemble average'' of free chiral 2d CFTs with fixed central charge, where each theory $V_L$ occurs in the ensemble with probability $(M_n|\Aut(L)|)^{-1}$. Thus, this Siegel--Weil identity reflects a kind of chiral variation of the duality involving non-chiral free theories considered in \cite{MW20,ACHT21}.

\medskip

The discussion so far has illustrated how hyperbolic genera of \emph{free} chiral CFTs are implicated in a toy model of holography. There is a natural generalization to arbitrary (i.e., not necessarily free) hyperbolic genera of holomorphic vertex operator algebras.

\begin{conj}[Siegel--Weil Identity]\label{conj:generalizedSiegelWeil}
Let $V$ be an arbitrary \strat{} holomorphic vertex operator algebra with central charge $c$, interacting sector (Heisenberg commutant) $C$ and free sector (associated lattice vertex operator algebra) having central charge $r\geq 3$. Then there is a generalized Siegel--Weil identity
\begin{equation*}
   \frac{1}{\mathrm{mass}(V)} \sum_{W\in\mathrm{hgen}(V)} \frac{\mathrm{ch}_W(\tau)}{|G_W|} =\sum_{\gamma\in \Gamma_\infty\backslash \mathrm{SL}_2(\Z)}\epsilon(\gamma)^c \frac{\mathrm{ch}_C(\gamma\tau)}{\eta(\gamma\tau)^r}.
\end{equation*}
\end{conj}
\begin{rem}
    The assumption that $r\geq 3$ is in order to ensure that the right-hand side of the equation converges. This is similar to the fact that the Narain ensemble average diverges for $(c_L,c_R)=(1,1)$ free boson theories \cite{MW20,ACHT21}.
\end{rem}

Our proposal is spiritually similar to (and can be thought of as a chiral variation of) the work of \cite{DHJ21}, where the authors show that averaging over moduli in non-chiral Wess--Zumino--Witten models can be achieved by computing a modular sum over the characters of parafermion theories.  Our formula is analogous, but with the parafermion characters replaced by those of the Heisenberg commutant~$C$. We may think of the right-hand side as arising by summing a particular TQFT over bulk geometries:  the TQFT is roughly $\mathbb{R}^r$ abelian Chern--Simons theory with trivial K-matrix coupled to ``topological matter'' represented by the modular tensor category $\Rep(C)$. The left-hand side is an ensemble average over (not-necessarily free) chiral CFTs that are hyperbolically equivalent to $W$.

\medskip
    
We shall verify \autoref{conj:generalizedSiegelWeil} in an example momentarily. But first, let us describe concretely how one may evaluate the right-hand side of this generalized Siegel--Weil identity, given an understanding of $\mathrm{ch}_C(\tau)$ and its behavior under modular transformations.

We assume for simplicity that $r$ is even. In this case, only even powers of the multiplier system $\epsilon$ appear in various formulae, and we can use that $\epsilon^p$ is a linear (as opposed to projective) representation of $\mathrm{SL}_2(\Z)$ when $p$ is an even integer. By the results of \cite{CG99,DLN15}, because $C$ is strongly rational, there exists an $N_1\in\Ns$ such that the representation with respect to which the character vector of $C$ transforms under modular transformations \cite{Zhu96} has kernel containing $\Gamma(N_1)$. In particular, $\mathrm{ch}_C(\tau)$ is then a weight-$0$ modular function for $\Gamma(N_1)$, i.e.\ 
\begin{equation*}
    \mathrm{ch}_C(\gamma\tau) = \mathrm{ch}_C(\tau)\quad \text{for all } \gamma\in\Gamma(N_1).
\end{equation*}
Similarly, there is an integer $N_2\in\Ns$ such that $\epsilon^{c-r}$ contains $\Gamma(N_2)$ in its kernel. Let $N\coloneqq\lcm(N_1,N_2)$, and set $\Gamma_\infty^+(N)\coloneqq\{\left(\begin{smallmatrix} 1 & Nn \\ 0 & 1 \end{smallmatrix}\right)\mid n\in\Z\}$. Then, letting $j(\gamma,\tau) \coloneqq (\tilde{c}\tau+\tilde{d})$ for $\gamma=\left(\begin{smallmatrix}\ast  & \ast \\ \tilde{c} & \tilde{d}\end{smallmatrix}\right)$,
we may compute 
\begin{align}\label{eqn:modularsum}
        \sum_{\gamma\in\Gamma_\infty\backslash\mathrm{SL}_2(\Z)} \epsilon(\gamma)^c \frac{\mathrm{ch}_C(\gamma\tau)}{\eta(\gamma\tau)^r}&= \frac{1}{2N}\sum_{\gamma\in\Gamma_\infty^+(N)\backslash\mathrm{SL}_2(\Z)} \!\!\epsilon(\gamma)^c \frac{\mathrm{ch}_C(\gamma\tau)}{\eta(\gamma\tau)^r} \nonumber\\
    &=\frac{1}{2N}\sum_{\lambda'\in \Gamma(N)\backslash\mathrm{SL}_2(\Z)}\sum_{\lambda \in \Gamma_\infty^+(N)\backslash\Gamma(N)}\!\!\epsilon(\lambda\lambda')^c \frac{\mathrm{ch}_C(\lambda\lambda'\tau)}{\eta(\lambda\lambda'\tau)^r}\nonumber\\
    &=\frac{1}{2N}\sum_{\lambda'}\frac{\epsilon(\lambda')^c\mathrm{ch}_C(\lambda'\tau)}{\eta(\lambda'\tau)^r}  \sum_{\lambda }\frac{1}{j(\lambda,\lambda'\tau)^{\frac r2}}\\ 
    &= \frac{1}{2N\eta(\tau)^r}\sum_{\lambda'\in \Gamma(N)\backslash \mathrm{SL}_2(\Z)}\!\!\!\!\epsilon(\lambda')^{c-r}\mathrm{ch}_C(\lambda' \tau) \tilde{E}_{N,r/2}^{(0,1)\lambda'}(\tau),\nonumber
\end{align}
where in the last line, we have expressed the sum over $\lambda$ in terms of (normalized) Eisenstein series for congruence subgroups using identities which can be found in \autoref{app:eisenstein}. Thus, the right-hand side of the Siegel--Weil identity may be reduced from an infinite sum to a finite sum involving Eisenstein series, whose $q$-expansions are known.

\medskip

We now verify \autoref{conj:generalizedSiegelWeil} for the hyperbolic genus B.
\begin{table}[ht]
 \caption{The $17$ \strat{}, holomorphic vertex operator algebras $V$ of central charge $c=24$ in the hyperbolic genus B (with associated lattice $L$).}\label{tab:genusB}
    \begin{tabular}{r|l|r|c}
    No. & $V_1$ & $\dim(V_1)$ & $\Aut(L)$ \cite{BLS23} \\ \hline\hline
    $5$ & $\mathsf{A}_{1,2}^{16}$ & $48$ & $W(A_1)\wr \mathrm{AGL}_4(2)$  \\
    $16$ & $\mathsf{A}_{3,2}^4\mathsf{A}_{1,1}^4$ & $72$ & $(W(A_3)^4\times W(A_1)^4).W(D_4)$\\
    $25$ & $\mathsf{D}_{4,2}^2 \mathsf{C}_{2,1}^4$ & $96$ & $(W(D_4)^2\times W(C_2)^4).(2\times \mathfrak{S}_4)$\\ 
    $26$ &$\mathsf{A}_{5,2}^2\mathsf{C}_{2,1}\mathsf{A}_{2,1}^2$ & $96$ & $(W(A_5)^2\times W(A_2)^2\times W(C_2)).\mathrm{Dih}_8$\\
    $31$ & $\mathsf{D}_{5,2}^2\mathsf{A}_{3,1}^2$ &$120$ & $(W(A_7)\times W(A_3)\times W(C_3)^2).\Z_2^2$\\
    $33$ & $\mathsf{A}_{7,2}\mathsf{C}_{3,1}^2\mathsf{A}_{3,1}$ & $120$ & $(W(D_5)^2\times W(A_3)^2).\mathrm{Dih}_8$\\
    $38$ & $\mathsf{C}_{4,1}^4$ & $144$ & $W(C_4)\wr \mathfrak{S}_4$\\
    $39$ & $\mathsf{D}_{6,2}\mathsf{C}_{4,1}\mathsf{B}_{3,1}^2$& $144$ & $(W(D_6)\times W(C_4)\times W(B_3)^2).\Z_2$\\
    $40$ &$\mathsf{A}_{9,2}\mathsf{A}_{4,1}\mathsf{B}_{3,1}$& $144$ & $(W(A_9)\times W(A_4)\times W(B_3)).\Z_2$\\
    $44$ & $\mathsf{E}_{6,2}\mathsf{C}_{5,1}\mathsf{A}_{5,1}$ & $168$ &$(W(E_6)\times W(A_5)\times W(C_5)).\Z_2$ \\
    $47$ & $\mathsf{D}_{8,2}\mathsf{B}_{4,1}^2$ & $192$ & $W(D_8)\times W(B_4)\wr \Z_2)$\\
    $48$ & $\mathsf{C}_{6,1}^2\mathsf{B}_{4,1}$ & $192$ &$W(C_6)\wr 2 \times W(B_4)$ \\
    $50$ & $\mathsf{D}_{9,2}\mathsf{A}_{7,1}$ & $216$ & $(W(D_9)\times W(A_7)).\Z_2$\\
    $52$ & $\mathsf{C}_{8,1}\mathsf{F}_{4,1}^2$ &$240$ & $W(C_8)\times W(F_4)\wr\Z_2$\\
    $53$ & $\mathsf{E}_{7,2}\mathsf{B}_{5,1}\mathsf{F}_{4,1}$&$240$ &$W(E_7)\times W(B_5)\times W(F_4)$\\
    $56$ & $\mathsf{C}_{10,1}\mathsf{B}_{6,1}$ & $288$ &$W(C_{10})\times W(B_6)$\\
    $62$ & $\mathsf{B}_{8,1}\mathsf{E}_{8,2}$ & $384$ &$W(B_8)\times W(E_8)$
    \end{tabular}
\end{table}
\begin{ex}[Siegel--Weil Identity for Genus B]
We consider the hyperbolic genus~B. This is the unique hyperbolic genus of $c=24$, \strat{}, holomorphic vertex operator algebras whose associated lattices have rank $r=16$. Some basic information about the 17 vertex operator algebras which reside in this genus is given in \autoref{tab:genusB}.

Consider an arbitrary vertex operator algebra $V$ in this hyperbolic genus with associated lattice $L$. We adopt the notation used in \autoref{subsec:mass}. To evaluate the left-hand side of the Siegel--Weil identity, we note that, by the last column in \autoref{tab:Schellekensgenera}, the group-theoretic factor for this hyperbolic genus is $[\Aut(\tilde{L}):G_{\tilde{V}}]=1$. By \autoref{cor:massconstant}, this implies that $\overline{\Aut}(C)$ is equal to all of $\Aut(\Rep(C))$ (that is, every ribbon auto-equivalence of $\Rep(C)$ can be induced by an automorphism of~$C$) and hence that $\tau^\ast \overline{\Aut}(C)$ is all of $\Aut(\mathcal{C}(L'/L))$, where $\tau^\ast$ was defined before \autoref{prop:GVcharacterization}. In particular, it follows that $\overline{\Aut}(L)\cap\tau^\ast \overline{\Aut}(C) = \overline{\Aut}(L)$, and by applying \autoref{cor:massconstant} again, that $G_V=\Aut(L)$. Thus, the groups $G_V$ for this genus all reduce to the automorphism groups of the associated lattices, which are reported in \autoref{tab:genusB}. 

Using the fact that $\mathrm{ch}_V(\tau) = J(\tau) + \dim(V_1)$ for every \strat{}, holomorphic vertex operator algebra of central charge $c=24$, it follows that  
\begin{equation}\label{eqn:lhssiegelweilgenusB}
\begin{split}
    \frac{1}{\mathrm{mass}(\mathrm{B})}\sum_{V\in\mathrm{hgen}(\mathrm{B})} \frac{\mathrm{ch}_V(\tau)}{|G_V|} &= J(\tau)+\frac{1}{\mathrm{mass}(\mathrm{B})}\sum_{V\in\mathrm{hgen}(\mathrm{B})}\frac{\dim(V_1)}{|\Aut(L)|} \\
    &= J(\tau) +\frac{1488}{17},
\end{split}
\end{equation}
where $J(\tau)=j(\tau)-744=q^{-1}+196884q+\dots$ On the other hand, the right-hand side of the generalized Siegel--Weil identity can be evaluated by taking $N=6$ in equation~\eqref{eqn:modularsum}, 
\begin{equation}\label{eqn:intermediatesiegelweil}
    \sum_{\gamma\in\Gamma_\infty\backslash\mathrm{SL}_2(\Z)} \frac{\mathrm{ch}_C(\gamma\tau)}{\eta(\gamma\tau)^{16}} = \frac{1}{12\eta(\tau)^{16}} \sum_{\lambda'\in\Gamma(6)\backslash \mathrm{SL}_2(\Z)} \frac{\mathrm{ch}_C(\lambda'\tau)\tilde{E}_{6,8}^{(0,1)\lambda'}(\tau)}{\epsilon(\lambda')^{16}},
\end{equation}
using the fact that the vacuum character of the Heisenberg commutant $C$ for genus~B is \cite{HS03}
\begin{equation*}
    \mathrm{ch}_C(\tau) = (h(\tau)+g_0(\tau))\eta(\tau)^{16}
\end{equation*}
where
\begin{align*}
    h(\tau) &= \eta(\tau)^{-8}\eta(2\tau)^{-8} \\
    g_0(\tau) &= \frac{1}{2}\bigl(h(\tau/2)+h((\tau+1)/2)\bigr).
\end{align*}
A straightforward, albeit tedious, computer calculation then shows that \eqref{eqn:intermediatesiegelweil} reduces again to $J(\tau)+\frac{1488}{17}$, in agreement with \eqref{eqn:lhssiegelweilgenusB}, consistent with \autoref{conj:generalizedSiegelWeil}.

\medskip

Genus~B corresponds to the symbol $(C,24)_{\Vect}$. As in the previous section, we may generalize genus B and obtain an infinite family of hyperbolic genera $\mathrm{B}_m$   corresponding to the symbols $(C,24+8m)_{\Vect}$ for $m\in\N$. \autoref{conj:generalizedSiegelWeil} then gives a prediction for the ``average'' partition function of a vertex operator algebra in these genera. For example, in the case of $m=1$, we conjecturally find 
\begin{align*}
    \frac{1}{\mathrm{mass}(\mathrm{B}_1)}\sum_{V\in\mathrm{hgen}(\mathrm{B}_1)}\frac{\mathrm{ch}_V(\tau)}{|G_V|} &\stackrel{?}{=} \sum_{\gamma\in\Gamma_\infty\backslash\mathrm{SL}_2(\Z)}\epsilon(\gamma)^8 \frac{\mathrm{ch}_C(\gamma\tau)}{\eta(\gamma\tau)^{24}} \\
    &=\frac{1}{12\eta(\tau)^{24}}\sum_{\lambda'\in\Gamma(6)\backslash\mathrm{SL}_2(\Z)} \frac{\mathrm{ch}_C(\lambda'\tau)\tilde{E}_{6,12}^{(0,1)\lambda'}(\tau)}{\epsilon(\lambda')^{16}}\\
    &= \left(j(\tau)-\frac{666360}{691}\right)j(\tau)^{1/3} \\
    &= q^{-4/3}+\frac{19112}{691}q^{-1/3}+\dots
\end{align*}
This shows in particular that the ``average'' $c=32$, \strat{}, holomorphic vertex operator algebra in genus $\mathrm{B}_1$ has a continuous global symmetry group of dimension $\langle \dim(V_1)\rangle =\frac{19112}{691}\approx 27.7$, which is compatible with the fact that all vertex operator algebras in genus $\mathrm{B}_1$ must have $\dim(V_1)\geq 24$. 
\end{ex}


\section{Future Directions}

There are many future directions which flow from our work.
\begin{enumerate}[wide]
\item A classification of $c=32$, \strat{}, holomorphic vertex operator algebras has long been thought to be impossible due to their vast number. However, it is conceivable that one may be able to achieve at least a partial classification of such vertex operator algebras if one is content with working up to hyperbolic equivalence. We have taken a few small steps in this direction in \autoref{subsec:c=32}.
\item Hyperbolic genera of strongly rational vertex operator algebras are provably finite by the mass formula of \autoref{thm:massformula}, coupled with \autoref{cor:massconstant}. One might fantasize that there is a ``quantum'' mass formula which counts vertex operator algebras in a bulk genus, and which could be used to prove the finiteness of bulk genera, \autoref{conj:bulkfiniteness}. This would presumably lead to a proof of the uniqueness of the moonshine module, \autoref{conj:moonshineuniqueness}.
\item There is third definition of what it means for two lattices to belong to the same genus, part~\ref{item:latgen1} of \autoref{defi:latticegenus}, which we have not explored in this work. It would be interesting to formulate a generalization of this definition to the setting of vertex operator algebras, as we have done with parts \ref{item:latgen2} and \ref{item:latgen3} of \autoref{defi:latticegenus}.
\item It is unclear to what extent the unitarity or pseudo-unitarity of a modular tensor category $\mathcal{C}$ forces the unitary or positivity of strongly rational vertex operator algebras $V$ with $\Rep(V)\cong\mathcal{C}$. Understanding this would certainly lead to a clearer picture of what bulk genera look like. 
\item Proving the Siegel--Weil identity of \autoref{conj:generalizedSiegelWeil} and fleshing out its holographic interpretation should be within reach of present methods.
\end{enumerate}


\appendix


\section{Hyperbolic and Bulk Genus}\label{sec:altproof}

For the reader's convenience, we sketch an alternative proof of \autoref{cor:hypbulk}, without using \autoref{thm:hypcomm}. We caution the reader that, as we ended up not using it, we did not work out the following ``proof'' in all detail.

\medskip

We recall that \autoref{cor:hypbulk} asserts that two \strat{} \voa{}s $V$ and $V'$ in the same hyperbolic genus are also in the same bulk genus.

It follows directly from $V\otimes V_{\II_{1,1}}\cong V'\otimes V_{\II_{1,1}}$ that $V$ and $V'$ have the same central charge. Moreover, it was already shown in \cite{Mor21} that $\Rep(V)$ and $\Rep(V')$ are equivalent as plain categories. What remains to be proved is that the equivalence given in \cite{Mor21} is actually a ribbon equivalence, i.e.\ that $\Rep(V)$ and $\Rep(V')$ are equivalent as modular tensor categories (and in particular as braided tensor categories). Moreover, the partial result from \cite{Mor21} is sufficient to prove the conjecture in the case where $\Rep(V)\cong\Vect$. Then $\Rep(V')\cong\Rep(V)$ as a modular tensor category because there is only one modular tensor category up to ribbon equivalence with just one simple object.

Before we prove the proposition, we point out the difficulty, namely that many results on representation categories like the HLZ-tensor product theory \cite{HLZ} are only stated for \voa{}s and not for (conformal) vertex algebras. While it follows from the results in \cite{Don93,DLM97,DL93} that the representation category of a lattice vertex algebra $V_L$ for an even, possibly indefinite lattice $L$ is the pointed modular tensor category $\mathcal{C}(L'/L)$ associated with the discriminant form $L'/L$ (note that it is in particular shown that all weak $V_L$-modules are direct sums of the finitely many irreducible modules, which are indexed by the cosets in $L'/L$), we cannot immediately say what the representation category of $V\otimes V_L$ is for a \strat{} \voa{} $V$.

We are grateful to Robert McRae for providing us with a crucial idea for the following proof, namely viewing the conformal vertex algebra $V\otimes V_{\II_{1,1}}$ as an infinite simple-current extension of a \vosa{}, and for pointing us to the relevant literature.
\begin{proof}[Sketch of alternative proof of \autoref{cor:hypbulk}]
Because $V_{\II_{1,1}}$ is not a \voa{}, we cannot directly apply the HLZ-tensor product theory \cite{HLZ} to it. To remedy this, we consider the (standard) Heisenberg \vosa{} $H\subset V_{\II_{1,1}}$, which is associated with the (standard) Cartan subalgebra $\hh$ of $V_{\II_{1,1}}$. The latter we view as abelian Lie algebra equipped with the non-degenerate bilinear form $\langle\cdot,\cdot\rangle$ inherited from the lattice $\II_{1,1}$. Now, $H$ is a \voa{}, but it is neither rational nor $C_2$-cofinite, i.e.\ it does not satisfy the usual regularity assumptions.

Following, e.g., \cite{CY21,CMY22,McR23} we consider the category $\Rep_{C_1}(H)$ of $C_1$-cofinite grading-restricted generalized $H$-modules (or equivalently the $C_1$-cofinite admissible
$H$-modules). We then apply Theorem~3.3.5 in \cite{CY21} to $H$, which states that the category $\Rep_{C_1}(H)$ is the same as the category of finite-length generalized $H$-modules and, in particular, admits the vertex algebraic braided tensor category structure of \cite{HLZ}, with a ribbon twist.

We then consider the tensor product \voa{} $V\otimes H$. Recall that $V$ is assumed to be \strat{}. Hence, the category of $C_1$-cofinite grading-restricted generalized $V$-modules is simply the usual representation category $\Rep(V)$, which is a (semisimple) modular tensor category by \cite{Hua08b}, and in particular has the braided tensor structure of \cite{HLZ}.

Now, by \cite{McR23}, the category $\Rep_{C_1}(V\otimes H)$ of $C_1$-cofinite grading-restricted generalized $V\otimes H$-modules admits the braided tensor category structure of \cite{HLZ} and is braided tensor equivalent to the Deligne product $\Rep(V)\boxtimes\Rep_{C_1}(H)$.

The lattice vertex algebra $V_{\II_{1,1}}$ is an infinite simple-current extension of the Heisenberg \vosa{} $H$. That is, we can view $V_{\II_{1,1}}$ as a commutative algebra object \cite{HKL15} (see Theorem~7.5 in \cite{CMY22} for the precise formulation, easily modified to account for the fact that $V_{\II_{1,1}}$ fails to satisfy the $L_0$-grading conditions that would make it a \voa{}) in the ind-completion or direct-limit completion $\Ind(\Rep_{C_1}(H))$ \cite{CMY22} (cf.\ \cite{AR18}), which is naturally endowed with a vertex algebraic braided tensor structure extending the one on $\Rep_{C_1}(H)$ by the main result of \cite{CMY22}.
In the same sense, $V\otimes V_{\II_{1,1}}$ is a commutative algebra object in the ind-completion $\Ind(\Rep_{C_1}(V\otimes H))$ of $\Rep_{C_1}(V\otimes H)$.

This gives us an induction functor $F\colon \Ind(\Rep_{C_1}(V\otimes H)) \to \mathcal{C}_{V\otimes V_{\II_{1,1}}}$, where objects of $\mathcal{C}_{V\otimes V_{\II_{1,1}}}$ are possibly ``non-local'' $V\otimes V_{\II_{1,1}}$-modules that are objects of $\Ind(\Rep_{C_1}(V\otimes H))$ as $H$-modules (e.g., twisted $V\otimes V_{\II_{1,1}}$-modules associated with automorphisms of $V_{\II_{1,1}}$ that fix $H$).
The subcategory $\Rep({V\otimes V_{\II_{1,1}}})\subset\mathcal{C}_{V\otimes V_{\II_{1,1}}}$ appearing in Theorem~7.7 of \cite{CMY22} (there decorated with a superscript $0$) denotes only the usual ``local'' $V\otimes V_{\II_{1,1}}$-modules that are objects of $\Ind(\Rep_{C_1}(V\otimes H))$ as $H$-modules, and it is also endowed with a braided tensor structure in the sense of \cite{HLZ}. 

It is not difficult to check that for a $V$-module $W$ and an irreducible (Fock) $H$-module $H_\lambda$ with $\lambda\in\hh$, $F(W\otimes H_\lambda)$ is an untwisted $V\otimes V_{\II_{1,1}}$-module in the usual sense if and only if $\lambda$ is in the dual lattice $\II_{1,1}'$, in which case $F(W\otimes H_\lambda)$ is isomorphic to $W\otimes V_{\lambda+\II_{1,1}}$.
Here, $\II_{1,1}'=\II_{1,1}$ because $\II_{1,1}$ is unimodular, and so $F(W\otimes H_\lambda)\cong W\otimes V_{\II_{1,1}}$.

We then consider the braided tensor functor
\begin{equation*}
\Rep(V)\to\Rep_{C_1}(V\otimes H)\to\Ind(\Rep_{C_1}(V\otimes H)) \stackrel{F}{\to} \Rep(V\otimes V_{\II_{1,1}})
\end{equation*}
that maps a $V$-module $W$ to
\begin{equation*}
W\mapsto W\otimes H\mapsto W\otimes H\mapsto F(W\otimes H)=W\otimes V_{\II_{1,1}}.
\end{equation*}
The first functor is braided tensor by the definition of the Deligne product (see \cite{McR23}), the second is by Theorem~1.2 in \cite{CMY22}. Finally, the induction functor~$F$ restricted to modules of the form $W\otimes H$ is a braided tensor functor because it is a tensor functor and because the image is contained in the subcategory of untwisted $V\otimes V_{\II_{1,1}}$-modules.

Moreover, one can show that the above functor is fully faithful (cf.\ \cite{CKM22}).
Hence, we can identify $\Rep(V)$ with a braided tensor subcategory of $V\otimes V_{\II_{1,1}}$-modules. Since $V_{\II_{1,1}}$ is holomorphic (the modules of lattice vertex algebras were determined in \cite{Don93,DLM97} regardless of whether the even lattice is positive-definite or not), this subcategory is the braided tensor subcategory of all $V\otimes V_{\II_{1,1}}$-modules that are finite direct sums of irreducible modules.

Overall, we have shown that the category of all $V\otimes V_{\II_{1,1}}$-modules that are finite direct sums of irreducible modules is braided tensor equivalent to $\Rep(V)$. The same is true for $\Rep(V')$ with $V\otimes V_{\II_{1,1}}\cong V'\otimes V_{\II_{1,1}}$. This establishes the assertion that $\Rep(V)\cong\Rep(V')$ are braided tensor equivalent. It is not difficult to see that this is, in fact, a ribbon equivalence.
\end{proof}


\section{Eisenstein Series for Congruence Subgroups}\label{app:eisenstein}

We telegraphically record some formulae involving Eisenstein series for congruence subgroups $\Gamma(N)$, following very closely Section 5 of \cite{Sch06}. 

Let $N\geq1$ and $k\geq 3$ be integers. Then for $c,d\in\Z$, define the Eisenstein series
\begin{equation*}
    E_{N,k}^{(c,d)}(\tau) = \sum_{\substack{(m,n)\in\Z^2\setminus(0,0) \\ (m,n)\equiv (c,d)\bmod{N}}} \frac{1}{(m\tau+n)^k}.
\end{equation*}
Each such function is a holomorphic modular form for $\Gamma(N)$ of weight $k$, and moreover satisfies, for each $\lambda = \bigl(\begin{smallmatrix} \alpha & \beta \\ \gamma & \delta\end{smallmatrix}\bigr)  \in \mathrm{SL}_2(\Z)$, 
\begin{equation*}
    E_{N,k}^{(c,d)}(\lambda\tau)=(\gamma \tau + \delta)^k E_{N,k}^{(c,d)\lambda }(\tau).
\end{equation*}
The $q$-expansions of these functions are given by 
\begin{align*}
    E_{N,k}^{(c,d)}(\tau) &= B_{N,k}^{(c,d)} + C_{N,k}\sum_{\substack{m\equiv c \bmod{N}\\ m \geq 1}} \sum_{n\geq 1} n^{k-1}e^{2\pi\i n d/N}q^{nm/N}  \\
     & \quad + (-1)^k C_{N,k} \sum_{\substack{m\equiv -c\bmod{N}\\ m \geq 1}}\sum_{n\geq 1} n^{k-1} e^{-2\pi\i n d/N}q^{nm/N},
\end{align*}
where 
\begin{align*}
    B_{N,k}^{(c,d)} &= \begin{cases} \zeta^d(k) + (-1)^k\zeta^{-d}(k), & \text{if } c\equiv 0\bmod{N},\\
    0, & \text{otherwise},
\end{cases} \\
    C_{N,k} &= \frac{(-1)^k (2\pi\i)^k}{N^k (k-1)!}.
\end{align*}
The zeta function appearing in the definition of $B$ is given by 
\begin{equation*}
    \zeta^d(k) = \sum_{\substack{n\equiv d\bmod{N}\\ n\geq 1 }}\frac{1}{n^k}.
\end{equation*}
An alternative normalization will be useful for our purposes. In particular, note that
\begin{align*}
    E_{N,k}^{(0,1)}(\tau) &= \sum_{\substack{(m,n)\in\Z^2\setminus(0,0) \\ (m,n)\equiv (0,1)\bmod{N}}} \frac{1}{(m\tau+n)^k} \\
    &= \sum_{\substack{(m,n)\in\Z^2\setminus(0,0) \\ (m,n)\equiv (0,1)\bmod{N}\\
    (m,n)=1}}\sum_{\substack{t \neq 0 \\ t\equiv 1 \bmod{N}}}\frac{1}{(mt\tau+nt)^k} \\
    &= B_{N,k}^{(0,1)}\sum_{\left(\begin{smallmatrix} \ast & \ast \\ m & n \end{smallmatrix}\right) \in \Gamma^+_\infty(N)\backslash\Gamma(N)}\frac{1}{(m\tau+n)^k} \eqqcolon B_{N,k}^{(0,1)} \tilde{E}_{N,k}^{(0,1)}(\tau).
\end{align*}
More generally, we define
\begin{equation*}
    \tilde{E}_{N,k}^{(c,d)}(\tau) = \frac{1}{B_{N,k}^{(0,1)}}E_{N,k}^{(c,d)}(\tau).
\end{equation*}


\bibliographystyle{alpha_noseriescomma}
\bibliography{references}{}

\end{document}